\documentclass[11pt]{llncs}
\pdfoutput=1

\usepackage{graphicx}
\usepackage{times}
\usepackage{mathptm}
\usepackage{cite}
\usepackage{hyperref}

\DeclareSymbolFont{AMSb}{U}{msb}{m}{n}
\DeclareSymbolFontAlphabet{\Bbb}{AMSb}

\newcommand{\G}{\mathcal{G}}
\newcommand{\C}{\mathcal{C}}
\newcommand{\IN}{_\mathrm{in}}
\newcommand{\OUT}{_\mathrm{out}}
\newcommand{\Po}{\mathcal{P}}

\makeatletter
\def\hb@xt@{\hbox to }
\makeatother

\let\oldendproof\endproof
\def\endproof{\qed\oldendproof}

 \graphicspath{{figures/}}
\begin{document}
\title{Steinitz Theorems for Orthogonal Polyhedra}

\author{David Eppstein\inst{1}  and Elena Mumford\inst{2}}

\institute{Computer Science Department\\
University of California, Irvine\\
\email{eppstein@uci.edu}\\[0.1in]
\and
Department of Mathematics and Computer Science\\
Technische Universiteit Eindhoven\\
\email{e.mumford@tue.nl}}

\maketitle
\setcounter{page}{0}

\begin{abstract}
We define a \emph{simple orthogonal polyhedron} to be a three-dimensional polyhedron with the topology of a sphere in which three mutually-perpendicular edges meet at each vertex.
By analogy to Steinitz's theorem characterizing the graphs of convex polyhedra, we find graph-theoretic characterizations of three classes of simple orthogonal polyhedra: \emph{corner polyhedra}, which can be drawn by isometric projection in the plane with only one hidden vertex, \emph{$xyz$ polyhedra}, in which each axis-parallel line through a vertex contains exactly one other vertex, and arbitrary simple orthogonal polyhedra. In particular, the graphs of $xyz$ polyhedra are exactly the bipartite cubic polyhedral graphs, and every bipartite cubic polyhedral graph with a 4-connected dual graph is the graph of a corner polyhedron. Based on our characterizations we find efficient algorithms for constructing orthogonal polyhedra from their graphs.
\end{abstract}

\newpage

\section{Introduction}

Steinitz's theorem~\cite{Gru-03,Ste-EMW-22,Zie-95} characterizes the skeletons of three-dimensional convex polyhedra in purely graph-theoretic terms: they are exactly the 3-vertex-connected planar graphs. In one direction, this is straightforward to prove: every convex polyhedron has a skeleton that is 3-connected and planar. The main content of Steinitz's theorem lies in the other direction, the statement that every 3-connected planar graph can be represented as a polyhedron. Steinitz's theorem, together with Balinski's theorem that every $d$-dimensional polytope has a $d$-connected skeleton~\cite{Bal-PJM-61}, form the foundation stones of polyhedral combinatorics; Grunbaum writes~\cite{Gru-03} that Steinitz's theorem is ``the most important and deepest known result on 3-polytopes.''

However, analogous results characterizing the skeletons of other classes of polyhedra or higher dimensional polytopes have been elusive. As Ziegler~\cite{Zie-95} writes, ``No similar theorem is known, and it seems that no similarly effective theorem is possible, in higher dimensions.'' Even in three dimensions, it remains unknown whether the complete graph $K_{12}$ may be embedded as a genus-six triangulated polyhedral surface, generalizing the toroidal embedding of $K_7$ as the Cs\'asz\'ar polyhedron~\cite{Csa-ASMS-49}.

\begin{figure}[t]
\centering\includegraphics[scale=0.45]{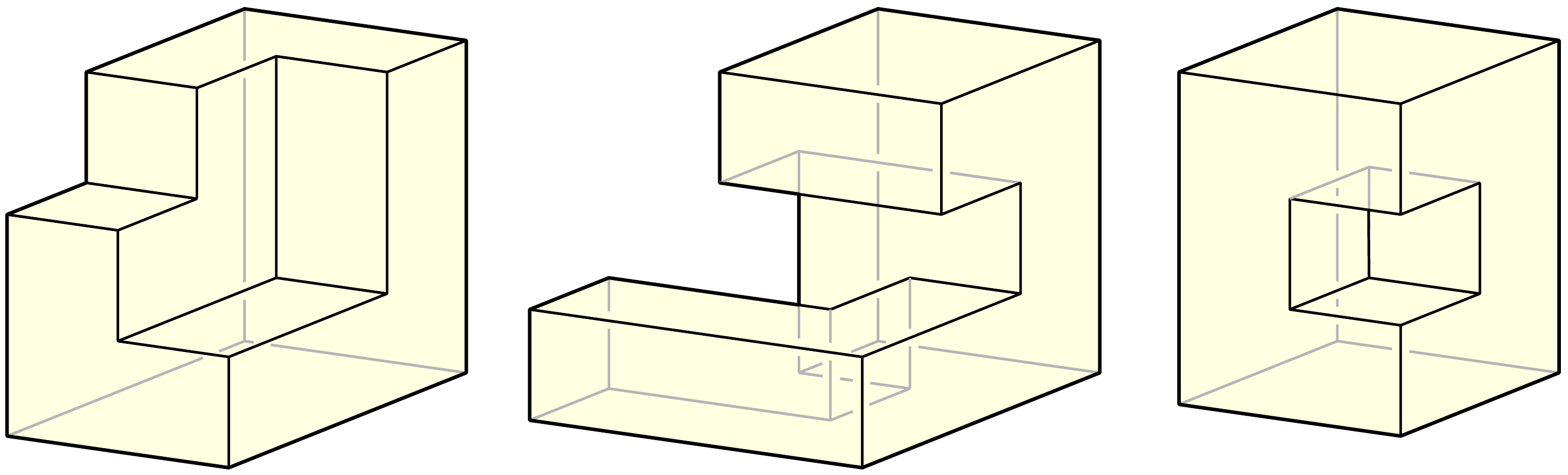}
\caption{Three types of simple orthogonal polyhedron: Left, a corner polyhedron. Center, an $xyz$ polyhedron that is not a corner polyhedron. Right, a simple orthogonal polyhedron that is not an $xyz$ polyhedron.}
\label{fig:orthotypes}
\end{figure}

In this paper, we characterize another class of three-dimensional non-convex polyhedra, which we call \emph{simple orthogonal polyhedra} (Fig.~\ref{fig:orthotypes}): polyhedra with the topology of a sphere, with simply-connected faces, and with exactly three mutually-perpendicular axis-parallel edges meeting at every vertex. We also consider two special cases of simple orthogonal polyhedra, which we call \emph{corner polyhedra} and \emph{$xyz$ polyhedra}. A corner polyhedron (Fig.~\ref{fig:orthotypes}, left) is a simple orthogonal polyhedron in which all but three faces are oriented towards the vector $(1,1,1)$; it can be drawn in the plane by isometric projection with only one of its vertices hidden (the one incident to the three back faces). An $xyz$ polyhedron (Fig.~\ref{fig:orthotypes}, center)  is a simple orthogonal polyhedron in which each axis-parallel line contains at most two vertices. We show:

\begin{itemize}
\item The graphs of corner polyhedra are exactly the cubic bipartite polyhedral graphs such that every separating triangle of the planar dual graph has the same parity. Here cubic means 3-regular, polyhedral means planar 3-connected, and we define the parity of a separating triangle later. The graphs with no separating triangles form the building blocks for all our other characterizations: every cubic bipartite polyhedral graph with a 4-connected planar dual is the graph of a corner polyhedron.
\item The graphs of $xyz$ polyhedra are exactly the cubic bipartite polyhedral graphs.
\item The graphs of simple orthogonal polyhedra are exactly the cubic bipartite planar graphs such that the removal of any two vertices leaves at most two connected components.
\end{itemize}
Based on our graph-theoretic characterizations of these classes of polyhedron, we find efficient algorithms for finding a polyhedral realization of the graph of any corner polyhedron, $xyz$ polyhedron, or simple orthogonal polyhedron.
Beyond the obvious applications of our results in graph drawing and architectural design, we believe that these results may have applications in image understanding, where an analysis of the structure of polyhedral and rectilinear objects has been an important subtopic~\cite{Huf-MI6-71,Wal-PCB-75,KirPap-JCSS-88}.

\begin{figure}[t]
\centering\includegraphics[scale=0.45]{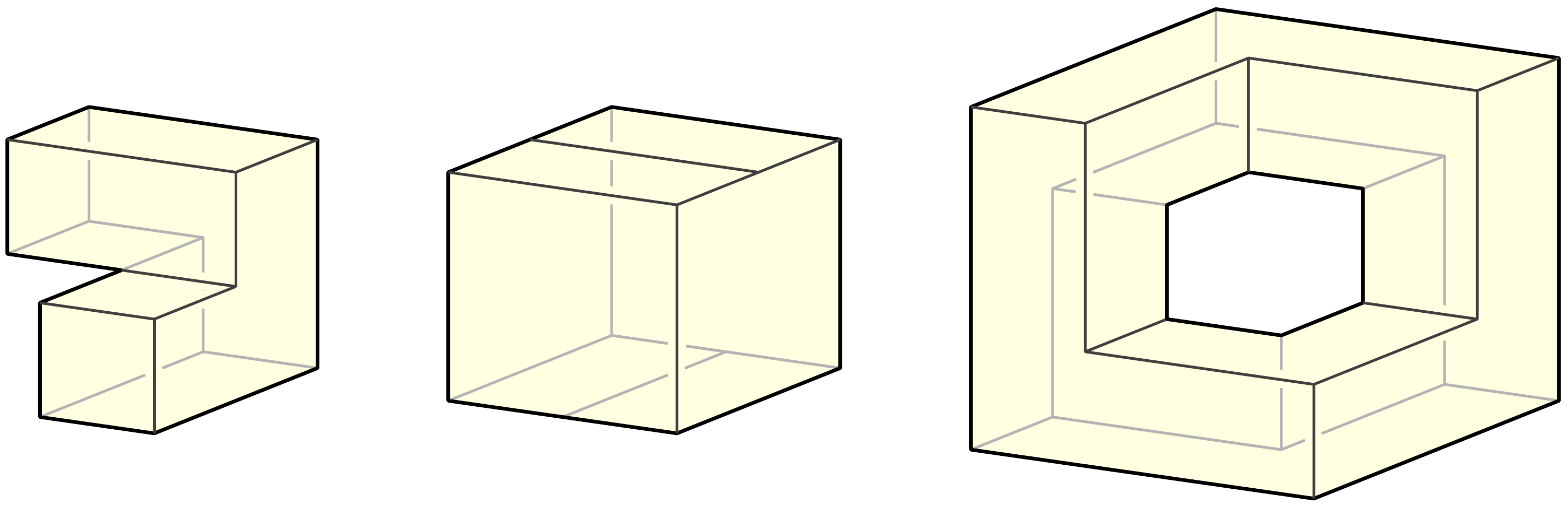}
\caption{Three orthogonal polyhedra that are not simple: Left, more than three edges meet at a vertex. Center, the bidiakis cube, with edges and faces meeting non-perpendicularly. Right, an orthogonally convex orthogonal polyhedron that does not have the topology of a sphere.}
\label{fig:nonsimple}
\end{figure}

Due to space considerations we defer the proofs of our results to appendices, and provide only a sketch of the main ideas of these proofs in the text of this paper.

\section{Related work}

Besides convex polyhedra and our results on orthogonal polyhedra,
two other classes of polyhedra have known graph-theoretic characterizations. They are the \emph{inscribable polyhedra} (convex polyhedra with all vertices on a common sphere or, almost equivalently, graphs of Delaunay triangulations)~\cite{DilSmi-DM-96,HodRivSmi-BAMS-92,Riv-AM-96} and a class of nonconvex polyhedra with star-shaped faces all but one of which are visible from a common viewpoint~\cite{HonNag-TR-08}.

The most direct predecessor of the work described here is our previous paper on three-dimensional bendless orthogonal graph drawing~\cite{Epp-GD-08}. We defined an \emph{$xyz$ graph} to be a cubic graph with axis-parallel edges such that the line through each edge does not pass through any other vertex. These graphs may also be defined in a coordinate-free way from their points, as there can be only one way of rotating a point set to form a connected $xyz$ graph~\cite{LoeMum-GD-08}. From every $xyz$ graph one may define an abstract topological surface by forming a face for every coplanar cycle; these faces may be 3-colored by the orientations of their defining planes. Conversely, for every 3-face-colored cubic topological cell complex on a manifold, assigning arbitrary distinct numbers to the faces and using these numbers as the Cartesian coordinates of the incident vertices leads to a representation as an $xyz$ graph.  As we proved, the planar $xyz$ graphs are exactly the bipartite cubic 3-connected planar graphs. Unlike polyhedra, an $xyz$ graph may have crossing points where pairs of edges or even triples of edges intersect (Fig.~\ref{fig:invertohedron}, left), and its face cycles may be linked in three-dimensional space. However, in some cases, an $xyz$ graph may be drawn as an orthogonal polyhedron, eliminating all edge crossings; for instance, we found an orthogonal polyhedron representation of the truncated octahedron (Fig.~\ref{fig:invertohedron}, right), and based on this example we posed as an open problem the algorithmic question of determining which $xyz$ graphs have a crossing-free representation. In this paper we answer that question in the planar case: all of them do, and more strongly all planar $xyz$ graphs have not just a crossing-free but a polyhedral representation.

\begin{figure}[t]
\centering
\includegraphics[height=1.75in]{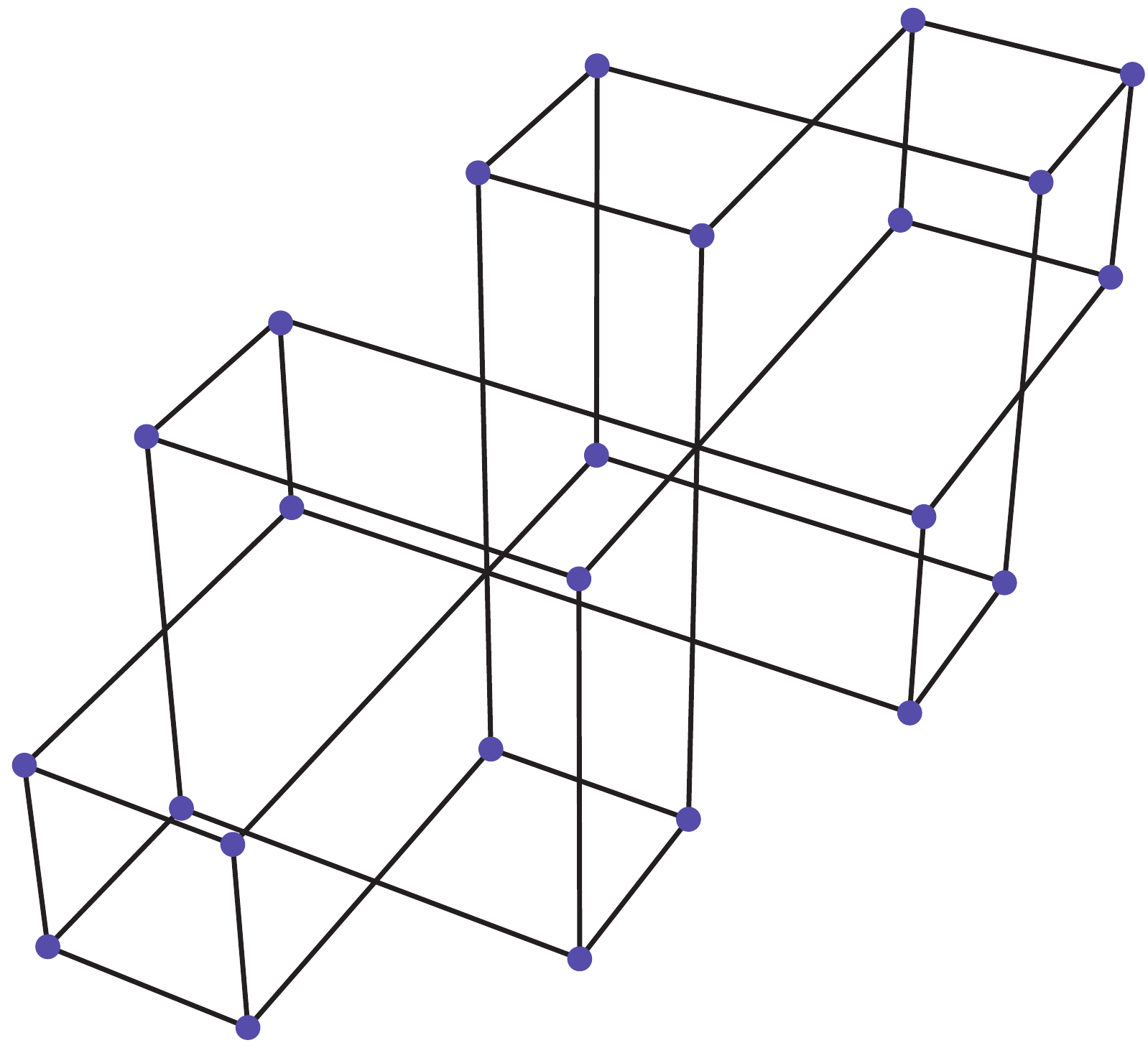}
\qquad\qquad
\includegraphics[height=1.75in]{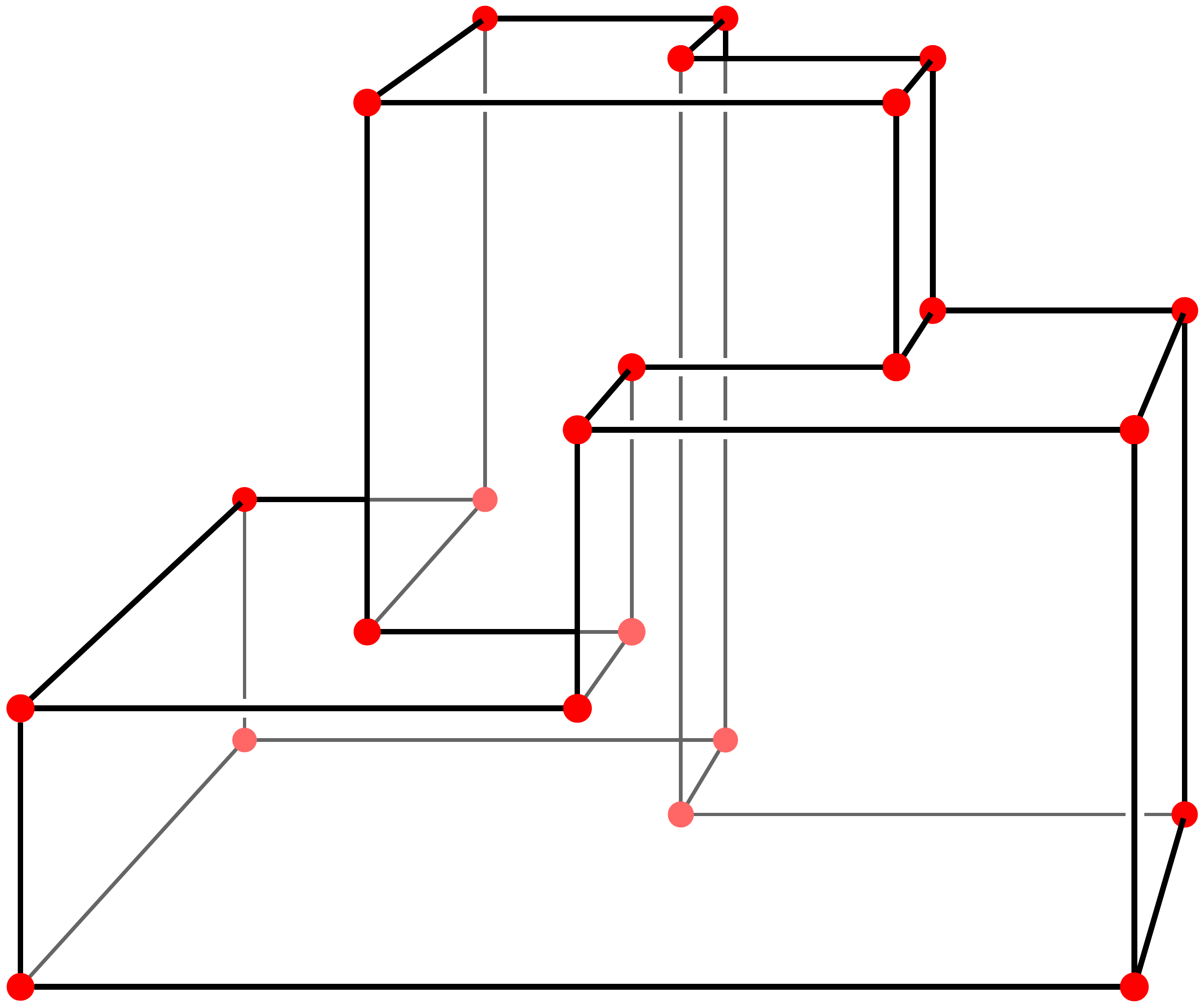}
\caption{Two $xyz$ graph representations of the truncated octahedron, from~\cite{Epp-GD-08}. The first has many edge crossings, while the second forms an orthogonal polyhedron but not a corner poyhedron. The results of this paper provide a corner polyhedron representation of the same graph.}
\label{fig:invertohedron}
\end{figure}

Biedl and Genc~\cite{BieGen-CCCG-04,BieGen-ESA-09} investigated analogues for orthogonal polyhedra of a different result about convex polyhedra, Cauchy's theorem~\cite{Cau-JEP-13} that specifying the shape of each face of a convex polyhedron fixes the shape of a whole polyhedron. In contrast, for nonconvex polyhedra, specifying the shape of each face is enough to fix the volume of the whole polyhedron under continuous motions~\cite{ConSabWal-BAG-97} but there exist flexible nonconvex polyhedra with fixed face shapes and an uncountably infinite number of global configurations~\cite{Con-IHES-77}. Analogously to Cauchy's theorem, fixing the shape of each face is enough to determine the shape of  an orthogonally convex polyhedron~\cite{BieGen-CCCG-04} or more generally of an orthogonal polyhedron with the topology of a sphere~\cite{BieGen-ESA-09}. Although these results concern a different problem, they suggest as do ours that orthogonal polyhedra may be closely analogous to convex polyhedra.

\emph{Rectangular layouts} form an important two-dimensional analogue of orthogonal polyhedra. These are planar drawings of cubic graphs for which each edge is axis-parallel and has no bends, and in which every face (including the outer face) is a rectangle. Rectangular layouts have applications in the visualization of geographic data~\cite{Rai-GR-34}, floorplan layout in architectural design~\cite{EarMar-AGT-79,Rin-EPB-88}, VLSI design~\cite{YeaSar-SJDM-95}, treemap information visualization~\cite{BruHuiWij-DV-00}, and graph drawing~\cite{KanHe-TCS-97}. A plane graph admits a rectangular layout, with a given partition of its outer faces into the sides of an outer rectangle, if and only if a variant of its dual graph (with one dual vertex for each side of the outer rectangle rather than a single vertex for the outer face) is a plane triangulated graph with an exterior quadrilateral and no separating triangles~\cite{KozKin-Nw-85}; a closely related characterization also holds for the cubic graphs that can be drawn on a grid with no bends but without requiring that the faces be rectangles~\cite{RahNisNaz-JGAA-03}. There is a combinatorial bijection between the rectangular layouts of a graph and its \emph{regular edge labelings} or \emph{transversal structures}, (improper) two-colorings of the edges of the dual graph together with an orientation for each edge satisfying certain constraints on the cyclic order in which the colored edges of each orientation meet each dual vertex~\cite{KanHe-TCS-97}. The regular edge labelings of a graph form a distributive lattice~\cite{Fus-GD-05,Fus-DM-08} and this lattice structure has algorithmic applications in finding rectangular layouts with additional properties~\cite{EppMumSpe-SCG-09,EppMum-WADS-09}. In our three-dimensional problem, as in the two-dimensional case, dual separating triangles form an obstacle to embedding. Additionally, our new results use a structure closely related to a regular edge labeling, with three edge colors rather than two, although the local constraints on colorings and orientations are different than in the two dimensional case.

Schnyder~\cite{Sch-SODA-90} developed algorithms for embedding planar graphs with the vertices on an integer grid, but with edges of arbitrary slopes, based on the concept now known as a \emph{Schnyder wood}, a (non-proper) 3-coloring of the edges of a maximal planar graph together with an orientation on each edge, with constraints about how the colors and orientations must be arranged at each vertex. Felsner and Zickfeld~\cite{FelZic-DCG-08} study geometric representations of Schnyder woods as orthogonal surfaces that are very similar to the corner polyhedra we study here. Their representation provides an embedding of the input graph onto the surface found by their representation, but the edges of the embedding do not in general follow the edges of the orthogonal surface. Our results on corner polyhedra also involve colorings and orientations of maximal planar graphs (the dual graphs of the graphs we wish to represent), and our impetus for considering this sort of combinatorial data on a graph came from these two papers as well as from the work on two-dimensional rectangular representations and regular edge labelings. However there seems to be no direct connection between our results and the results of Schnyder, Felsner, and Zickfeld: our colorings are proper and our orientations do not form Schnyder woods, our polyhedral representation is not of the colored and oriented graph but of its dual, and unlike Felsner and Zickfeld we develop a polyhedral representation of a graph in which the graph edges and the polyhedron edges coincide.

More generally there has been a large body of research on planar and spatial embeddings of graphs on low-dimensional grids, or otherwise having a small number of edge slopes. Much of this work allows the edges of the graph to bend in order to follow the edges of the grid, and seeks to minimize the number of bends~\cite{AziBie-GD-04,Kan-WG-93,Tam-SJC-87,Woo-JGAA-03,Woo-TCS-03}; as we show, for the graphs of corner polyhedra, isometrically projection leads to a hexagonal grid drawing with three slopes and only two bends, improving a bound of three bends by Kant~\cite{Kan-WG-93} which however applies more generally to all 3-connected cubic planar graphs.
Graph drawing researchers have also studied the \emph{slope number} of a graph, the minimum number of distinct edge slopes needed to draw the graph in the plane with straight line edges and no bends~\cite{DujEppSud-CGTA-07,DujSudWoo-CGTA-07,PacPal-EJC-06,KesPacPal-CGTA-08,MukSze-CGTA-09}. Every graph of an orthogonal polyhedron, and every $xyz$ graph, has slope number three, since a three-dimensional orthogonal representation may be transformed into a planar drawing with three slopes (allowing edge crossings) by axonometric projection. However, not every graph with slope number three comes from an orthogonal drawing in this way; for instance, $K_3$ has slope number 3 but has no orthogonal drawing, as does Fig.~\ref{fig:nonpolyhedral}.

Both $xyz$ graphs and our polyhedral representations can be viewed as embedding the given graph onto the three-dimensional integer grid, with axis-aligned edges that can have arbitrary lengths. Embedding a graph onto a grid with unit-length edges is NP-complete~\cite{BhaCos-IPL-87}. Embeddings with unit length edges that additionally preserve distances between vertices farther than one unit apart can be found in polynomial time, when they exist, for two- and three-dimensional integer lattices~\cite{Epp-EJC-05} and hexagonal and diamond lattices~\cite{Epp-GD-08-diamond}, but the graphs that have such embeddings (the \emph{partial cubes}) form a restricted subclass of all graphs. The embeddings we consider in this paper only apply to cubic (3-regular) graphs, and of the infinitely many known cubic partial cubes all but one (the Desargues graph) are bipartite polyhedral graphs~\cite{Epp-EJC-06}, and may therefore be represented as orthogonal polyhedra using the algorithms we describe here.

A hint that care is needed in defining orthogonal polyhedra is given by Donoso and O'Rourke~\cite{DonORo-01}. As they show, spherical and toroidal polyhedra in which all faces are rectangles (even allowing adjacent faces to be coplanar) must also have all faces and edges axis-parallel for some orientation of the polyhedron, but there exist higher-genus polyhedra with rectangular faces that have more than three edge orientations. Biedl et al.~\cite{BieDemDem-CCCG-98} define two interesting subclasses of orthogonal polyhedra, which they call \emph{orthostacks} and \emph{orthotubes}; they also consider orthogonal polyhedra that are somewhat more general than ours, in that they allow the graph of the polyhedron to be disconnected (resulting in faces that are not simple polygons) while we do not.

The bipartite cubic polyhedral graphs and their dual graphs, the Eulerian triangulations, have also been studied independently of their geometric representations, in connection with Barnette's conjecture that all bipartite cubic polyhedral graphs are Hamiltonian~\cite{Bar-RPC-69}. Batagelj, Brinkmann and McKay~\cite{Bat-CGT-87,BriMcK-CMCC-07} described a method for reducing any Eulerian triangulation to a simpler graph in the same class, based on which Brinkmann and McKay show how to efficiently generate all sufficiently small Eulerian triangulations; they also generate 4-connected Eulerian triangulations by filtering them from the larger set of all Eulerian triangulations. Our proof that 4-connected Eulerian triangulations are dual to corner polyhedra uses a different reduction scheme that remains within the class of 4-connected Eulerian triangulations.

\section{Corner polyhedra and rooted cycle covers}
\label{sec:corner}

\begin{figure}[t]
\centering\includegraphics[height=2.75in]{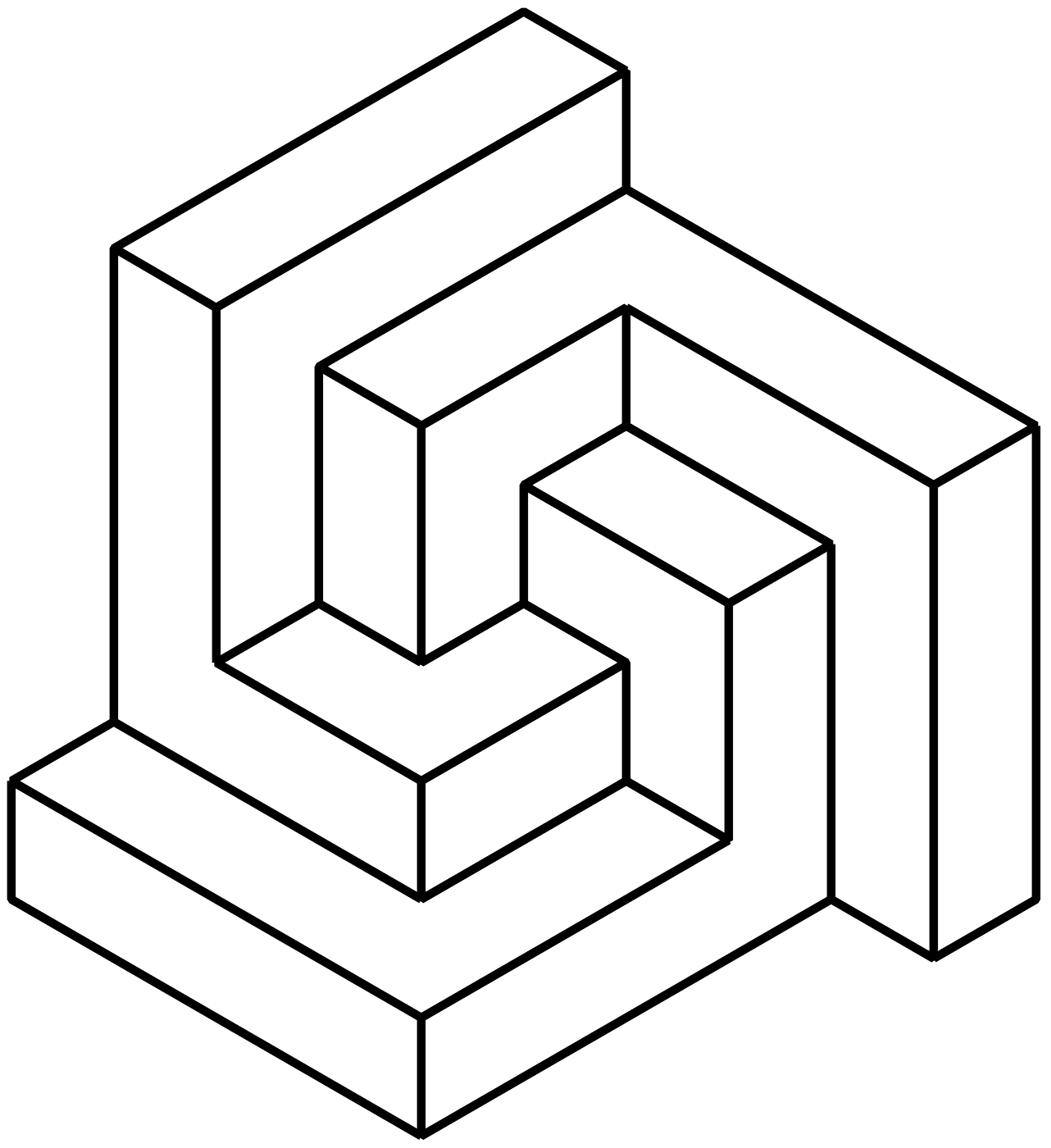}\qquad\qquad
\includegraphics[height=2.75in]{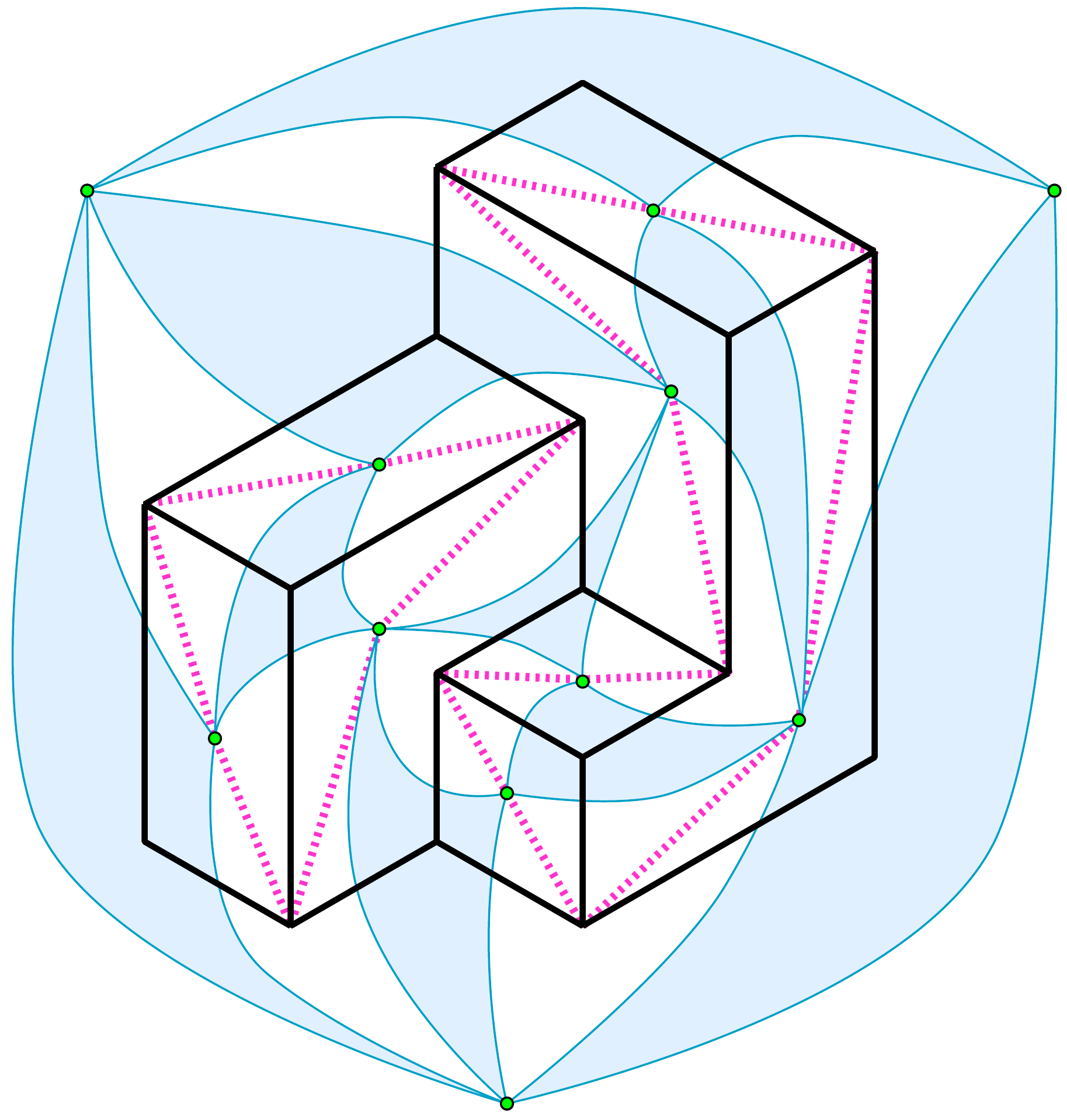}
\caption{Left: isometric projection of a corner polyhedron for a truncated rhombic dodecahedron. Right: connecting each sharp corner of a visible face to the dual vertex of the face produces a rooted cycle cover.}
\label{fig:corners}
\end{figure}

As stated in the introduction, we define a corner polyhedron to be a simple orthogonal polyhedron with the additional property that three faces (the \emph{back faces}) are oriented towards the vector $(-1,-1,-1)$, and all remaining faces (the \emph{front faces}) are oriented towards the vector $(1,1,1)$. The three back faces necessarily share a vertex, the \emph{hidden vertex}. Parallel projection of a corner polyhedron onto a plane perpendicular to the vector $(1,1,1)$ gives rise to a drawing, the so-called \emph{isometric projection}, in which the axis-parallel edges of the three-dimensional polyhedron are mapped to three sets of parallel lines that form angles of $\pi/3$ with respect to each other; see Fig.~\ref{fig:corners} (left) for an example. If the edges of the corner polyhedron have integer lengths, the resulting isometric projection is a drawing of all of the vertices of the polyhedron, except the hidden vertex, on the hexagonal lattice. It is possible to include the hidden vertex as well by connecting it to its three neighbors by lattice paths, one of which is straight and the other two have one bend each; the resulting drawing of the whole graph has two bends.

In an isometric drawing of a corner polyhedron, none of the faces can have an interior angle of $5\pi/3$, for any simple orthogonal polyhedron with such a projected angle would have more than three back faces. Additionally, each face (being the projection of a planar orthogonal polygon) has edges of only two of the three possible slopes. Therefore, each face of the drawing has the shape of a \emph{double staircase}: there are two vertices at which the interior angle is $\pi/3$, and the two sequences of interior angles on the paths between these vertices alternate between interior angles of $2\pi/3$ and $4\pi/3$.
(Conversely, it follows by Thurston's results on height functions~\cite{Thu-AMM-90} that a drawing in the hexagonal lattice for which all faces have this shape comes from a three-dimensional orthogonal surface.)
Because a simple orthogonal polyhedron forms a planar graph in which all faces have an even number of edges, it must be bipartite; one of its two color classes consists of the vertices having sharp face angle, and the other color class contains all the other vertices.

The two-to-one correspondence between interior faces and vertices with sharp angles gives rise to an important structure on the graph of the polyhedron, which we find simpler to describe in terms of its dual graph. The dual graph of a corner polyhedron (as with a bipartite cubic polyhedral graph more generally) is an \emph{Eulerian triangulation}, a maximal planar graph in which every vertex has even degree. It has a unique planar embedding, for which all the faces are triangles; the triangles may be two-colored so that the two triangles that share each edge have different colors. Within each interior face of the projected corner polyhedron, we connect the dual vertex to its two sharp corners. The result of forming these connections is a structure that we call a \emph{rooted cycle cover}: a set of vertex-disjoint cycles in the dual Eulerian triangulation, that cover every dual vertex except for the three vertices of the \emph{root triangle} dual to the hidden vertex, and that include exactly one edge from every triangle with the same color as the root triangle. Conversely, as we show, every rooted cycle cover of an Eulerian triangulation gives rise to a corner polyhedron representation of its dual graph. This equivalence between a combinatorial structure (a rooted cycle cover) and a geometric structure (a corner polyhedron) is a key component of our characterization of the graphs of corner polyhedra.

Specifically, we prove the following results:

\begin{theorem}
\label{thm:corner-cover}
A graph $\G$ can be represented as a corner polyhedron, with a specified vertex $v$ as the single hidden vertex, if and only if the dual graph of $\G$ has a cycle cover rooted at the triangle dual to $v$.
\end{theorem}

\begin{theorem}
\label{thm:4-connected}
If $\G$ is a cubic bipartite polyhedral graph with a 4-connected dual, then it can be represented as a corner polyhedron.
\end{theorem}

If $\Delta$ is an Eulerian triangulation (the dual to a cubic bipartite planar graph),
with a chosen root triangle $\delta$, then we may uniquely two-color the triangles of $\Delta$ so that any two adjacent triangles are adjacent. For any separating triangle $\gamma$ of $\Delta$, this coloring will assign equal colors to the three triangles that are on the side of $\gamma$ that does not contain $\delta$ and that are incident to one of the edges of $\gamma$. We say that $\gamma$ has \emph{even parity} if these three triangles have the same color as $\delta$, and \emph{odd parity} otherwise.

\begin{theorem}
\label{thm:corner-characterization}
If $\G$ is a cubic bipartite graph with a non-4-connected dual, and $v$ is any vertex of $\G$, then $\G$ has a corner representation for which $v$ is the hidden vertex if and only if all separating triangles have odd parity with respect to the root triangle dual to $v$.
\end{theorem}

In the rest of this section, we provide a rough sketch of our proofs of the claims stated above. The detailed proof of Theorem~\ref{thm:corner-cover} comprises Appendices II, III, and~IV, the proof of Theorem~\ref{thm:4-connected} is in Appendices V and~VI, and the proof of Theorem~\ref{thm:corner-characterization} is in Appendix~VII.

Given a corner polyhedron representation of a graph $\G$, we form a rooted cycle cover of the dual of $\G$ as described above. The graph of the polyhedron has a unique 3-edge-coloring given by the orientations of its edges. These colors may be carried over to the dual Eulerian triangulation $\Delta$, but in $\Delta$ they do not form an edge coloring but rather a \emph{rainbow partition}, a partition of the edges of $\Delta$ into three monochromatic subgraphs such that each triangle of $\Delta$ participates in each of these subgraphs. Each monochromatic subgraph must be biconnected, which implies that it can be oriented as an $st$-planar graph with the two terminals on the chosen root triangle. From the corner polyhedron representation we may also determine an orientation for each edge of $\Delta$ based on which of the two adjacent primal faces is on which side of the edge in the isometric drawing; this orientation can be shown to be simultaneously $st$-planar in each monochromatic subgraph, and the requirement that each face of the corner polyhedron is a double staircase may be translated via planar duality into local consistency conditions on the orientations of the edges incident to each dual vertex.

Conversely, given a cycle cover, we can use it to define an orientation on the edges of the dual Eulerian triangulation, in which the orientations of the edges alternate around each vertex except within two triangles incident to the vertex, which both have the same color as the root triangle in the two-coloring of triangles dual to the bipartition of the input graph. We call a rainbow partition together with an orientation having this property a \emph{regular edge labeling}. Based on this property, we can show that a regular edge labeling is simultaneously $st$-planar in each monochromatic subgraph, and that unions of two monochromatic subgraphs (with the orientation reversed in one of the two) are again $st$-planar.
By numbering the vertices of each of these bichromatic subgraphs consistently with the $st$-planar orientation, and using these numbers as the coordinates of the face planes of $\G$, we may construct a representation of $\G$ as a corner polyhedron. It follows from this construction that $\G$ is the graph of a corner polyhedron if and only if $\G$ has a rooted cycle cover.

We show that every 4-connected Eulerian triangulation can be decomposed into smaller 4-connected Eulerian triangulations using three operations: splitting the graph on a 4-cycle, removing a pair of adjacent degree-four vertices, and collapsing two opposite edges of a degree-four vertex. We use this structure as the basis for a proof by induction that every 4-connected Eulerian triangulation has a cycle cover. Therefore, by the above equivalence between cycle covers and corner polyhedron representations, every cubic bipartite graph with a 4-connected dual can be represented as a corner polyhedron.

If a bipartite cubic graph does not have a 4-connected dual, but all of its separating triangles have odd parity, we may split the dual graph at all of its separating triangles, find a cycle cover separately for each subgraph created by this splitting process, and form a cycle cover of the original graph as the union of these separate cycle covers. On the other hand, if there exists a separating triangle with the same parity as the chosen root vertex, we show that no cycle cover can exist.

\section{$xyz$ polyhedra}
\label{sec;xyz}

In our previous paper~\cite{Epp-GD-08} we defined an \emph{$xyz$ graph} to be a cubic graph embedded in three dimensional space, with axis parallel edges, such that the line through each edge passes through no other vertices of the graph. We can extend this definition to an \emph{$xyz$ polyhedron}, a simple orthogonal polyhedron whose skeleton forms an $xyz$ graph.
Alternatively, we can consider a weaker definition: a \emph{singly-intersecting} simple orthogonal polyhedron is a simple orthogonal polyhedron with the property that, for any two faces with a nonempty intersection, their intersection is a single line segment.  Geometrically, the intersection of two faces lies along the line of intersection of their planes, so a singly-intersecting polyhedron must be an $xyz$ polyhedron, but not necessarily vice versa. However, the graphs of the two classes of polyhedra are the same: by perturbing the face planes of a singly-intersecting polyhedron, one may obtain an $xyz$ polyhedron that represents the same graph.

As we showed in our previous paper, a planar $xyz$ graph must be 3-connected and bipartite, and the same results hold for $xyz$ polyhedra. Our main result is a converse to this:

\begin{theorem}
\label{thm:xyz}
The following three classes of graphs are equivalent:
\begin{itemize}
\item Cubic 3-connected bipartite planar graphs,
\item Graphs of $xyz$ polyhedra, and
\item Graphs of singly-intersecting simple orthogonal polyhedra.
\end{itemize}
\end{theorem}

The idea of the proof is to use induction on the number of separating triangles in the dual Eulerian triangulation. If there are no separating triangles, the given graph has a corner polyhedron representation and we are done. Otherwise, we find a separating triangle that splits the dual graph into two smaller Eulerian triangulations, one of them four-connected. By induction, the other one has a polyhedral representation, and we can replace one vertex of this polyhedron by a very small copy of a corner polyhedron representing the other split component, forming a representation of the overall polyhedron.

The repeated replacement of polyhedron vertices by small corner polyhedra may eventually lead to features of exponentially small size, but this issue can be sidestepped by replacing the coordinates of the faces by small integers, leading to a polyhedral representation in which all vertex coordinates are integers in the interval $[1,n/4]$.

The full proof is in Appendix~VIII.

\section{Simple orthogonal polyhedra}

\begin{figure}[t]
\centering\includegraphics[width=2.5in]{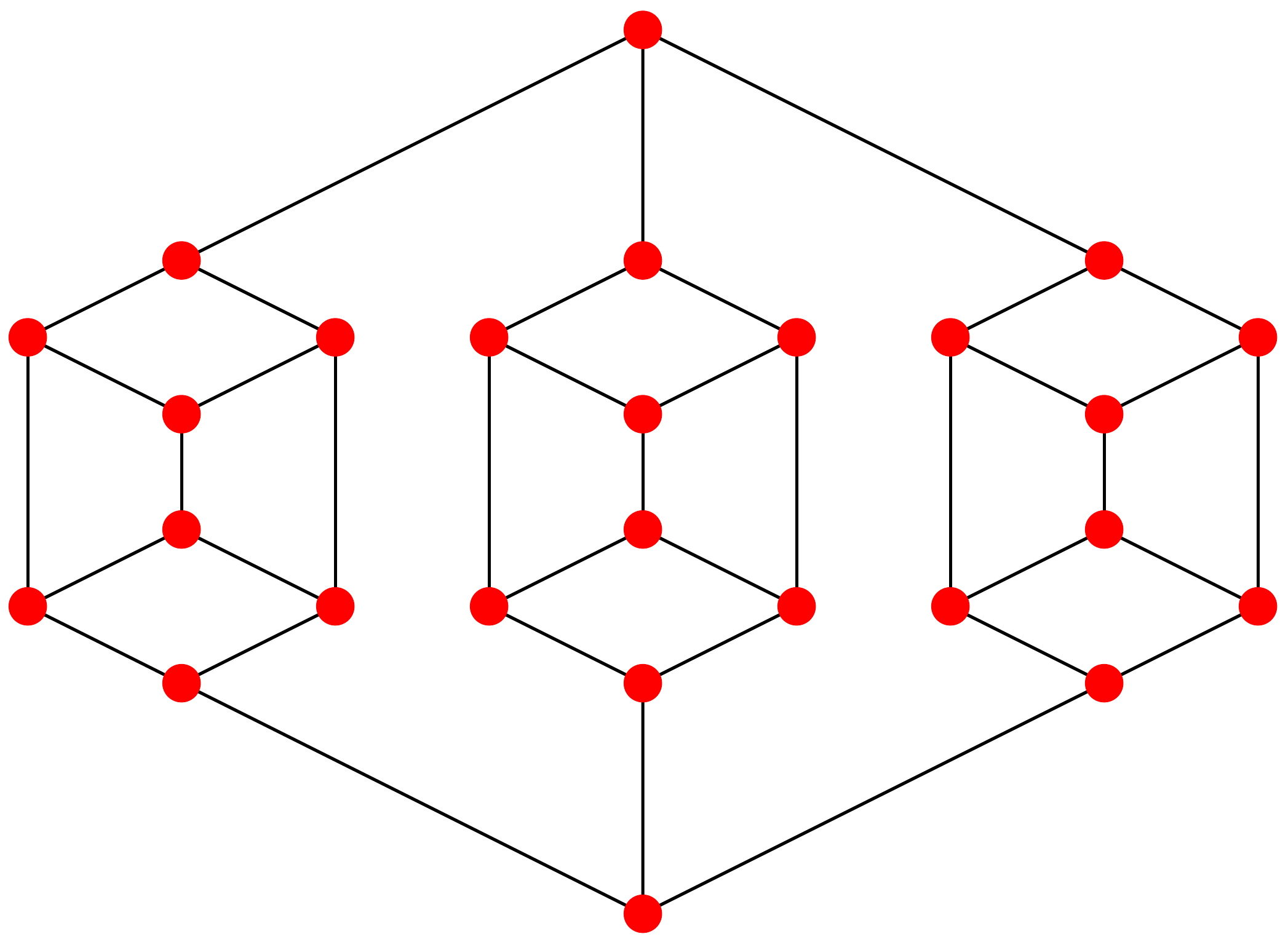}
\caption{A 2-connected bicubic planar graph that is not the graph of a simple orthogonal polyhedron.}
\label{fig:nonpolyhedral}
\end{figure}

As can be seen in Fig.~\ref{fig:orthotypes} (right), the graph of an arbitrary simple orthogonal polyhedra may not always be 3-connected, although it is always 2-connected. Pairs of faces of the polyhedron may meet in multiple edges, and the removal of any two of these edges (or the removal of endpoints from any two of these edges) leaves a disconnected graph. Therefore, there exist graphs simple orthogonal polyhedra that are not graphs of $xyz$ polyhedron, and we need to use a more general class of graphs to characterize the simple orthogonal polyhedra.
Replacing the 3-connectivity condition in the characterization of $xyz$ polyhedra by 2-connectivity would be too general, however. Not every 2-connected bipartite 3-regular graph is the graph of a simple orthogonal polyhedron; for instance, the graph depicted in Fig.~\ref{fig:nonpolyhedral} is not the graph of a simple orthogonal polyhedron, as the results in this section will show.

Instead, our characterization uses the SPQR tree, a standard tool for representing the planar embeddings of a graph in terms of its triconnected components~\cite{Mac-DMJ-37,HopTar-SJC-73,DiBTam-FOCS-89,DiBTam-ICALP-90,GutMut-GD-01}.  The triconnected components of a graph may be multigraphs rather than simple graphs; for instance, the graph shown in Fig.~\ref{fig:nonpolyhedral} has seven triconnected components: three cubes (the R nodes of an SPQR tree), three 4-cycles (the S nodes of an SPQR tree), and a multigraph with two vertices and three edges (the P node of an SPQR tree).

\begin{theorem}
\label{thm:atom}
The following three classes of graphs are equivalent:
\begin{itemize}
\item Cubic 2-connected graphs in which every triconnected component is either a bipartite polyhedral graph or an even cycle,
\item Bipartite cubic planar graphs in which the removal of any two vertices leaves at most two connected components (counting an edge between the two vertices as a component, if one exists), and
\item Graphs of simple orthogonal polyhedra.
\end{itemize}
\end{theorem}

For instance, the graph of Fig.~\ref{fig:nonpolyhedral} cannot be the graph of a simple orthogonal polyhedron, because removing its top and bottom vertex leaves three connected components, and because it has a multigraph as one of its triconnected components.

The main idea of the proof of Theorem~\ref{thm:atom} (in Appendix~IX) is to decompose the graph into triconnected components. In one direction, we show that when a simple orthogonal polyhedron is decomposed into triconnected components, each non-cyclic component inherits a (nonpolyhedral) geometric embedding based on which we can rule out the possibility that any component forms a P node in the SPQR tree. In the other direction, when we are given a graph in which each triconnected component has the stated form, we represent each bipartite polyhedral triconnected component as a simple orthogonal polyhedron using the results of the previous section, and use the even cycles in the SPQR tree to guide a sequence of gluing steps that combine each of these polyhedral pieces into a single polyhedron that represents the whole graph.
Our proof technique leads to a stronger result: if $\G$ is the graph of a simple orthogonal polyhedron, then every planar embedding of $\G$ can be represented as a simple orthogonal polyhedron.

\section{Algorithms}

Below we outline an algorithm that takes a 2-connected cubic planar graph as an input and embeds it as a simple orthogonal polyhedron, when such a representation exists. The algorithms for taking as input a 3-connected graph and representing it either as an $xyz$ polyhedron or as a corner polyhedron, when such a representation exists, are similar but with fewer steps.
\begin{enumerate}
\item Decompose the graph into its triconnected components, as represented by an SPQR tree, in linear time~\cite{GutMut-GD-01,HopTar-SJC-73}. Check that the SPQR tree does not contain any P nodes (triconnected components that are multigraphs rather than simple graphs). If it does, report that no orthogonal polyhedral representation exists and abort the algorithm.
\item Transform each atom (triconnected component that is not a cycle) into its dual Eulerian triangulation, using a linear time planar embedding algorithm~\cite{HopTar-JACM-74}. If any atom is nonplanar or has a non-Eulerian dual, report that no orthogonal polyhedral representation exists and abort the algorithm.
\item Partition each Eulerian triangulation into 4-connected Eulerian triangulations by splitting it on its separating triangles. All separating triangles may be found in $O(n)$ time~\cite{ChiNis-SJC-85,ChrEpp-TCS-91}.
\item Recursively decompose each 4-connected Eulerian triangulation into simpler 4-connected Eulerian triangulations using separating 4-cycles, pairs of adjacent degree-4 vertices, and isolated degree-4 vertices. While returning from the recursion, undo the steps of the decomposition and build a cycle cover for the Eulerian triangulation.
\item Convert the cycle covers into regular edge labelings by a simple local pattern matching rule.
\item For each pair of colors $x$ and $y$ in the rainbow partition, construct the subgraph $\Delta_{xy}$ formed by edges with those two colors, oriented by reversing the orientations of one of the two colors from the orientation given by the regular edge labeling, and find an $st$-numbering of each such graph using breadth-first search.
\item For each graph dual to one of the 4-connected Eulerian triangulations, use the $st$-numbering to construct a representation of the graph as a corner polyhedron: the coordinates of each vertex of the corner polyhedron are triples of numbers from the $st$-numbering, one from each of the three bichromatic subgraphs of $\Delta$.
\item Glue the corner polyhedra together to form orthogonal polyhedra dual to each non-4-connected Eulerian triangulation. In order to perform this step and the next one efficiently, we represent vertex coordinates implicitly throughout these steps, as positions within a doubly linked list, and after the gluing is completed convert these implicit positions back into numeric values.
\item  Glue 3-connected polyhedra together to form arbitrary simple orthogonal polyhedra.
\end{enumerate}
The algorithm needs $O(n)$ expected time when implemented using randomized hash tables; deterministically, it can be implemented to run in $O(n(\log\log n)^2/(\log\log\log n))$ time with linear space. The most complicated step, and the only step that uses more than $O(n)$ deterministic running time, is the one in which we decompose each 4-connected Eulerian triangulation into simpler 4-connected Eulerian triangulations; this step uses a data structure for testing adjacency of pairs of vertices in a dynamic plane graph. If the adjacency testing data structure is implemented using linear-space deterministic integer searching data structures~\cite{AndTho-STOC-00}, the total running time is $O(n(\log\log n)^2/(\log\log\log n))$, whereas if it is implemented using hash tables, the total expected running time is $O(n)$.

\begin{theorem}
We may construct a representation of a given graph as a corner polyhedron, $xyz$ polyhedron, or simple orthogonal polyhedron, when such a representation exists, in $O(n)$ randomized expected time, or deterministically in $O(n(\log\log n)^2/(\log\log\log n))$ time with linear space.
\end{theorem}

Detailed descriptions of each step together with the running time analyses are given in Appendix~X.

\section{Conclusions}

\begin{figure}[t]
\centering\includegraphics[height=2.25in]{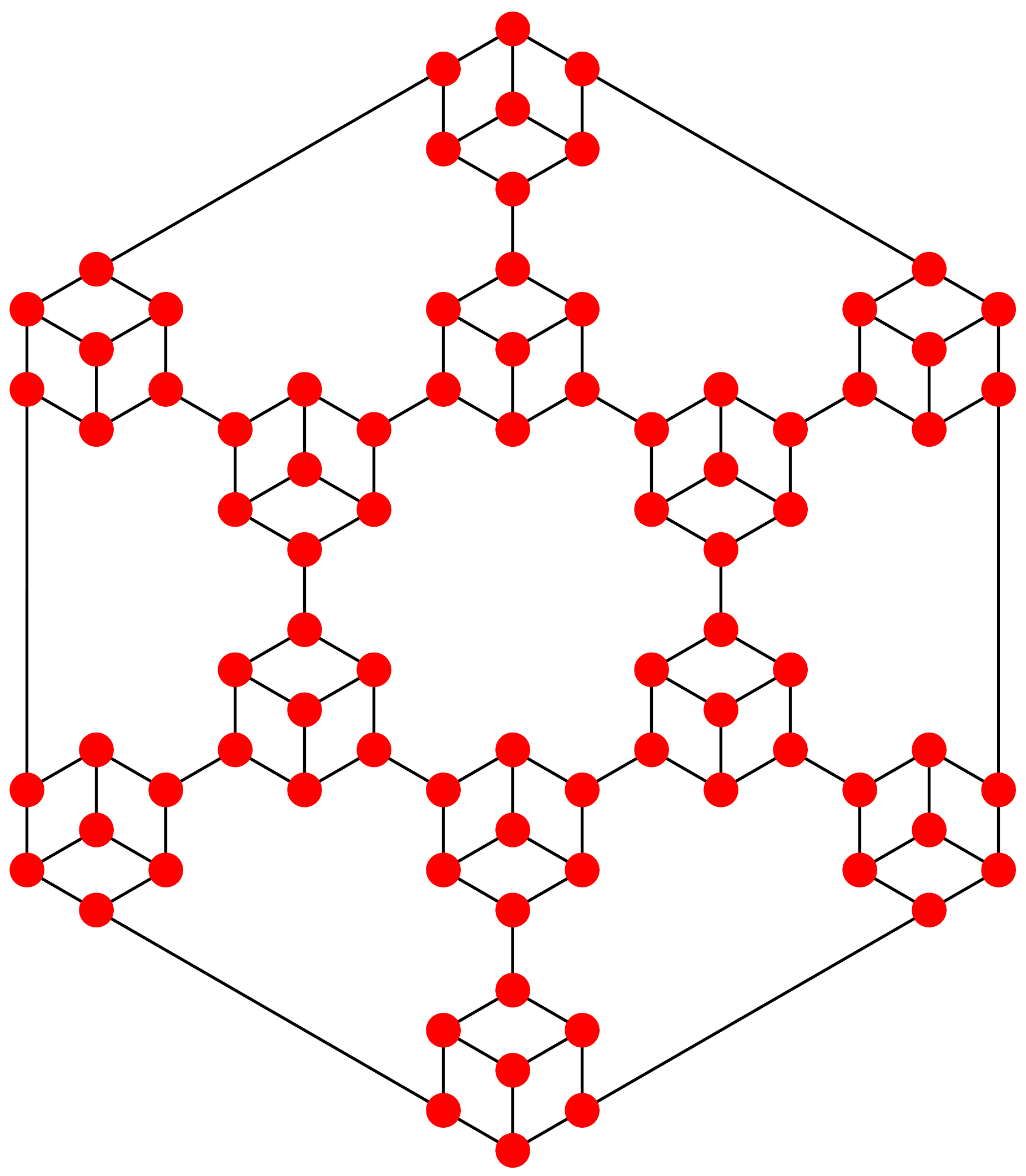}\qquad\qquad
\includegraphics[height=2.25in]{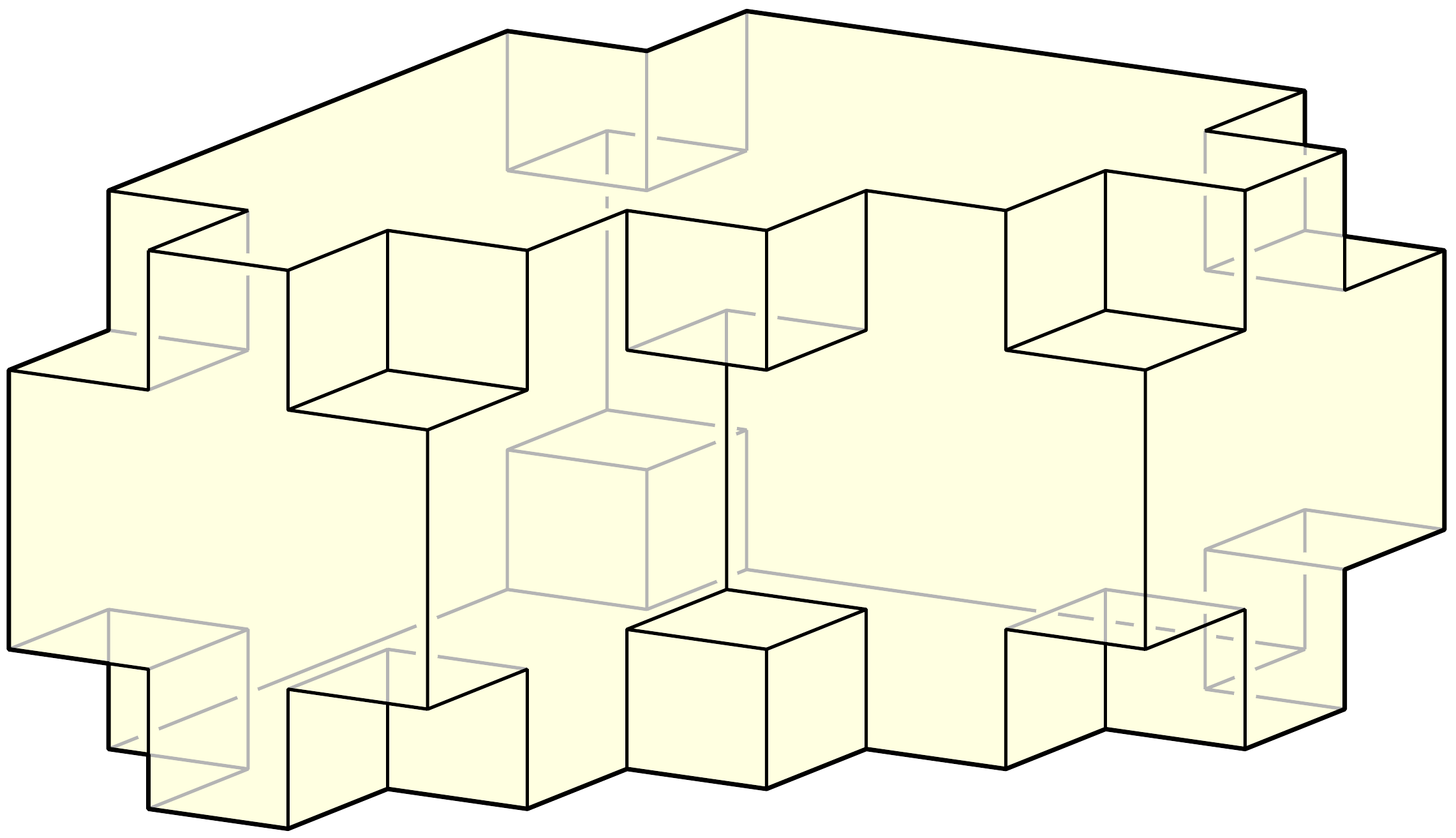}
\caption{A bipartite cubic polyhedral graph (left) that has no orthogonally convex representation as a simple orthogonal polyhedron, and its representation as a nonconvex simple orthogonal polyhedron (right).}
\label{fig:escherized}
\end{figure}

We have defined three interesting classes of orthogonal polyhedra, and provided exact graph-theoretic characterizations of the graphs that may be represented by these polyhedra. In particular, every bipartite cubic polyhedral graph has a representation as an orthogonal polyhedron.

The following problems remain open for additional investigation:
\begin{itemize}
\item Corner polyhedra are orthogonally convex, and orthogonally convex simple polyhedra may be represented as $xyz$ polyhedra, so the orthogonally convex simple polyhedra are sandwiched between two of the classes of polyhedron that we can precisely characterize and for which we provide polynomial time recognition algorithms. However, not every bipartite cubic polyhedral graph has an orthogonally convex representation: if enough dual separating triangles share edges with each other, they may interfere with each other and force any orthogonal polyhedral representation to be nonconvex (Fig.~\ref{fig:escherized}). Is there a simple condition on the position of the dual separating triangles that characterizes orthogonally convex simple orthogonal polyhedra, and can we test this condition in polynomial time?

\item The \emph{orthostacks} defined by Biedl et al.~\cite{BieDemDem-CCCG-98} are also intermediate between corner polyhedra and $xyz$ polyhedra. Can we characterize their graphs?

\item Our results hold only for orthogonal polyhedra with three perpendicular edges at each vertex. What if we relax this requirement, and either allow edges to be non-perpendicular (as in the bidiakis cube, Fig.~\ref{fig:nonsimple}, center) or to have more than three edges per vertex (Fig.~\ref{fig:nonsimple}, left). In this case, as the figures show, the graph of the polyhedron does not need to be bipartite. Is there some way of replacing problematic vertices by small subgraphs in which all vertices have degree three and all edges are perpendicular, allowing the methods from this paper to apply?

\item Given the hardness of nonplanar $xyz$ graph recognition~\cite{Epp-GD-08} it seems likely that it will also be difficult to determine whether a given graph is the graph of an orthogonal polyhedron with nonzero genus (Fig.~\ref{fig:nonsimple}, right), but what about graphs for which an $xyz$ graph representation is already known? In that case, how difficult is it to detemine whether the faces of the $xyz$ representation can be untangled to form a polyhedral representation?

\item Although we provided a linear time algorithm for finding orthogonal polyhedron representations, its constant factors are large (one case involves up to 1225 adjacency tests). Additionally, while not an obstacle for practical implementation, the need for hashing based data structures in order to achieve linear time is a theoretical irritation. Is there a simple deterministic linear time algorithm for finding cycle covers in 4-connected Eulerian triangulations?
\end{itemize}

\newpage
\section*{Acknowledgements}

Work of David Eppstein was supported in part by NSF grant
0830403 and by the Office of Naval Research under grant
N00014-08-1-1015. We thank Patrizio Angelini, Mike Dillencourt, Fabrizio Frati, Mike Goodrich, Maarten L{\"o}ffler, and Brendan McKay for helpful conversations related to the subject of this paper.

{\raggedright
\bibliographystyle{abuser}
\bibliography{steinitz}}

\begin{thebibliography}{10}

\bibitem{AndTho-STOC-00}
A.~A. Andersson and M.~Thorup.
\newblock {Tight(er) worst-case bounds on dynamic searching and priority
  queues}.
\newblock {\em Proc. 32nd ACM Symp. Theory of Computing (STOC 2000)},
  pp.~335--342, 2000,
  \href{http://dx.doi.org/10.1145/335305.335344}{doi:10.1145/335305.335344}.

\bibitem{AziBie-GD-04}
S.~Aziza and T.~Biedl.
\newblock {Hexagonal grid drawings: algorithms and lower bounds}.
\newblock {\em Proc. 12th Int. Symp. Graph Drawing (GD 2004)}, pp.~18--24.
  Springer-Verlag, Lecture Notes in Computer Science 3383, 2005.

\bibitem{Bal-PJM-61}
M.~L. Balinski.
\newblock {On the graph structure of convex polyhedra in $n$-space}.
\newblock {\em Pacific J. Math.} 11(2):431--434, 1961,
  \url{http://projecteuclid.org/euclid.pjm/1103037323}.

\bibitem{Bar-RPC-69}
D.~W. Barnette.
\newblock {Conjecture 5}.
\newblock {\em Recent Progress in Combinatorics}, p.~343. Academic Press, 1969.

\bibitem{Bat-CGT-87}
V.~Batagelj.
\newblock {An improved inductive definition of two restricted classes of
  triangulations of the plane}.
\newblock {\em Combinatorics and Graph Theory, Warsaw, 1987}, pp.~11--18. PWN,
  Banach Center Publ.~25, 1989.

\bibitem{BhaCos-IPL-87}
S.~Bhatt and S.~Cosmodakis.
\newblock {The complexity of minimizing wire lengths in VLSI layouts}.
\newblock {\em Inform. Proc. Lett.} 25:263--267, 1987,
  \href{http://dx.doi.org/10.1016/0020-0190(87)90173-6}{doi:10.1016/0020-0190(%
87)90173-6}.

\bibitem{BieDemDem-CCCG-98}
T.~Biedl, E.~Demaine, M.~Demaine, A.~Lubiw, M.~Overmars, J.~O'Rourke,
  S.~Robbins, and S.~Whitesides.
\newblock {Unfolding some classes of orthogonal polyhedra}.
\newblock {\em Proc. 10th Canadian Conference on Computational Geometry
  (CCCG'98)}, 1998, \url{http://erikdemaine.org/papers/CCCG98b/}.

\bibitem{BieGen-CCCG-04}
T.~Biedl and B.~Genc.
\newblock {When can a graph form an orthogonal polyhedron?}
\newblock {\em Proc. 16th Canadian Conf. Computational Geometry (CCCG 2004)},
  pp.~53--56, 2004, \url{http://www.cccg.ca/proceedings/2004/15.pdf}.

\bibitem{BieGen-ESA-09}
T.~Biedl and B.~Genc.
\newblock {Cauchy's theorem for orthogonal polyhedra of genus 0}.
\newblock {\em Proc. 17th European Symp. Algorithms (ESA 2009)}, pp.~71--83.
  Springer-Verlag, Lecture Notes in Computer Science 5757, 2009,
  \href{http://dx.doi.org/10.1007/978-3-642-04128-0\_7}{doi:10.1007/978-3-642-%
04128-0\_7}, \url{http://www.cs.uwaterloo.ca/research/tr/2008/CS-2008-26.ps}.

\bibitem{BriMcK-CMCC-07}
G.~Brinkmann and B.~D. McKay.
\newblock {Fast generation of planar graphs}.
\newblock {\em Communications in Mathematical and in Computer Chemistry}
  58(2):323--357, 2007,
  \url{http://cs.anu.edu.au/people/bdm/papers/plantri-full.pdf}.

\bibitem{BruHuiWij-DV-00}
M.~Bruls, K.~Huizing, and J.~J. van Wijk.
\newblock {Squarified treemaps}.
\newblock {\em Data Visualization 2000: Proc. Joint Eurographics and IEEE TCVG
  Symp. on Visualization}, pp.~33{--}42. Springer-Verlag, 2000,
  \url{http://www.win.tue.nl/~vanwijk/stm.pdf}.

\bibitem{Cau-JEP-13}
A.~L. Cauchy.
\newblock {Sur les polygones et poly{\`e}dres}.
\newblock {\em J. {\'E}cole Polytechnique} 19:87--98, 1813.

\bibitem{ChiNis-SJC-85}
N.~Chiba and T.~Nishizeki.
\newblock {Arboricity and subgraph listing algorithms}.
\newblock {\em SIAM J. Comput.} 14(1):210{--}223, 1985,
  \href{http://dx.doi.org/10.1137/0214017}{doi:10.1137/0214017}.

\bibitem{ChrEpp-TCS-91}
M.~Chrobak and D.~Eppstein.
\newblock {Planar orientations with low out-degree and compaction of adjacency
  matrices}.
\newblock {\em Theoretical Computer Science} 86(2):243--266, September 1991,
  \href{http://dx.doi.org/10.1016/0304-3975(91)90020-3}{doi:10.1016/0304-3975(%
91)90020-3}.

\bibitem{Con-IHES-77}
R.~Connelly.
\newblock {A counterexample to the rigidity conjecture for polyhedra}.
\newblock {\em Pub. Math. de l'IH{\'E}S} 47:333--338, 1977,
  \href{http://dx.doi.org/10.1007/BF02684342}{doi:10.1007/BF02684342}.

\bibitem{ConSabWal-BAG-97}
R.~Connelly, I.~Sabitov, and A.~Walz.
\newblock {The bellows conjecture}.
\newblock {\em Beitr{\"a}ge Algebra Geom.} 38:1--10, 1997,
  \url{http://www.emis.de/journals/BAG/vol.38/no.1/1.html}.

\bibitem{Csa-ASMS-49}
A.~Cs{\'a}sz{\'a}r.
\newblock {A polyhedron without diagonals}.
\newblock {\em Acta Sci. Math. Szeged} 13:140--142, 1949.

\bibitem{DiBTam-FOCS-89}
G.~Di~Battista and R.~Tamassia.
\newblock {Incremental planarity testing}.
\newblock {\em Proc. 30th Symp. Foundations of Computer Science (FOCS 1989)},
  pp.~436--441, 1989,
  \href{http://dx.doi.org/10.1109/SFCS.1989.63515}{doi:10.1109/SFCS.1989.63515%
}.

\bibitem{DiBTam-ICALP-90}
G.~Di~Battista and R.~Tamassia.
\newblock {On-line graph algorithms with SPQR-trees}.
\newblock {\em Proc. 17th Internat. Colloq. Automata, Languages and Programming
  (ICALP 1990)}, pp.~598--611. Springer-Verlag, Lecture Notes in Computer
  Science 443, 1990,
  \href{http://dx.doi.org/10.1007/BFb0032061}{doi:10.1007/BFb0032061}.

\bibitem{DilSmi-DM-96}
M.~B. Dillencourt and W.~D. Smith.
\newblock {Graph-theoretical conditions for inscribability and Delaunay
  realizability}.
\newblock {\em Discrete Math.} 161(1--3):63--77, 1996,
  \href{http://dx.doi.org/10.1016/0012-365X(95)00276-3}{doi:10.1016/0012-365X(%
95)00276-3}.

\bibitem{DonORo-01}
M.~Donoso and J.~O'Rourke.
\newblock {Nonorthogonal polyhedra built from rectangles},
  \href{http://arxiv.org/abs/cs/0110059}{arXiv:cs/0110059}.
\newblock Electronic preprint, 2001.

\bibitem{DujEppSud-CGTA-07}
V.~Dujmovi{\'c}, D.~Eppstein, M.~Suderman, and D.~R. Wood.
\newblock {Drawings of planar graphs with few slopes and segments}.
\newblock  38(3):194--212, 2007,
  \href{http://dx.doi.org/10.1016/j.comgeo.2006.09.002}{doi:10.1016/j.comgeo.2%
006.09.002}, \href{http://arxiv.org/abs/math/0606450}{arXiv:math/0606450}.

\bibitem{DujSudWoo-CGTA-07}
V.~Dujmovi{\'c}, M.~Suderman, and D.~R. Wood.
\newblock {Graph drawings with few slopes}.
\newblock {\em Computational Geometry Theory \& Applications} 38(3):181--193,
  2007,
  \href{http://dx.doi.org/10.1016/j.comgeo.2006.08.002}{doi:10.1016/j.comgeo.2%
006.08.002}, \href{http://arxiv.org/abs/math/0606446}{arXiv:math/0606446}.

\bibitem{EarMar-AGT-79}
C.~F. Earl and L.~J. March.
\newblock {Architectural applications of graph theory}.
\newblock {\em Applications of Graph Theory}, pp.~327{--}355. Academic Press,
  1979.

\bibitem{Epp-EJC-05}
D.~Eppstein.
\newblock {The lattice dimension of a graph}.
\newblock {\em Eur. J. Combinatorics} 26(5):585{--}592, July 2005,
  \href{http://dx.doi.org/10.1016/j.ejc.2004.05.001}{doi:10.1016/j.ejc.2004.05%
.001}, \href{http://arxiv.org/abs/cs.DS/0402028}{arXiv:cs.DS/0402028}.

\bibitem{Epp-EJC-06}
D.~Eppstein.
\newblock {Cubic partial cubes from simplicial arrangements}.
\newblock {\em Electronic J. Combinatorics} 13(1, R79):1{--}14, September 2006,
  \href{http://arxiv.org/abs/math.CO/0510263}{arXiv:math.CO/0510263},
  \url{http://www.combinatorics.org/Volume_13/Abstracts/v13i1r79.html}.

\bibitem{Epp-GD-08-diamond}
D.~Eppstein.
\newblock {Isometric diamond subgraphs}.
\newblock {\em Proc. 16th Int. Symp. Graph Drawing (GD 2008)}, pp.~384{--}389.
  Springer-Verlag, Lecture Notes in Computer Science 5417, 2008,
  \href{http://dx.doi.org/10.1007/978-3-642-00219-9\_37}{doi:10.1007/978-3-642%
-00219-9\_37}, \href{http://arxiv.org/abs/0807.2218}{arXiv:0807.2218}.

\bibitem{Epp-GD-08}
D.~Eppstein.
\newblock {The topology of bendless three-dimensional orthogonal graph
  drawing}.
\newblock {\em Proc. 16th Int. Symp. Graph Drawing (GD 2008)}, pp.~78{--}89.
  Springer-Verlag, Lecture Notes in Computer Science 5417, 2008,
  \href{http://dx.doi.org/10.1007/978-3-642-00219-9\_9}{doi:10.1007/978-3-642-%
00219-9\_9}, \href{http://arxiv.org/abs/0709.4087}{arXiv:0709.4087}.

\bibitem{EppMum-WADS-09}
D.~Eppstein and E.~Mumford.
\newblock {Orientation-constrained rectangular layouts}.
\newblock {\em Proc. Algorithms and Data Structures Symposium (WADS 2009)},
  pp.~266--277. Springer-Verlag, Lecture Notes in Computer Science 5664, 2009,
  \href{http://arxiv.org/abs/0904.4312}{arXiv:0904.4312}.

\bibitem{EppMumSpe-SCG-09}
D.~Eppstein, E.~Mumford, B.~Speckmann, and K.~A.~B. Verbeek.
\newblock {Area-universal rectangular layouts}.
\newblock {\em Proc. 25th ACM Symp. Comp. Geom.}, pp.~267{--}276, 2009,
  \href{http://dx.doi.org/10.1145/1542362.1542411}{doi:10.1145/1542362.1542411%
}, \href{http://arxiv.org/abs/0901.3924}{arXiv:0901.3924}.

\bibitem{FelZic-DCG-08}
S.~Felsner and F.~Zickfeld.
\newblock {Schnyder woods and orthogonal surfaces}.
\newblock {\em Discrete and Computational Geometry} 40(1):103--126, 2008,
  \href{http://dx.doi.org/10.1007/s00454-007-9027-9}{doi:10.1007/s00454-007-90%
27-9}, \url{http://ftp.math.tu-berlin.de/~felsner/Paper/swaos.pdf}.

\bibitem{Fus-GD-05}
{\'E}.~Fusy.
\newblock {Transversal structures on triangulations, with application to
  straight-line drawing}.
\newblock {\em Proc. 13th Int. Symp. Graph Drawing (GD 2005)}, pp.~177{--}188.
  Springer-Verlag, Lecture Notes in Computer Science 3843, 2006,
  \href{http://dx.doi.org/10.1007/11618058\_17}{doi:10.1007/11618058\_17},
  \url{http://algo.inria.fr/fusy/Articles/FusyGraphDrawing.pdf}.

\bibitem{Fus-DM-08}
{\'E}.~Fusy.
\newblock {Transversal structures on triangulations: A combinatorial study and
  straight-line drawings}.
\newblock {\em Discrete Mathematics} 309(7):1870--1894, 2009,
  \href{http://dx.doi.org/10.1016/j.disc.2007.12.093}{doi:10.1016/j.disc.2007.%
12.093}, \href{http://arxiv.org/abs/math/0602163}{arXiv:math/0602163}.

\bibitem{Gru-03}
B.~Gr{\"u}nbaum.
\newblock {\em {Convex Polytopes}}.
\newblock Graduate Texts in Mathematics 221. Springer-Verlag, 2nd edition,
  2003.

\bibitem{GutMut-GD-01}
C.~Gutwenger and P.~Mutzel.
\newblock {A linear time implementation of SPQR-trees}.
\newblock {\em Proc. 8th Int. Symp. Graph Drawing (GD 2000)}, pp.~77--90.
  Springer-Verlag, Lecture Notes in Computer Science 1984, 2001.

\bibitem{Hea-QJPAM-98}
P.~J. Heawood.
\newblock {On the four-colour map theorem}.
\newblock {\em Quarterly J. Pure Appl. Math.} 29:270{--}285, 1898.

\bibitem{HodRivSmi-BAMS-92}
C.~D. Hodgson, I.~Rivin, and W.~D. Smith.
\newblock {A characterization of convex hyperbolic polyhedra and of convex
  polyhedra inscribed in the sphere}.
\newblock {\em Bull. Amer. Math. Soc.} 27:246--251, 1992,
  \href{http://dx.doi.org/10.1090/S0273-0979-1992-00303-8}{doi:10.1090/S0273-0%
979-1992-00303-8}.

\bibitem{HonNag-TR-08}
S.-H. Hong and H.~Nagamochi.
\newblock {Extending Steinitz' Theorem to Non-convex Polyhedra}.
\newblock Tech. Rep. 2008-012, Department of Applied Mathematics \& Physics,
  Kyoto University, 2008,
  \url{http://www.amp.i.kyoto-u.ac.jp/tecrep/abst/2008/2008-012.html}.

\bibitem{HopTar-SJC-73}
J.~Hopcroft and R.~Tarjan.
\newblock {Dividing a graph into triconnected components}.
\newblock {\em SIAM J. Comput.} 2(3):135--158, 1973,
  \href{http://dx.doi.org/10.1137/0202012}{doi:10.1137/0202012}.

\bibitem{HopTar-JACM-74}
J.~Hopcroft and R.~Tarjan.
\newblock {Efficient Planarity Testing}.
\newblock {\em J. ACM} 21(4):549--568, 1974,
  \href{http://dx.doi.org/10.1145/321850.321852}{doi:10.1145/321850.321852}.

\bibitem{Huf-MI6-71}
D.~A. Huffman.
\newblock {Impossible objects as nonsense sentences}.
\newblock {\em Machine Intelligence 6}, pp.~295--323. Edinburgh University
  Press, 1971.

\bibitem{Kan-WG-93}
G.~Kant.
\newblock {Hexagonal grid drawings}.
\newblock {\em Proc. Int. Worksh. Graph-Theoretic Concepts in Computer
  Science}, pp.~263--276. Springer-Verlag, Lecture Notes in Computer Science
  657, 1993,
  \href{http://dx.doi.org/10.1007/3-540-56402-0\_53}{doi:10.1007/3-540-56402-0%
\_53}.

\bibitem{KanHe-TCS-97}
G.~Kant and X.~He.
\newblock {Regular edge labeling of 4-connected plane graphs and its
  applications in graph drawing problems}.
\newblock {\em Theoretical Computer Science} 172(1{--}2):175{--}193, 1997,
  \href{http://dx.doi.org/10.1016/S0304-3975(95)00257-X}{doi:10.1016/S0304-397%
5(95)00257-X}.

\bibitem{KesPacPal-CGTA-08}
B.~Keszegh, J.~Pach, D.~P{\'a}lv{\"o}lgyi, and G.~T{\'o}th.
\newblock {Drawing cubic graphs with at most five slopes}.
\newblock {\em Computational Geometry Theory \& Applications} 40(2):138--147,
  2008,
  \href{http://dx.doi.org/10.1016/j.comgeo.2007.05.003}{doi:10.1016/j.comgeo.2%
007.05.003},
  \url{http://www.math.nyu.edu/~pach/publications/SlopeDegThree041607.pdf}.

\bibitem{KirPap-JCSS-88}
L.~M. Kirousis and C.~H. Papadimitriou.
\newblock {The complexity of recognizing polyhedral scenes}.
\newblock {\em Journal of Computer and System Sciences} 37(1):14--38, 1988,
  \href{http://dx.doi.org/10.1016/0022-0000(88)90043-8}{doi:10.1016/0022-0000(%
88)90043-8}, \url{http://lca.ceid.upatras.gr/~kirousis/publications/j29.pdf}.

\bibitem{KozKin-Nw-85}
K.~Ko{\'z}mi{\'n}ski and E.~Kinnen.
\newblock {Rectangular duals of planar graphs}.
\newblock {\em Networks} 5(2):145{--}157, 1985,
  \href{http://dx.doi.org/10.1002/net.3230150202}{doi:10.1002/net.3230150202}.

\bibitem{LoeMum-GD-08}
M.~L{\"o}ffler and E.~Mumford.
\newblock {Connected rectilinear graphs on point sets}.
\newblock {\em Proc. 16th Int. Symp. Graph Drawing (GD 2008)}, pp.~313--318.
  Springer-Verlag, Lecture Notes in Computer Science 5417, 2008,
  \href{http://dx.doi.org/10.1007/978-3-642-00219-9\_30}{doi:10.1007/978-3-642%
-00219-9\_30},
  \url{http://www.cs.uu.nl/research/techreps/repo/CS-2008/2008-028.pdf}.

\bibitem{Mac-DMJ-37}
S.~Mac~Lane.
\newblock {A structural characterization of planar combinatorial graphs}.
\newblock {\em Duke Math. J.} 3(3):460--472, 1937,
  \href{http://dx.doi.org/10.1215/S0012-7094-37-00336-3}{doi:10.1215/S0012-709%
4-37-00336-3}.

\bibitem{MukSze-CGTA-09}
P.~Mukkamala and M.~Szegedy.
\newblock {Geometric representation of cubic graphs with four directions}.
\newblock {\em Computational Geometry Theory \& Applications} 42(9):842--851,
  2009,
  \href{http://dx.doi.org/10.1016/j.comgeo.2009.01.005}{doi:10.1016/j.comgeo.2%
009.01.005},
  \url{http://www.cs.rutgers.edu/~szegedy/PUBLICATIONS/paper_padmini.pdf}.

\bibitem{PacPal-EJC-06}
J.~Pach and D.~P{\'a}lv{\"o}lgyi.
\newblock {Bounded degree graphs can have arbitrarily large slope numbers}.
\newblock {\em Electronic Journal of Combinatorics} 13(1), 2006,
  \url{http://www.combinatorics.org/Volume_13/PDF/v13i1n1.pdf}.

\bibitem{RahNisNaz-JGAA-03}
M.~S. Rahman, T.~Nishizeki, and M.~Naznin.
\newblock {Orthogonal drawings of plane graphs without bends}.
\newblock {\em J. Graph Algorithms \& Applications} 7(4):335--362, 2003,
  \url{http://jgaa.info/accepted/2003/Rahman+2003.7.4.pdf}.

\bibitem{Rai-GR-34}
E.~Raisz.
\newblock {The rectangular statistical cartogram}.
\newblock {\em Geographical Review} 24(2):292{--}296, 1934,
  \href{http://dx.doi.org/10.2307/208794}{doi:10.2307/208794}.

\bibitem{Rin-EPB-88}
I.~Rinsma.
\newblock {Rectangular and orthogonal floorplans with required rooms areas and
  tree adjacency}.
\newblock {\em Environment and Planning B: Planning and Design} 15:111{--}118,
  1988.

\bibitem{Riv-AM-96}
I.~Rivin.
\newblock {A characterization of ideal polyhedra in hyperbolic 3-space}.
\newblock {\em Annals of Mathematics} 143(1):51--70, 1996,
  \href{http://dx.doi.org/10.2307/2118652}{doi:10.2307/2118652}.

\bibitem{Sch-SODA-90}
W.~Schnyder.
\newblock {Embedding planar graphs on the grid}.
\newblock {\em Proc. 1st ACM-SIAM Symp. Discrete Algorithms}, pp.~138--148,
  1990.

\bibitem{Ste-EMW-22}
E.~Steinitz.
\newblock {Polyeder und Raumeinteilungen}.
\newblock {\em Encyclop{\"a}die der mathematischen Wissenschaften, Band 3
  (Geometries)}, pp.~1--139, 1922.

\bibitem{Tam-SJC-87}
R.~Tamassia.
\newblock {On embedding a graph in the grid with the minimum number of bends}.
\newblock {\em SIAM J. Comput.} 16(3):421--444, 1987,
  \href{http://dx.doi.org/10.1137/0216030}{doi:10.1137/0216030}.

\bibitem{Tar-SJC-72}
R.~Tarjan.
\newblock Depth first search and linear graph algorithms.
\newblock {\em SIAM J. Comput.} 1(2):146--160, 1972,
  \href{http://dx.doi.org/10.1137/0201010}{doi:10.1137/0201010}.

\bibitem{Thu-AMM-90}
W.~P. Thurston.
\newblock {Conway's tiling groups}.
\newblock {\em Amer. Mathematical Monthly} 97(8):757--773, 1990,
  \href{http://dx.doi.org/10.2307/2324578}{doi:10.2307/2324578}.

\bibitem{Wal-PCB-75}
D.~Waltz.
\newblock {Understanding line drawings of scenes with shadows}.
\newblock {\em The Psychology of Computer Vision}, pp.~19--91. McGraw-Hill,
  1975.

\bibitem{Woo-JGAA-03}
D.~R. Wood.
\newblock {Lower bounds for the number of bends in three-dimensional orthogonal
  graph drawings}.
\newblock {\em J. Graph Algorithms \& Applications} 7(1):33--77, 2003,
  \url{http://jgaa.info/accepted/2003/Wood2003.7.1.pdf}.

\bibitem{Woo-TCS-03}
D.~R. Wood.
\newblock {Optimal three-dimensional orthogonal graph drawing in the general
  position model}.
\newblock {\em Theoretical Computer Science} 299(1--3):151--178, 2003,
  \href{http://dx.doi.org/10.1016/S0304-3975(02)00044-0}{doi:10.1016/S0304-397%
5(02)00044-0}, \url{http://www.ms.unimelb.edu.au/~woodd/papers/Wood-TCS03.pdf}.

\bibitem{YeaSar-SJDM-95}
G.~K.~H. Yeap and M.~Sarrafzadeh.
\newblock {Sliceable floorplanning by graph dualization}.
\newblock {\em SIAM Journal of Discrete Mathematics} 8(2):258{--}280, 1995,
  \href{http://dx.doi.org/10.1137/S0895480191266700}{doi:10.1137/S089548019126%
6700}.

\bibitem{Zie-95}
G.~M. Ziegler.
\newblock {\em {Lectures on Polytopes}}.
\newblock Graduate Texts in Mathematics 152. Springer-Verlag, Berlin,
  Heidelberg, and New York, 1995.

\end{thebibliography}
\newpage

\section*{Appendix I: Eulerian triangulations}

Our proofs require several technical lemmas about the 3-connected cubic planar graphs and their duals, the Eulerian triangulations, that we collect here.

We begin with a standard property of bipartite planar graphs. As is standard when discussing planar graphs, we refer to a graph together with a fixed planar embedding of the graph as a \emph{plane graph}.

\begin{lemma}
\label{lem:even-faces}
A plane graph $\G$ is bipartite if and only if each face of the embedding of $\G$ has an even number of edges.
\end{lemma}

\begin{proof}
If $\G$ has an odd face, it is obviously not bipartite, as it contains an odd cycle formed from the face. Suppose on the other hand that all faces of $\G$ are even, and let $C$ be any cycle in $\G$. By the Jordan curve theorem, $C$ separates $\G$ into an interior and an exterior. The number of edges (modulo~2) in $C$ is the sum of the number of edges (modulo~2) of each of the faces of $\G$ in the interior of $C$, because each edge of $C$ belongs to one interior face and contributes one to this sum while each edge that is in the interior of $C$ belongs to two interior faces and is cancelled from the sum when the lengths of both faces are added modulo~2. But because each face in $\G$ is even, the sum of the interior faces to $C$ is even and therefore $C$ itself is even; as a graph in which all cycles are even, $\G$ is bipartite.
\end{proof}

An \emph{Eulerian graph} is a graph with an Euler tour; that is, a connected graph for which all vertex degrees are even. We define an \emph{Eulerian triangulation} to be a maximal planar graph that is Eulerian.

\begin{lemma}
\label{lem:dual-Eulerian}
A plane graph $\G$ is bipartite if and only if its dual graph is Eulerian. A plane graph $\G$ is 3-connected, 3-regular, and bipartite if and only if its dual graph is an Eulerian triangulation.
\end{lemma}

\begin{proof}
The duality of being bipartite and being Eulerian follows immediately from Lemma~\ref{lem:even-faces}.

Consider a plane graph $\G$ that is 3-connected, 3-regular, and bipartite. Let $\Delta$ be its dual graph. Then $\Delta$ is a simple graph rather than a multigraph, by the 3-connectedness of $\G$: a self-loop or a pair of multiple edges between two vertices in $\Delta$ would correspond to a single-vertex or two-vertex cut in $\G$. Every face of $\Delta$ is a triangle, from the 3-regularity of $\G$. Thus, $\Delta$ is a maximal planar graph. And every face of $\G$ has an even number of edges, from which it follows that $\Delta$ is Eulerian.

In the other direction, if $\Delta$ is an Eulerian triangulation, then its dual graph $\G$ is clearly bipartite and 3-regular. We are left to show that if $\Delta$ is an Eulerian triangulation then $\G$ is 3-connected Assume $\G$ is not 3-connected. Then there exists either a vertex-cut $u,v$ or a cut-vertex $w$ in $\G$. In the former case we have an edge $e_v$ of $v$ and an edge $e_u$ of $u$ that are adjacent to the same pair of faces of $\G$---see Fig.~\ref{fig:cut}, left. In the latter case $w$ has an edge that is adjacent to exactly one face---see Fig.~\ref{fig:cut}, right. Thus $\Delta$ contains either a 2-cycle or a loop, which contradicts the assumption that $\Delta$ is a triangulation.
\begin{figure}[h]
  \centering
%  \vspace{-1.25\baselineskip}
 \includegraphics{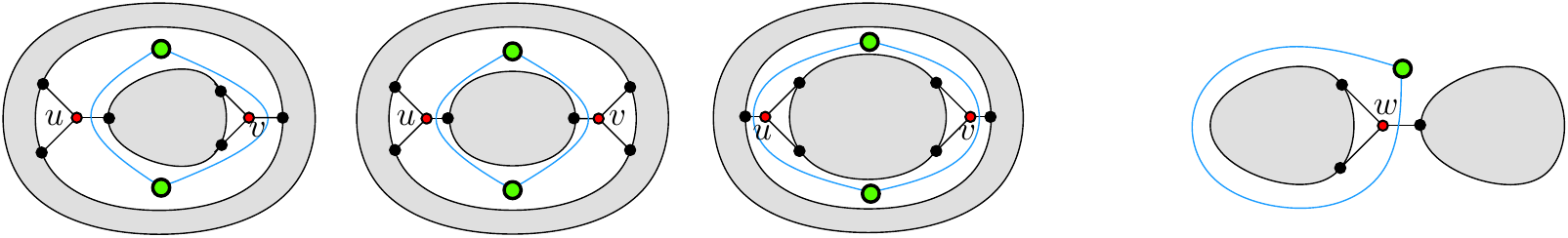}
  \caption{A cut of size 2 in $\G$ and the corresponding 2-cycle in $\Delta$; A cut-vertex in $\G$ and the corresponding loop in $\Delta$.}
  \label{fig:cut}
\end{figure}
\end{proof}

Let $\G$ be a 3-connected, 3-regular, bipartite plane graph. The bipartiteness of $\G$ allows us to two-color its vertices in such a way that adjacent vertices are assigned different colors. This naturally induces a two-coloring of faces of $\Delta$. In this paper we assume that the chosen colors are white and blue and the coloring is such that the outer face of $\Delta$ is white.

Furthermore, every 3-connected bipartite cubic planar graph $\G$ has a 3-face-coloring that is unique up to permutations of the colors~\cite{Hea-QJPAM-98}. If we assign each edge of the graph the color that is not used by its two adjacent faces, the result is a 3-edge-coloring of $\G$ in which the edges of each face alternate between two edge colors, and two adjacent faces have edges that alternate between two different pairs of colors~\cite{Epp-GD-08}.
We may carry the same color assignment over into the dual graph, $\Delta$; however, it will not form an edge coloring of $\Delta$, but rather will partition the edges of $\Delta$ into three color classes with the properties that every triangle of $\Delta$ has an edge of each of the three colors, each vertex of $\Delta$ has edges that alternate in cyclic order between two colors, and any two adjacent vertices in $\Delta$ have edges that alternate between different pairs of colors. We call such an edge partition of $\Delta$ a \emph{rainbow partition} to indicate the properties that it is not itself an edge coloring but that each triangle has a rainbow coloring in which all three colors are used. However, we also include the alternation of colors around each vertex of $\Delta$ as one of the defining properties of a rainbow partition.
\begin{lemma}
\label{lem:monochromatic-biconnectivity}
Let $\Delta_x$ be the subgraph of an Eulerian triangulation $\Delta$ formed by the edges of a single color $x$, $x \in \{ red, blue, green\}$ in a rainbow partition of $\Delta$. Then $\Delta_x$ is biconnected.
\end{lemma}
\begin{proof}
Assume without loss of generality that $\Delta_x$ is induced by red edges.

We start with showing that $\Delta_x$ is connected.
Consider any two vertices $u,v$ in $\Delta_x$. Let $P$ be a path between $u$ and $v$ in $\Delta$, such that at least one edge of that path is not red.
We will show that we can construct a path $P_{red}$ between $u$ and $v$ in $\Delta_x$.
If $u$ and $v$ are adjacent in $\Delta$, their shared edge must be red, because of the property that the edges incident to $u$ and to $v$ alternate between different pairs of colors both of which include red. Otherwise, we can assume that $P$ contains more than one edge.
Direct the edges of $P$ from $u$ to $v$. Let $v_1$ be the first vertex on $P$ such that the outgoing edge is not red (say, blue), and let $v_2$ be the vertex after $v_1$ on $P$.
$v_2$ does not have any red edges (otherwise a face formed by $v_1$, $v_2$ and their common neighbor would have 2 red edges). Hence (a) $v_2 \neq v$ (because $v_2$ is not in $\Delta_x$ and $v$ is), and (b) the neighbors of $v_2$ form a cycle of red edges. We use this red cycle to detour around $v_2$ until we hit the next vertex $v_3$ of $P$. More precisely, we replace the path $(v_1,v_2,v_3)$ in $P$ by a simple subpath of the red cycle connecting $v_1$ and $v_3$. We repeat the procedure until we arrive at $v$---see Fig.~\ref{fig:delta-connect} for an illustration.

\begin{figure}[h]
  \centering
 \includegraphics{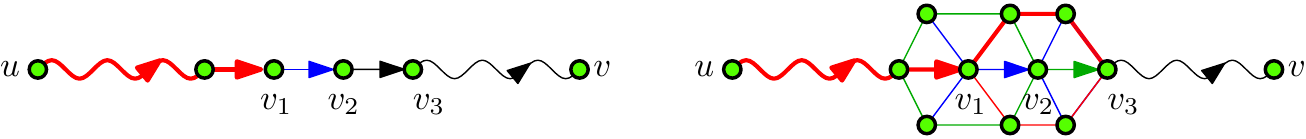}
  \caption{A path $P$ connecting $u$ and $v$ in $\Delta$ (left); a red detour around $v_2$ (the first vertex with the wrong colored entrance edge)(right).}
  \label{fig:delta-connect}
\end{figure}

Next we address 2-connectivity of $\Delta_x$. We must show that, for any three vertices  $u, v, w$ in $\Delta_x$ there is a path from $u$ to $v$ that avoids $w$. Since $\Delta_x$ is connected, there is a simple path $P$ in $\Delta_x$ connecting $u$ and $v$. Assume that $w$ is also in $P$, for otherwise we are done.

Let $w'$ and $w''$ be the neighbors of $w$ such that $(w',w)$ and $(w, w'')$ belong to $P$ (possibly $u = w'$ and/or $v = w''$). Let $w_2, ... w_{k-1}$ be the set of neighbors of $w$ that lie between $w_1 = w' $ and $w'' = w_k$ on the same side of $P$. Note that since the degree of each vertex in $\Delta$ is at least 4, then at least on one side of $P$ we have $k > 2$. Then all edges $(w,w_{2i})$, $2 \leq i < k/2$ have the same color and this color is not red. Let it be green. Then for each vertex $w_{2i}$ we have that its neighbors form a cycle of red edges in $\Delta$. We can detour the path $P$ starting at $w_1$ and traversing the neighbors of each of $w_{2i}$ on the side of the path $w_1,... w_k$ opposite to $w$.
\begin{figure}[h]
  \centering
 \includegraphics{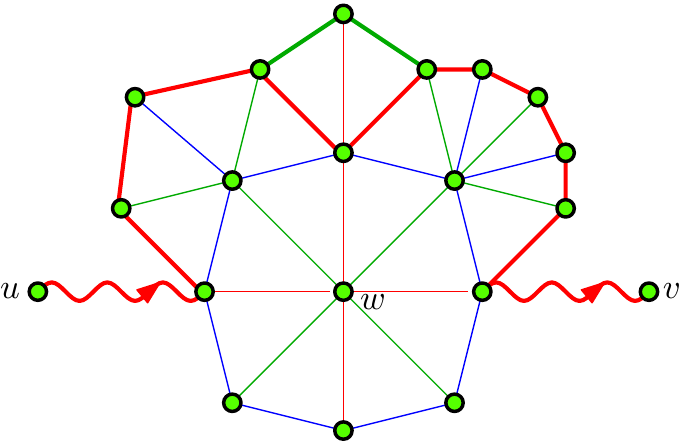}
  \caption{Detouring a red path around $w$.}
  \label{fig:delta-2-connect}
\end{figure}
Thus removing any vertex $w$ from $\Delta_x$ would not disconnect any pair $u,v$, which implies that $\Delta_x$ is 2-connected.
\end{proof}

Subsequently, we will refer to $\Delta_x$ as a \emph{monochromatic subgraph} of~$\Delta$.

\begin{lemma}
\label{lem:sep-4-cycles}
Suppose $\Delta$ is an Eulerian triangulation, with its edges colored in a rainbow partition, that contains a 4-cycle $C$. Then either all edges of $C$ have the same color or two colors, say red and blue, in the pattern $red-red-blue-blue$.
\label{lem:4-cycle}
\end{lemma}
\begin{proof}
We need to show that there exists no cycle with one of the following sequences of colors around the cycle: $red-blue-red-blue$, $red-red-blue-blue$, $red-blue-red-green$, $red-red-blue-green$.
A cycle with one of the first three patterns would imply a path $uvpq$ in $\Delta$ colored $red-blue-red$, violating the requirement that $v$ and $p$ alternate between different pairs of colors.
A cycle colored $red-red-blue-green$ implies a path $uvwpq$ colored $blue-red-red-green$. Then the edges at $v$ alternate between blue and red. Since $w$ has a red edge but must alternate between a different pair of colors, it must alternate between colors green and red. But $p$ also has red and green edges, violating the requirement that adjacent vertices $w$ and $p$ alternate between different pairs of colors.
\end{proof}

Our characterizations of corner polyhedra and the proof of our characterization of $xyz$ polyhedra both involve a careful study of the separating triangles in an Eulerian triangulation. We prove next some general facts that are needed in both cases.

\begin{lemma}
\label{lem:sep-tri-3x}
Let $\Delta$ be an Eulerian triangulation, with its edges colored in a rainbow partition, and let $\delta$ be a separating triangle in $\Delta$. Then $\delta$ also has one edge of each color.
\end{lemma}

\begin{proof}
Let the three colors be red, blue, and green; recall that, at each vertex of $\Delta$, the incident edges in cyclic order around that vertex must alternate between two colors, and that no two adjacent vertices have the same two colors. Suppose for a contradiction that two edges of $t$ are red, and that their shared endpoint alternates between red and green edges (the other five color combinations are symmetric). Then the other two vertices of $\delta$ alternate red and blue, so the green edges inside $\delta$ and the green edges outside $\delta$ are connected only through a single articulation vertex, contradicting the biconnectivity of the monochromatic subgraphs of $\Delta$ that we proved in Lemma~\ref{lem:monochromatic-biconnectivity}.
\end{proof}

\begin{lemma}
\label{lem:sep-tri-eulerian}
Let $\Delta$ be an Eulerian triangulation, and let $\delta$ be a separating triangle in $\Delta$. Then the two maximal planar subgraphs of $\Delta$ that are formed by splitting $\Delta$ along the edges of $\delta$ are each themselves Eulerian.
\end{lemma}

\begin{proof}
Within each subgraph, each vertex of $\delta$ must have even degree due to the alternation at that vertex of the edge colors of a rainbow partition of $\Delta$ and the fact (proved in Lemma~\ref{lem:sep-tri-3x}) that the two edges of $\delta$ at that vertex have different colors in the rainbow partition. The degree of the vertices that are not part of $\delta$ must also be even because it is unchanged from the degree in $\Delta$, and we assumed that $\Delta$ is Eulerian.
\end{proof}

\section*{Appendix II: From corner polyhedra to cycle covers}
We show that a corner polyhedron induces an interesting structure on its dual graph. But before we describe the structure lets look at some properties of the polyhedra that lead to it.

% Proof that every corner polyhedron corresponds to a rooted cycle cover goes here.

Let $\Po$ be a corner polyhedron and let $\Po^*$ be an isometric projection of its skeleton graph.

We define a normal to a face $f$ of $\Po$ to be a non-trivial vector  $\nu$ that is perpendicular to $f$ and is directed towards the exterior of the polyhedron. We say that a face $f$ of $\Po$ is oriented towards vector $e$ if we have $(\nu, e) > 0$. We call  $f$ a \emph{forward face} if it is oriented towards the vector $(1,1,1)$  (i.e. faces whose normal has non-negative coordinates), and we call it a \emph{back face} otherwise. These notions of forward and back faces naturally carry over to the faces of the isometric projection $\Po^*$ of $\Po$.
\begin{lemma} Let $f$ be an arbitrary face of $\Po^*$. Then for any internal angle $\alpha$ of $f$ $\alpha \in \{\pi/3, 2\pi/3, 4\pi/3\}$
\label{lem:no5piover3}
\end{lemma}
\begin{proof} Since $\Po^*$ is an isometric projection of a skeleton of an orthogonal polytope, every edge of $\Po^*$ is parallel to one of the three directions that pairwise form $\pi/3$ angle. In particular, the edges of the face $f$ have two of the three slopes, and adjacent edges have distinct slopes. Hence all we need to show to prove the lemma is that $f$ does not have $5\pi/3$ as an inner angle.

Assume for a contradiction that there exists a pair  $e_1, e_2$ of consecutive edges along $f$ forming an interior angle $5\pi/2$. We direct $e_1$ and $e_2$ away from their common vertex $v$. Since $\angle (e_1, e_2) =5\pi/3$ one of the edges directed positively (w.r.t. the corresponding axis) and the other is directed negatively. Let $e_3$ be the third edge of $\Po^*$ adjacent to $v$ directed away from $v$---see Fig.~\ref{fig:5piover3}
\begin{figure}[h]
  \centering
%  \vspace{-1.25\baselineskip}
 \includegraphics{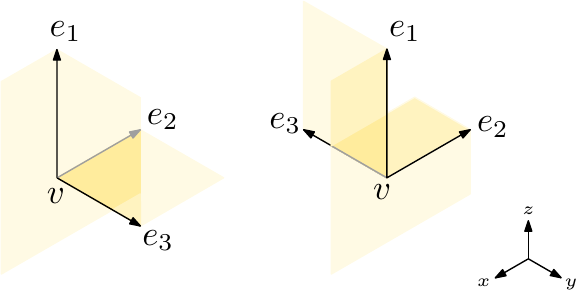}
  \caption{A face with $5\pi/3$ interior angle.}
  \label{fig:5piover3}
\end{figure}
If $e_3$ is the normal of $f$---see Fig.~\ref{fig:5piover3} (left)---then the normal for face $f_{13}$ spanned by $e_1$ and $e_3$ is $-e_2$ and the normal for $f_{23}$ is $-e_1$. Thus one of these normals is directed negatively, which contradicts the definition of a corner polyhedron. Otherwise if $-e_3$ is the normal for $f$---see Fig.~\ref{fig:5piover3}(right)--- the normals for $f_{13}$ and $f_{23}$ are $e_2$ and $e_1$ and we arrive at the same contradiction.
\end{proof}

\begin{lemma}
\label{lem:2piover3}
Each face of $\Po^*$ has the shape of a double staircase: there are two vertices at which the interior angle is $\pi/3$, and the two sequences of interior angles on the paths between these vertices alternate between interior angles of $2\pi/3$ and $4\pi/3$.
\end{lemma}
\begin{proof}
First assume for a contradiction that there exists a face $f$ in $\Po^*$ that has no interior angle of $\pi/3$. Since every edge of $f$ is parallel to the edge after the edge it is adjacent to in cyclic order around the face, the interior angles alternate values $2\pi/3$ and $4\pi/3$ around the face. Which makes $f$ an open polyline instead of a simple polygon.

Next, note that any face $f$ of $\Po^*$ has at least two interior angles of $\pi/3$ since all angles of $f$ need to sum up to $(k-2)\pi$, where $k$ is the number of vertices of $f$.
\begin{figure}[h]
  \centering
%  \vspace{-1.25\baselineskip}
 \includegraphics{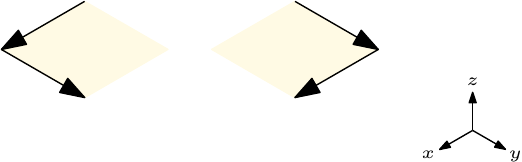}
  \caption{Two differently oriented interior $\pi/3$ angles.}
  \label{fig:2-int-piover3}
\end{figure}

Let $v_1$ and $v_2$ be two vertices of $f$ with interior angles $\pi/3$. We say that these interior angles are oriented the same way, if when we orient edges forming the angles positively (according to positive direction of the corresponding axis), then for both vertices the interior of the face is to the right (to the left)  side of the obtained directed path for both vertices.

Assume for a contradiction that there exists a third vertex $v_3$ with interior angle $\pi/3$. Then at least two of these three angles are oriented the same way, which means that one of the paths connecting these two vertices contains an interior angle of $5\pi/3$. We arrive at a contradiction with Lemma~\ref{lem:no5piover3}.
\end{proof}

\begin{lemma}
\label{lem:angle-sums}
At each vertex of $\Po^*$, there are either three angles of $2\pi/3$, or there is one angle of $4\pi/3$ and two angles of $\pi/3$.
\end{lemma}

\begin{proof}
By Lemma~\ref{lem:2piover3}, the only allowable angles are $\pi/3$, $2\pi/3$, and $4\pi/3$. These are the only ways for three angles with these values to add up to a total angle of $2\pi$.
\end{proof}

\begin{lemma}
\label{lem:color-from-angles}
In the vertex two-coloring of $\Po^*$, a vertex $v$ has the same color as the hidden vertex if and only if it has three angles of $2\pi/3$ incident to it, and $v$ has the opposite color from the hidden vertex if and only if it has one angle of $4\pi/3$ and two angles of $\pi/3$ incident to it.
\end{lemma}

\begin{proof}
Suppose that $uv$ is an edge in $\Po^*$, and let $f$ be a face containing edge $uv$. If $u$ has an angle of $2\pi/3$ incident to it in $f$, let $f$ be one of the three  then by Lemma~\ref{lem:2piover3} it must be interior to one of the two chains of alternating $2\pi/3$ and $4\pi/3$ angles in $f$, so the adjacent vertex $v$ must either have a $4\pi/3$ angle in the same chain, or it must be one of the two vertices with angles of $\pi/3$ that ends the chain. If $u$ has an angle of $\pi/3$ in $f$, it must be one of the two vertices with that angle in $f$, and its adjacent vertex $v$ is one of the vertices with angles of $2\pi/3$ in the chains connecting these two vertices. And if $u$ has an angle of $4\pi/3$ in $f$, it must be an intermediate vertex in one of the two chains forming $f$ and its neighbor $v$ must have an angle of $2\pi/3$. Thus in all cases the type of vertex $u$ is opposite that of vertex $v$, so the partition of vertices according to the angles described by Lemma~\ref{lem:angle-sums} must coincide with the two-coloring of~$\Po^*$.

It remains to show that the vertices with the angles of $4\pi/3$ are the ones with the opposite color from the hidden vertex. But this is clearly true for the three neighbors of the hidden vertex, and because there is a unique two-coloring of $\Po^*$ the result follows for all its remaining vertices.
\end{proof}

Now, construct a graph $\C^*$ that has as its the vertices of $\Po^*$ that are adjacent to a $\pi/3$ angle. Let the edges of $\C^*$ connect pairs of these vertices belonging to the same face.

\begin{lemma}
$\C^*$ is a collection of vertex disjoint simple cycles.
\end{lemma}

\begin{proof}
The neighbors of a vertex $v$ in $\C^*$ come from the angles of $\pi/3$ incident to $v$ in $\Po^*$. By lemma~\ref{lem:angle-sums}, if there are any such angles there are exactly two of them.
\end{proof}

What we actually interested in is not the cycle set $\C^*$ itself but a structure dual to it, which we call a \emph{rooted cycle cover} and denote $\C$.

We construct $\C$ as follows. We subdivide every edge of $\C^*$, replacing it by a two-edge path whose middle vertex indicates the face whose $\pi/3$ angles the edge connects. Next we turn every two-edge-path between two new middle vertices by a single edge. By construction we obtain a collection of cycles $\C$ isomorphic to $\C^*$---see Fig.~\ref{fig:corners}, right.

\begin{lemma}
Let $\Delta$ be the Eulerian triangulation dual to corner polyhedron $\Po$. Then the vertex set of $\C$ contains every interior vertex of $\Delta$, and the edge set of $\C$ is a subset of edges of~$\Delta$.
\end{lemma}

\begin{proof}
The interior vertices of $\Delta$ correspond to the forward faces of $\Po^*$.
By Lemma~\ref{lem:2piover3} every  forward face of $\Po^*$ has two $\pi/3$ interior angles, hence by construction $\C$ contains every interior vertex of $\G$.
By construction two vertices $v_1$ and $v_2$ of $\C$ share an edge if and only if the corresponding faces $f_1$ and $f_2$ of $\Po^*$ have a common vertex, which in its turn (since $\Po^*$ is a cubic graph) means that $f_1$ and $f_2$ share an edge in $\Po^*$ or in other words $v_1$ and $v_2$ are adjacent in $\C$.
\end{proof}

$\Po^*$ is a bipartite graph, and its vertex set has a unique (up to permutation of colors) 2-coloring such that any two vertices sharing an edge are colored differently. This induces a face 2-coloring of the dual graph of $\Po^*$, an Eulerian triangulation $\Delta$. We adopt the convention that the color of the outer triangle is white and that the other triangle color is blue, as shown in our figures. Let $uvw$ be the outer face of $\Delta$.

\begin{lemma}
\label{lem:white-rabbit}
Every inner white triangle contains exactly one edge of $\C$.
\end{lemma}

\begin{proof}
By Lemma~\ref{lem:color-from-angles} an inner white triangle is dual to a vertex of $\Po^*$ where there is one $4\pi/3$ angle and two $\pi/3$ angles. This triangle therefore contains exactly one edge, the edge connecting the two faces with the $\pi/3$ angles.
\end{proof}

We define a \emph{rooted cycle cover} more generally to be a structure with this form: a collection of cycles covering all of the interior vertices of $\Delta$ that has exactly one edge in every inner white triangle of $\Delta$. The results of this section can be summarized as showing that any corner polyhedron gives rise to a rooted cycle cover on the dual graph, with the root triangle dual to the hidden vertex. As we will show in the subsequent sections, this combinatorial abstraction of the geometry of a corner polyhedron provides enough information to reconstruct another corner polyhedron for the same graph.

\section*{Appendix III: Regular edge labeling}

A simple orthogonal polyhedron $\Po$ induces a rainbow partition for its dual graph $\Delta$, where each edge of $\Delta$ gets its color based on the orientation of its dual axis parallel edge. Although $\Delta$ is an undirected graph, we can use the left-to-right and bottom-to-top orders of faces of $\Po$ to define a direction for each edge of $\Delta$. More precisely, we do the following. We orient every edge of $\Po$ positively (i.e. such that its only non-trivial coordinate is positive) and then orient the corresponding dual edge in $\Delta$ such that it crosses its primal edge from left to right.

When $\Po$ is a corner polyhedron, this new labeling of $\Delta$ by directions and colors combines the properties of the rainbow partition (the edges of every triangle of $\Delta$ all have different colors, edges around each vertex of $\Delta$ alternate between two colors in cyclic order, and any two adjacent vertices in $\Delta$ have edges that alternate between different pairs of colors) with a similar alternation property for the directions of its edges:

\begin{lemma}
\label{lem:orient-alternate}
At each interior vertex $v$ of $\Delta$, all but two of the triangles incident to $v$ have one incoming and one outgoing edge; the two exceptional triangles are both white. The orientations of the edges at each exterior vertex alternate between incoming and outgoing.
\end{lemma}
\begin{proof}
According to Lemma~\ref{lem:2piover3} each face of the isometric projection of $\Po$ has two vertices with interior angles  of $\pi/3$ connected by two paths each alternating its interior angles of $2\pi/3$ and $4\pi/3$. In three dimensions, this means that each face $f$ of $\Po$ has two vertices connected by two paths of axis-aligned edges that are monotone in each coordinate direction. Now let $v$ be the vertex  of $\Delta$ dual to $f$. Any two dual edges incident to~$v$ and consecutive in the cyclic ordering of edges incident to~$v$ have different orientations with respect to $v$ if they are dual to edges in the same monotone path around $f$ and have the same orientation if they belong to different monotone paths---see Fig.~\ref{fig:xy2fat} for an illustration.

Next we show that the two exceptional triangles adjacent to a vertex $v$ are both white in the two-coloring of $\Delta$. These two triangles correspond to vertices adjacent to $\pi/3$ angles, so this result follows immediately from Lemma~\ref{lem:color-from-angles}.

Finally, label  the three external vertices of $\Delta$ $x$, $y$ and $z$ depending on which direction the corresponding back face of $\Po$ is perpendicular to. Consider the vertex $x$. The edges adjacent to $x$ correspond to the edges forming a $x$-monotone path $p_x$ around face $f_x$. The edges of $p_x$ are oriented positively, hence all edges of one color adjacent to $x$ are oriented the same way, and edges of different colors have different orientations. Thus edges alternate orientations around $x$. The same holds for $y$ and $z$.
\end{proof}
%
%
%(Compare briefly to the other notions of regular edge labelings we already know about.)
In analogy with regular edge labelings of dual graphs of rectangular layouts in two dimensions~\cite{KanHe-TCS-97} that originate in a very similar manner we are going to call the structure with properties described above a \emph{regular edge labeling} of Eulerian triangulation $\Delta$. That is, a regular edge labeling is an assignment of directions and colors to the edges of $\Delta$ so that the colors form a rainbow partition, each exterior vertex of $\Delta$ has edges with alternating directions, and each interior vertex $v$ of $\Delta$ has edges that alternate directions except within two white triangles, where the directions of the edges at $v$ do not alternate.

Just as in the two-dimensional case, we will eventually show that the correspondence between a polyhedron and a regular edge labeling of its dual graph works both ways---that is if we can construct a regular edge labeling for an Eulerian triangulation $\Delta$ we can represent its dual as a simple orthogonal polyhedron and more specifically as a corner polyhedron. The following lemmas are necessary for demonstrating this correspondence.

\begin{lemma}
In a regular edge labeling, for every interior vertex $v$ one of the two exceptional white triangles has two incoming edges, and one of them has two outgoing edges.
\end{lemma}

\begin{proof}
If $X$ and $Y$ are the two exceptional white triangles, then the sequence of triangles between them (say, clockwise from $X$ to $Y$) strictly alternates between blue and white triangles. If $X$ has two incoming edges, then the blue triangles in this sequence all have their first edge in this clockwise order as the incoming one and their second edge in the clockwise order as the outgoing one, and the white triangles vice versa. So when we get to $Y$, the first edge in the clockwise order will be outgoing, and $Y$ will have two outgoing edges. Similarly, if we start with a triangle $Y$ that has two outgoing edges, the strict alternation of the triangles in clockwise order from $Y$ to $X$ means that when we get back to $X$ it will be forced to have two incoming edges. It's not possible for the two exceptional triangles both to be incoming, nor for both of them to be outgoing.
\end{proof}

\begin{lemma}
\label{lem:blue-cycle}
In a regular edge labeling, the edges of each blue triangle are oriented in a directed cycle.
\end{lemma}

\begin{proof}
This follows immediately from the fact that the orientations of the edges at each vertex of the triangle must be opposite to each other.
\end{proof}

\begin{lemma}
Let $\C$ be a rooted cycle cover of an Eulerian triangulation $\Delta$. Then $\Delta$ admits a regular edge labeling.
\end{lemma}

\begin{proof}
We have already shown in Appendix I that $\Delta$ admits a rainbow partition, so we are only left to show that we can orient the edges of $\Delta$ such that labeling has the property formulated in
Lemma~\ref{lem:orient-alternate}.

\begin{figure}[t]
\centering\includegraphics[width=4in]{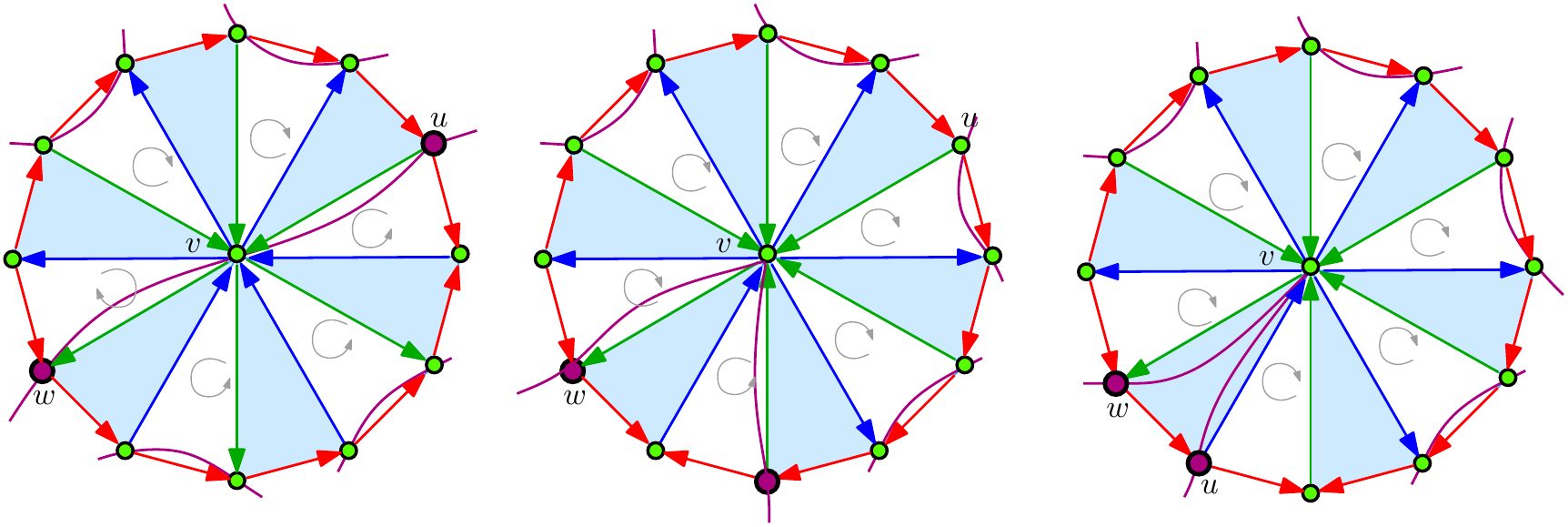}
\caption{Edge coloring of a vertex neighborhood.}
\label{fig:v-neighb}
\end{figure}

Recall that we have a two coloring of faces of $\Delta$ in which the outer face is colored white, and that every white triangle other than the outer face has exactly one edge of $\C$.
We orient the edges as follows:

\begin{enumerate}
  \item In each white triangle we direct the edge that is part of $\C$ clockwise and we direct the other two edges counterclockwise.
  \item We reverse the orientations of the edges of each white triangle that is inside an odd number of cycles of~$\C$.
\end{enumerate}

As we now show, this choice of directions forms a regular edge labeling.
To see this, consider a white triangle adjacent to an internal vertex $v$. There is a cycle $C_v$ of the cycle cover that contains $v$. Let $v$ be adjacent to blue and green edges and let $w$ and $u$ be the neighbors of $v$ that belong to $C_v$---see Fig.~\ref{fig:v-neighb}.

By construction every white triangle except for the ones to which $wv$ and $wv$ belong has an incoming and outgoing edge at $v$. Furthermore, the white triangles that are on the same side of $C_v$ as each other have all edges of the same color at $v$ oriented the same way.

Consider the white triangle $\delta_w$ containing $w$. $wv$ is a cycle cover edge, hence it is oriented oppositely around the face to the other edge of $\delta$ adjacent to $v$, hence these edges are oriented the same way. The same holds for the white triangle $\delta_u$ containing $uv$. Note that $\delta_w \neq \delta_u$ since every inner white triangle contains exactly one edge of the cycle cover.

Thus every triangle around $v$ except for the two white ones containing $w$ and $u$ have one incoming and one outgoing edge at $v$, so we can conclude that the constructed labeling of $\delta$ is regular edge labeling.
\end{proof}

An \emph{$st$-planar graph} is a planarly embedded directed acyclic graph
with a single source and sink, both on its outer face. The two-dimensional regular edge labelings corresponding to rectangular layouts have the property that each monochromatic subgraph of the labeling is $st$-planar. As we show, this same property holds for our three-dimensional regular edge labelings.

\begin{figure}[t]
\centering\includegraphics[width=4in]{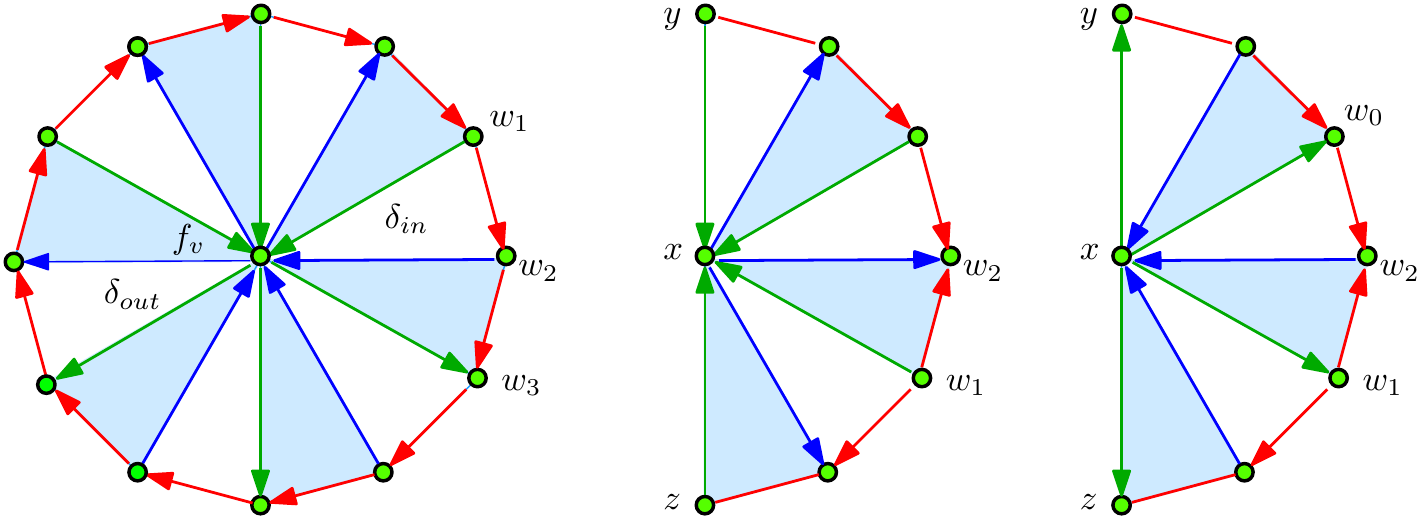}
\caption{A face cycle in a regular edge labeling.}
\label{fig:face-cycle}
\end{figure}

\begin{lemma}
\label{lem:mono-st-graphs}
Let $\Delta$ be oriented and colored to form a regular edge labeling. Then
each monochromatic subgraph $\Delta_x$ of $\Delta$ is an $st$-planar graph.
The source and sink of $\Delta_x$ both belong to the outer triangle of $\Delta$, and each
vertex of the outer triangle serves as a source on exactly one of the
monochromatic subgraphs $\Delta_x$, $\Delta_y$, and $\Delta_z$ and as a sink on exactly one other
monochromatic subgraph.
\end{lemma}
\begin{proof}
By definition of regular edge labeling each inner vertex of $\Delta$ has incoming and outgoing edges of both colors---it has an exceptional triangle with incoming edges, which are of distinct colors and an exceptional triangle with outgoing edges, which are of distinct colors. Hence each inner vertex of $\Delta_x$ has both incoming and out going edges.

Next, we show that no face of $\Delta_x$ forms a directed cycle.

Assume for a contradiction that there exists a face $f$ of $\Delta_x$ that violates this condition: the edges around $f$ form a directed cycle, without loss of generality a clockwise cycle. Face $f$ corresponds to a vertex $v_f$ of $\Delta$ such that the vertices of $f$ are the cycle of neighbors of $v_f$ in $\Delta$---see Fig.~\ref{fig:face-cycle} (left) for an illustration. Then, because $v_f$ is an interior vertex of $\Delta$, it has an incident white triangle $\delta_{in}$ with two incoming edges and another incident white triangle $\delta_{out}$ with two outgoing edges. Consider the directed edge $w_1w_2$ of $\delta_{in}$ on $f$ and let $w_3$ be the next vertex of $f$ after $w_3$ in clockwise order. Edges $w_1v$ and $w_2v$ are incoming at $v$, edge $w_3v$ has $v$ as a source by definition of regular edge labeling, and edge $w_2w_3$ is directed away from $w_2$. Thus vertex $w_2$ is adjacent to a blue triangle $vw_2w_3$ that has two outgoing edges at $w_2$, which contradicts the definition of regular edge labeling, hence the edges of $f$ cannot form a clockwise cycle.

Finally, we will show that outer face of $\Delta_x$ consists of an edge connecting the two outer vertices and a directed path.

The outer vertices of $\Delta$ strictly alternate the directions of edges around it hence each of the two outer vertices $y$ and $z$ of $\Delta_x$ have all edges directed the same way and since there is an edge between them, one of them is a source and the other is a sink.

Consider the outer path $p_{yz}$ connecting $y$ and $z$. It consists of the path of neighbors of the third outer vertex $x$ of $\Delta$. Since the edges of $x$ alternate directions around $x$, all edges of $p_{yz}$ that belong to white triangles around $x$ are oriented the same way w.r.t. $x$. Assume there is a blue triangle with an edge $w_1w_2$ of $p_{yz}$ oriented in the opposite direction. If $xw_1$ had $x$ as a source, then since edges alternate directions around $x$ we have a blue triangle $xw_1w_2$ with two incoming edges at $w_1$---see Fig.~\ref{fig:face-cycle}, middle. Otherwise consider the vertex $w_0$ which a neighbor of $w_1$ on $f$ distinct from $w_2$. For the same reason we have a blue triangle $w_0fw_1$ with two incoming edges at $w_0$---see Fig.~\ref{fig:face-cycle}, right. Both cases contradict the definition of regular edge labeling. Hence we can conclude that $p_{yz}$ is a directed path.

Thus $\Delta_x$ satisfies the conditions of Lemma~\ref{lem:stplanar} below and hence is $st$-planar.
\end{proof}

\begin{lemma}
\label{lem:stplanar}
Let $\G$ be a planar graph, oriented so that the outer face cycle
forms two directed paths, no face forms a directed cycle, and each
vertex that is not on the outer face has both incoming and outgoing
edges. Then $\G$ is an $st$-planar graph.
\end{lemma}

\begin{proof}
Suppose for a contradiction that $\G$ is not acyclic, and let $C$ be a
directed cycle in $G$ that encloses as few faces as possible. $C$
cannot enclose only a single face, because no face forms a directed
cycle, so there must be an edge $e$ of $\G$ that is incident to and
within $C$. If $e$ is directed away from $C$, then we can follow a
directed path in $\G$ until reaching either a repeated vertex or
another vertex of $C$, because each vertex within $C$ has outgoing
edges. If we reach a repeated vertex, the part of the path from this
vertex to itself forms a cycle enclosing fewer faces, and if we reach
a vertex of $C$, then one of the two parts into which this path splits
$C$ forms a cycle enclosing fewer faces. In the case that $e$ is
directed towards $C$, then following a path backwards from $e$ again
leads to a cycle enclosing fewer faces. This contradiction shows that
$\G$ must be acyclic, and its only sources and sinks must be the two
vertices on the outer face that are not interior to the directed paths
forming this face.
\end{proof}

\section*{Appendix IV: From regular edge labeling to corner polyhedra}

In this section, we show that, for every 3-connected planar bipartite cubic graph $\G$ with dual Eulerian triangulation $\Delta$, if $\Delta$ has a regular edge labeling then $\G$ represents a corner polyhedron.

Recall that a regular edge labeling consists of an $st$-planar orientation of each monochromatic subgraph $\Delta_x$ of $\Delta$, together with some consistency conditions on how the orientations of the edges of two colors can meet at a vertex of~$\Delta$. From the regular edge labeling, for any pair of colors $x$ and $y$, we now define another directed graph $\Delta_{xy}$, as follows. The vertices of $\Delta_{xy}$ are the vertices of $\Delta$, together with one new sink vertex. The edge set of $\Delta_{xy}$ is the union of the edges of $\Delta_{x}$ and of $\Delta_{y}$, together with two new edges connecting the monochromatic vertices of the root triangle of $\Delta$ with the sink vertex. Without loss of generality we may assume that, at the vertex $v$ of the root triangle that has edges of both colors $x$ and $y$, the edges of color $x$ go outwards from $v$ and the edges of color $y$ go inwards to $v$; swap $x$ and $y$ if necessary to enforce this condition. Then in $\Delta_{xy}$, the edges of color $x$ are given the same orientation as they have in the regular edge labeling, while the edges of color $y$ are given the opposite orientation from the one they have in the regular edge labeling. The two edges connecting to the sink vertex are oriented towards it. An example is depicted in Fig.~\ref{fig:deltaxy}. We may embed $\Delta_{xy}$, as shown in the figure, so that $v$ and the sink are both on the outer face of the embedding.

\begin{figure}[t]
\centering\includegraphics[width=6in]{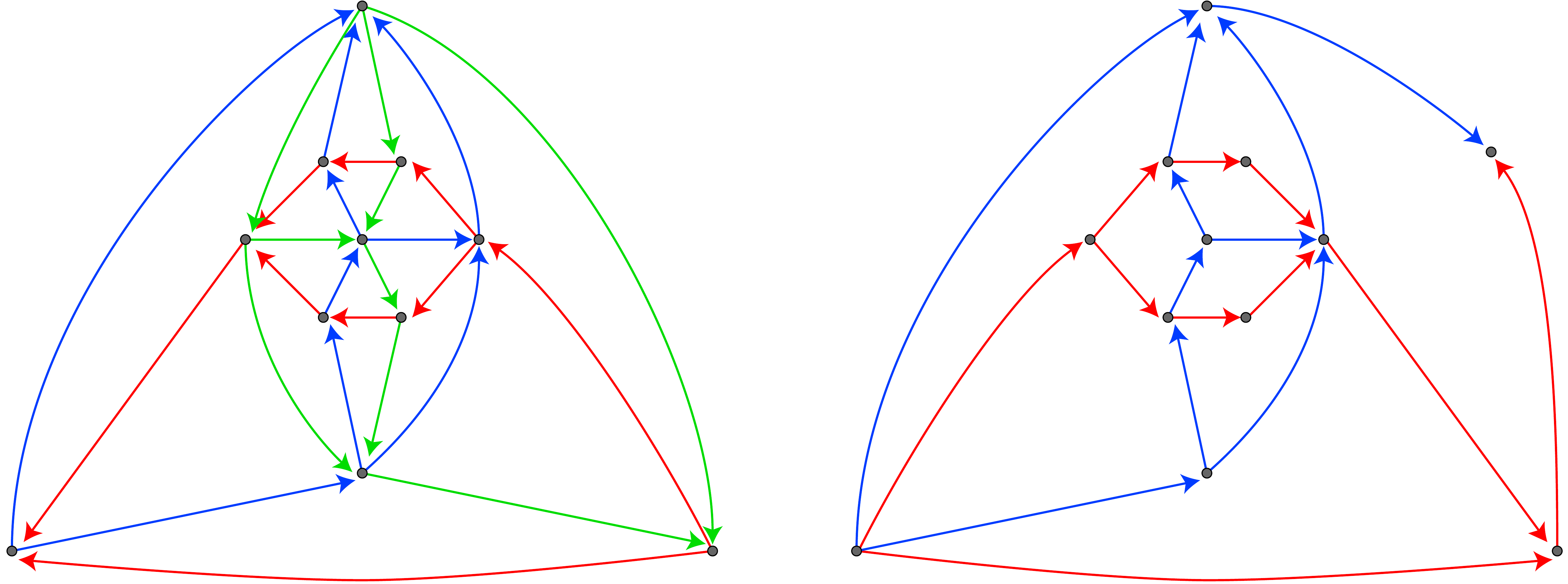}
\caption{Constructing the $st$-planar graph $\Delta_{xy}$ (right) from a regular edge labeling (left).}
\label{fig:deltaxy}
\end{figure}

\begin{lemma}
$\Delta_{xy}$ is $st$-planar.
\end{lemma}

\begin{proof}
In $\Delta_{xy}$, each vertex other than $v$ and the sink has both incoming and outgoing edges because of the $st$-planarity of the graphs $\Delta_x$ and $\Delta_y$ from which $\Delta_{xy}$ is formed. Additionally, each face of $\Delta_{xy}$ is a quadrilateral, with two consecutive edges of color $x$ and two consecutive edges of color $y$: the faces not incident to the sink are formed by removing the edge of color $z$ between two triangles of $\Delta$, while the two faces incident to the sink are also clearly quadrilaterals of this form.

The outer quadrilateral face is oriented acyclically, with one source and one sink.
Each inner face, also, is oriented acyclically: The property of alternating orientations around each vertex in a regular edge labeling, except within two exceptional white triangles, means that the blue triangle within any quadrilateral face must be consistently oriented either clockwise or counterclockwise in $\Delta$ (Lemma~\ref{lem:blue-cycle}), so when we reverse the orientation of one of its edges in $\Delta_{xy}$ it loses its consistent orientation and eliminates the possibility that the whole quadrilateral is cyclically oriented.

As a planar graph with one source and one sink on the outer face, with every other vertex having both incoming and outgoing edges, and with
all faces oriented acyclically, $\Delta_{xy}$ must be $st$-planar by Lemma~\ref{lem:stplanar}.
\end{proof}

We use the fact that this graph $\Delta_{xy}$ is $st$-planar to find an \emph{$st$-numbering} of its vertices. An $st$-numbering is an assignment of numbers to the vertices of an $st$-oriented graph in such a way that, for each edge $uv$, the number for $u$ is smaller than the number for $v$; as a consequence, for each vertex other than the source and sink, there will be a neighboring vertex with a smaller number and another neighboring vertex with a larger number. An $st$-numbering is easy to compute from the $st$-orientation of the graph, in more than one way. If we require additionally that each vertex get a distinct number (which will correspond to the geometric property that no two faces be coplanar) then we may simply assign each vertex its position in a breadth-first traversal of the $st$-oriented graph $\Delta_{xy}$. If, on the other hand, we wish the numbers to be drawn from as small a set as possible (corresponding to finding a representation as an orthogonal polyhedron with coordinates within a small grid) then we may assign each vertex its distance from the source vertex, as determined again by a breadth-first traversal of the graph. We may add the same constant to each number in the numbering, if necessary, to ensure that the number of the source vertex is zero; for the numberings produced by breadth-first traversal or breadth-first layers, the number of the source will automatically be zero.

We use the numbers produced by this numbering as coordinates of our desired orthogonal polyhedron. Specifically, within each graph $\Delta_{xy}$ there are monochromatic vertices (belonging to one but not the other of $\Delta_x$ or $\Delta_y$) and bichromatic vertices (belonging to both monochromatic subgraphs). The bichromatic vertices of $\Delta_{xy}$ correspond to a family of faces that should lie on planes parallel to the $x$ and $y$ axes in a geometric representation of the given graph $\G$; we use the $st$-numbering as the $z$-coordinate of this plane. Each vertex $v$ of the given cubic bipartite graph $\G$ corresponds in the dual Eulerian triangulation to a triangle of vertices that are bichromatic in $\Delta_{yz}$, $\Delta_{xz}$, and $\Delta_{xy}$, and we use the $st$-numberings of these three vertices as respectively the $x$, $y$, and $z$-coordinates of~$v$. In the rest of the section, we argue that the resulting placement of vertices produces a representation of $\G$ as a corner polyhedron.

\begin{lemma}
\label{lem:placement-properties}
The vertex placement described above, for a graph $\G$ with a regular edge labeling of its dual Eulerian triangulation $\Delta$, has the following properties:
\begin{enumerate}
\item\label{lpp:positive} Each vertex of $\G$ lies in or on the positive orthant.
\item\label{lpp:boundary} The vertex $h$ dual to the root triangle lies at the origin, the neighbors of $h$ lie on the coordinate axes, the other vertices on faces incident to $h$ lie on the coordinate planes, and all remaining vertices are strictly interior to the positive orthant.
\item\label{lpp:edge} The two endpoints of every edge of $\G$ lie on an axis-parallel line.
\item\label{lpp:face} The set of vertices of any face of $\G$ lie on an axis-parallel plane.
\end{enumerate}
\end{lemma}

\begin{proof}
Properties \ref{lpp:positive} and~\ref{lpp:boundary} follow immediately from the facts that the $st$-numbering of each graph $\Delta_{xy}$ is non-negative, and is zero only at the source vertex. Therefore each vertex $v$ of $\G$ has non-negative coordinates, with one of these coordinates zero only when $v$ is adjacent to a face that is dual to a vertex in $\Delta$ that is the source in one of these three graphs. The root triangle is dual to the vertex that  is adjacent to all three of these source faces, so this vertex lies at the origin. Its neighbors are adjacent to two source faces, so two of their coordinates are zero and the third nonzero; thus, they lie on the coordinate axes. The remaining vertices on the source faces have the corresponding coordinate zero and lie on a coordinate plane.

For Property~\ref{lpp:edge}, observe that any two adjacent vertices in $\G$ belong to two common faces (the faces on either side of their shared edge) and therefore have two coordinates in common. Similarly, for Property~\ref{lpp:face}, all vertices on a face of $\G$ take one of their coordinates from that face and therefore lie on the axis-parallel plane consisting of all points with that coordinate value.
\end{proof}

\begin{figure}[t]
\centering\includegraphics[width=4in]{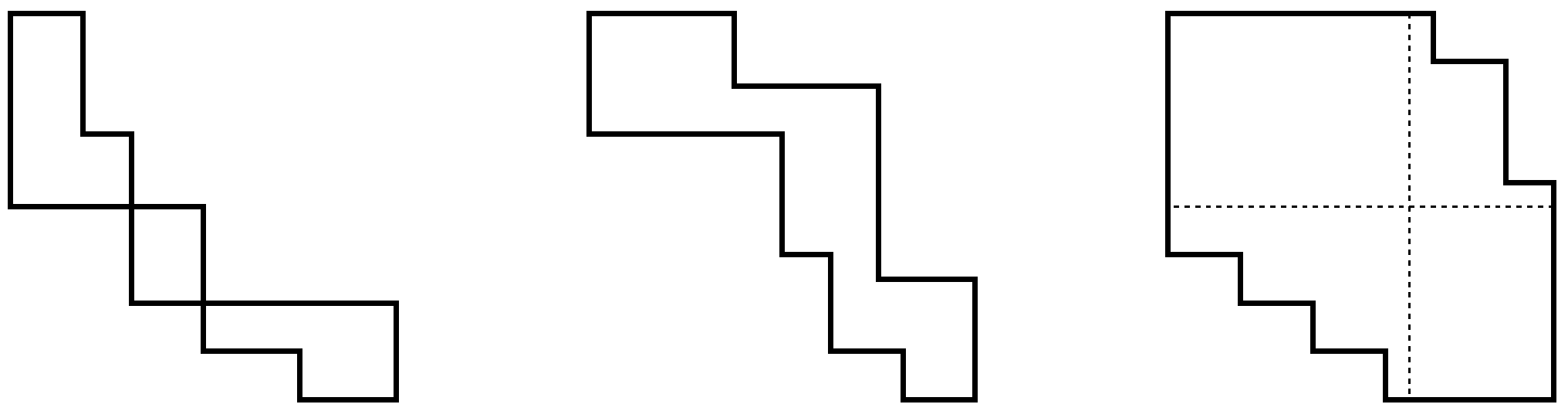}
\caption{Left: A pair of monotone chains that do not form a double staircase polygon. Center: A double staircase polygon that is not fat. Right: A fat double staircase polygon.}
\label{fig:fat2stair}
\end{figure}

Recall that, in a corner polyhedron, each face should have the shape of a \emph{double staircase}: an orthogonal polygon with two extreme vertices connected by two polygonal chains that are monotone in both the coordinate directions, where one chain starts with an $x$-parallel line segment and ends with a $y$-parallel segment, and the other chain starts with a $y$-parallel segment and ends with an $x$-parallel segment. However, simply connecting two extreme vertices by two chains in this way is not sufficient to describe a non-self-intersecting polygon (Fig.~\ref{fig:fat2stair}, left). In order to describe a combinatorial condition that forces a pair of chains connecting two extreme vertices to form a simple polygon, we define the class of \emph{fat double staircase} polygons to be the double staircase polygons in which, for each pair of extreme edges in one of the two coordinate directions, there exists an axis-parallel line segment that connects the two edges and lies within the interior of the polygon (Fig.~\ref{fig:fat2stair}, right). A fat double staircase may equivalently be characterized by total orderings of its boundary segments: the $y$ coordinates of the horizontal segments are ordered from top to bottom along the upper left chain and then from top to bottom along the lower right chain, while the $x$ coordinates of the vertical segments are ordered from left to right, first along the lower right chain and then along the upper left chain.

\begin{lemma}
\label{lem:fat}
With the vertex placement described above, for a graph $\G$ with a regular edge labeling of its dual Eulerian triangulation $\Delta$, every face of $\G$ forms a fat double staircase.
\end{lemma}

\begin{proof}
We first assume that a given face $f$ is not one of the three faces incident with the vertex at the origin. Without loss of generality $f$ is parallel to the $xy$-plane.
Recall that, according to the constraints of a regular edge labeling, the vertex $v$ dual to $f$ in $\Delta$ has edges that alternate between ingoing and outgoing, except in two of the triangles adjacent to $v$, one of which has two ingoing edges and one of which has two outgoing edges. These two triangles will be dual to the two extreme vertices of the fat double staircase formed by~$f$. It remains to show that the two paths in $f$ connecting these extreme vertices form monotonic chains and that the total ordering of the $x$-parallel and $y$-parallel segments in these two chains is the ordering required in a fat double staircase. The $y$-coordinates of the $x$-parallel segments are the $st$-numbers of the neighbors of $v$ that are bichromatic in the numbering of graph $\Delta_{yz}$; these numbers are monotonic along each of the two chains of $f$ because the neighbors corresponding to any two consecutive $y$-parallel segments are connected in $\Delta_z$ by a two-edge oriented path, and the $st$-numbers must be monotonic along this path---see Fig.~\ref{fig:xy2fat}, left. Symmetrically, the $x$-coordinates of the $y$-parallel segments in each chain are also monotonic. Therefore, we do indeed have two monotonic orthogonal polygonal chains connecting the two extreme vertices of $f$.

\begin{figure}[t]
\centering\includegraphics[width=6in]{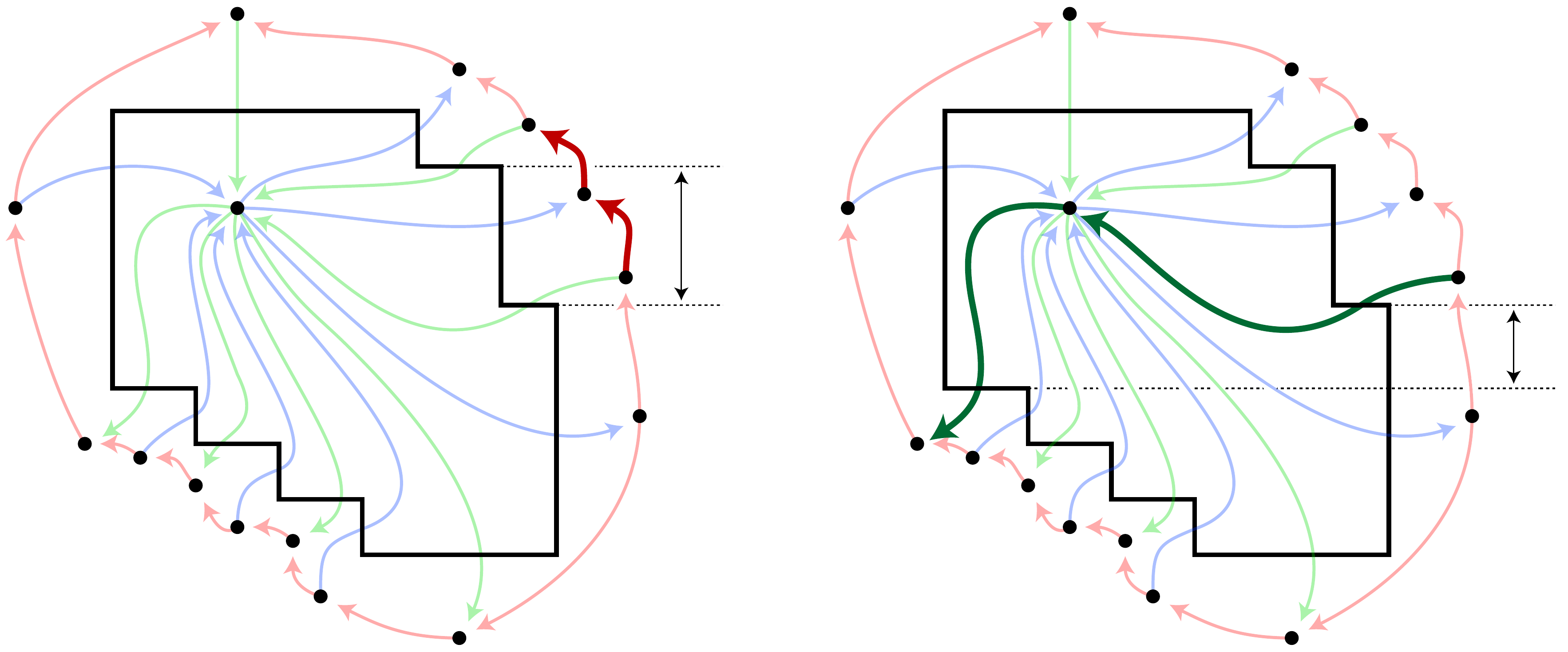}
\caption{Two-edge paths in $\Delta_{yz}$ cause its $st$-numbering to order the boundary segments of a face of $\G$ in the order required for a fat double staircase.}
\label{fig:xy2fat}
\end{figure}

All of the $x$-parallel segments in the upper right chain have greater $y$-coordinates than all of the $x$-parallel segments in the lower left chain, because there exists a two-edge path in $\Delta_y$ from the neighbor of $v$ corresponding to the lowest $x$-parallel segment of the upper right chain to the neighbor of $v$ corresponding to the highest $x$-parallel segment of the lower right chain (Fig.~\ref{fig:xy2fat}, right). Note that this two-edge path in $\Delta_y$ has the opposite color and opposite orientation from the two-edge paths in $\Delta_z$ used to prove monotonicity within each of the two chains; therefore, when we reverse the orientation of one of $\Delta_y$ and $\Delta_z$ to form $\Delta_{yz}$, both types of path will be oriented consistently. The existence of this path causes the neighbor of $v$ corresponding to the lowest $x$-parallel segment on the upper right chain to receive an $st$-number in the $st$-numbering of $\Delta_{yz}$ that is greater than the neighbor of $v$ corresponding to the highest $x$-parallel segment on the lower left chain. Symmetrically, the rightmost $y$-parallel segment on the lower left chain has a smaller $x$-coordinate than the leftmost $y$-parallel segment on the upper right chain. Therefore, $f$ is a fat double staircase.

For the remaining three faces, incident with the vertex at the origin, the result follows by a similar but simpler analysis: the face has two edges on the coordinate axes, and the remaining edges form a monotone plane within the interior of the positive quadrant, therefore the face must form a fat double staircase.
\end{proof}

\begin{corollary}
With the vertex placement described above, for a graph $\G$ with a regular edge labeling of its dual Eulerian triangulation $\Delta$, every face of $\G$ forms a non-self-intersecting polygon.
\end{corollary}

\begin{lemma}
With the vertex placement described above, for a graph $\G$ with a regular edge labeling of its dual Eulerian triangulation $\Delta$, $\G$ is a corner polyhedron.
\end{lemma}

\begin{proof}
The three face of $\G$ dual to the vertices of the root triangle in $\Delta$ are fat staircase polygons that lie on the three boundary quadrants of the positive orthant; therefore, they do not cross each other nor do they cross any of the other faces, which lie within the interior of the positive orthant.
All of the remaining faces, by Lemma~\ref{lem:fat}, are non-self-intersecting polygons (more specifically, fat double staircases).

In the correspondence between regular edge labelings and fat double staircase polygons described in Lemma~\ref{lem:fat}, a face $f$ has a vertex that (when isometrically projected) has an angle of $\pi/3$ if the triangle dual to that vertex has two equally-oriented edges at the vertex dual to $f$, an angle of $2\pi/3$ if the triangle dual to that vertex is cyclically oriented, and an angle of $4\pi/3$ if the triangle is acyclically oriented but has edges both going into and out of $f$. In all cases, the three angles of the polygons incident to a vertex of $\G$ have angles adding to $2\pi$ and the only way for three edges at a common vertex to form these three angles is for the three polygons to by arranged cyclically around the vertex rather than for any one of them to overlap with the other two. Thus, locally, at each vertex of $\G$, each three faces are consistently oriented. When isometrically projected onto the plane, these faces cannot double back on themselves and therefore must form a non-self-overlapping tiling of a region of the plane by polygons.

Thus, the faces of $\G$ can be partitioned into two surfaces (the three back faces and the remaining faces) that don't self-intersect nor cross with each other; it follows that $\G$ is a simple orthogonal polyhedron, and further $\G$ is a corner polyhedron because all of its faces except for the three back ones are visible from $(1,1,1)$.
\end{proof}

The results of the last two sections and this one may be combined to complete the proof of Theorem~\ref{thm:corner-cover}: A graph $\G$ can be represented as a corner polyhedron, with a specified vertex $v$ as the single hidden vertex, if and only if the dual graph of $\G$ has a cycle cover rooted at the triangle dual to $v$.

\section*{Appendix V: Decomposition of 4-connected Eulerian triangulations}

We have now established that graphs of corner polyhedra and graphs with rooted cycle covers are the same thing. In the next section we will show that all 4-connected Eulerian triangulations have rooted cycle covers. In order to do so, we describe in this section a set of reduction operations that allow any 4-connected Eulerian triangulation to be simplified---split into two graphs with smaller number of vertices or reduced to a single graph with smaller number of vertices---in such a way the the obtained graphs remain 4-connected. We will use these simplification steps to guide an inductive proof of the existence of a rooted cycle cover.

\begin{figure}[p]
  \centering
%  \vspace{-1.25\baselineskip}
 \includegraphics{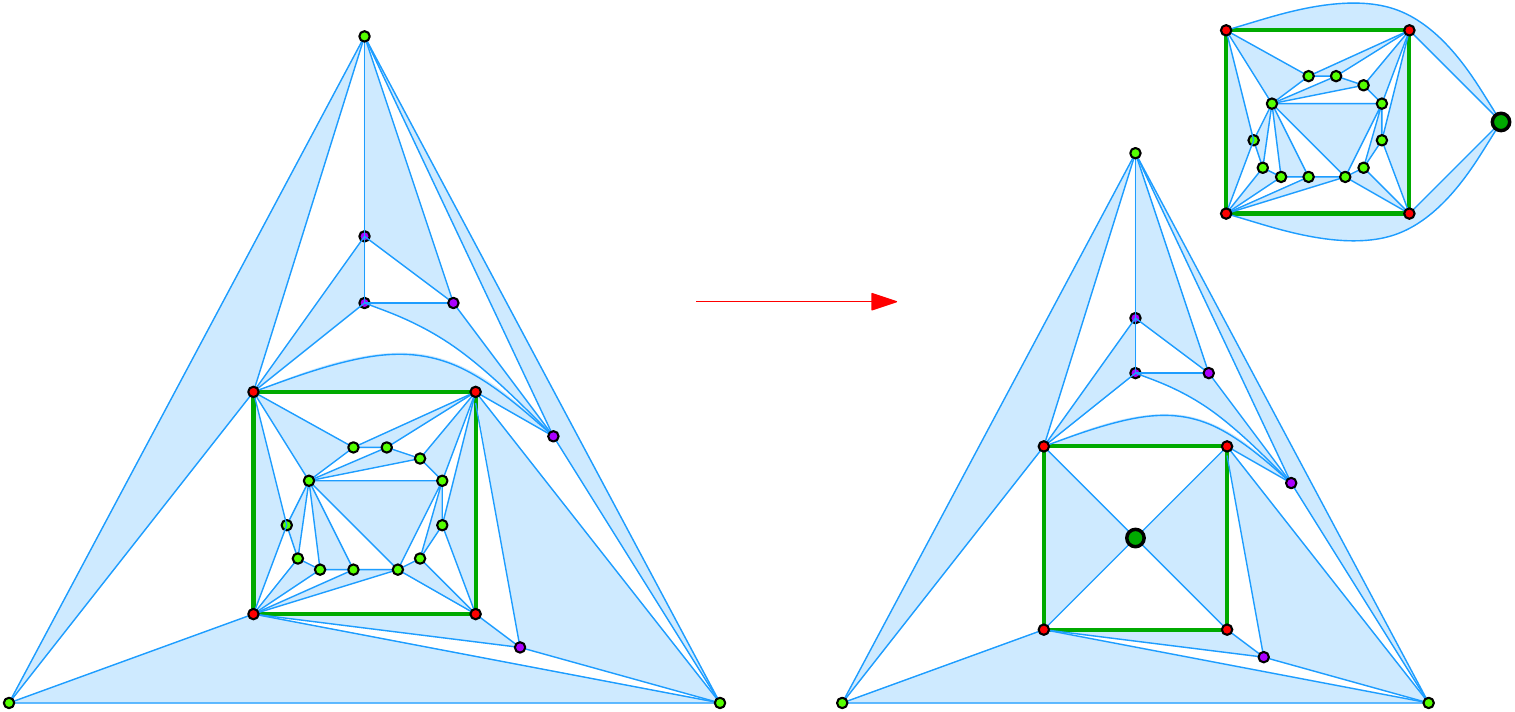}
  \caption{Rule I for simplifying 4-connected Eulerian triangulations: split on a monochromatic 4-cycle.}
  \label{fig:rule1}
\end{figure}

\begin{figure}[p]
  \centering
%  \vspace{-1.25\baselineskip}
  \includegraphics{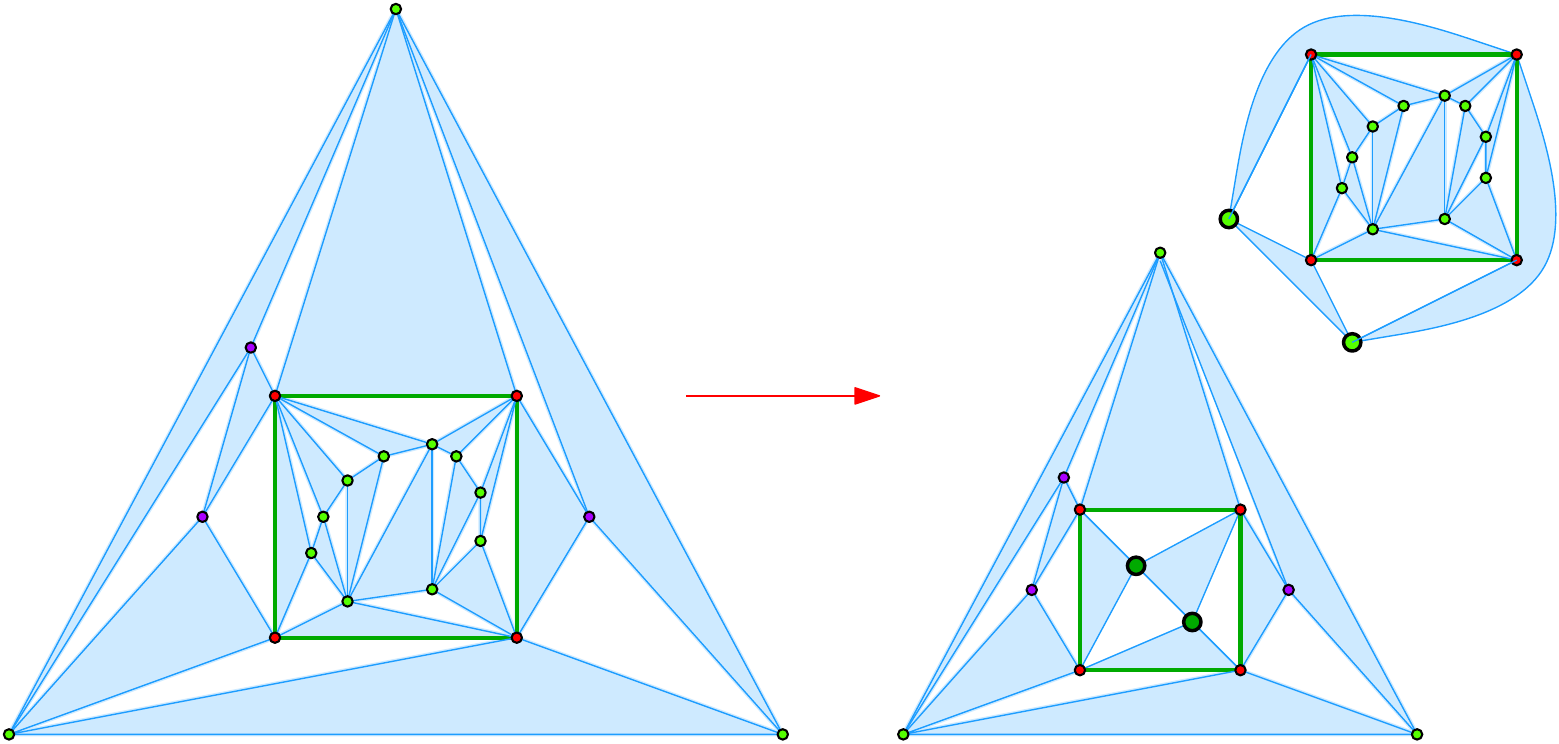}
  \caption{Rule II for simplifying 4-connected Eulerian triangulations: split on a bichromatic 4-cycle.}
  \label{fig:rule2}
\end{figure}
\begin{figure}[p]
  \centering
%  \vspace{-1.25\baselineskip}
 \includegraphics[scale=1.25]{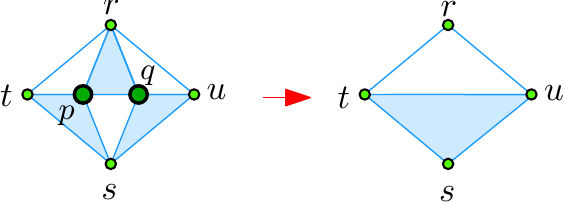}
  \caption{Rule III for simplifying 4-connected Eulerian triangulations: remove two adjacent degree-four vertices, when their neighbors are far enough apart to avoid creating a separating triangle.}
  \label{fig:rule3a}
\end{figure}

We use the following four rules to simplify a 4-connected Eulerian triangulation $\Delta$:

%%%%%%%%
\begin{description}
  \item[(I)]  Suppose that $\Delta$ has a 4-cycle C such that there are three or more
vertices inside the cycle and three or more vertices outside the
cycle and all edges of C have the same color in the rainbow partition of $\Delta$. Then we can split $\Delta$ into
two subgraphs, by replacing the inside of $C$ by a single degree-four
vertex in one of the subgraphs and the outside of C by a single
degree-four vertex in the other subgraph---see Fig.~\ref{fig:rule1} for an illustration. This does not change any
distances between vertices in either of the two subgraphs so there
still are no separating triangles.

  \item[(II)] Suppose that $\Delta$ has a 4-cycle that, as above, separates three or more vertices inside from three or more vertices outside, but that the edges of $C$ have two colors in the rainbow partition of of $\Delta$. Then we can again split $\Delta$ into two subgraphs,
this time by replacing the inside of $C$ by two adjacent degree-four
vertices in one subgraph and the outside of $C$ by two adjacent
degree-four vertices in the other subgraph---see Fig.~\ref{fig:rule2} for an illustration.  Again, the obtained graph is 4-connected
because we do not change the distances between vertices.

  \item[(III)] Suppose $\Delta$ has two adjacent degree-four vertices $p$ and $q$. Let the vertices $r$ and $s$ be the third
vertices of the triangles on edge $pq$, and $t$ and $u$ be the other two
neighbors of $p$ and $q$---see Fig.~\ref{fig:rule3}, left. Suppose further that $r$ and $s$ have degree greater than four, that $tu$ is not an edge of $\Delta$, and that $tru$ and $tsu$ are the only paths of length two from $t$ to $u$.  Then we can replace $p$ and $q$ with an edge connecting $t$ and $u$ without
creating separating triangles---see Fig.~\ref{fig:rule3a}.

\item[(IV)] Suppose $\Delta$ has a degree-four vertex $v$ such that among its four neighbors $abcd$ all have degree greater than four. Then we can contract edges $bv$ and $bd$, as long as this contraction does not create any separating triangles---see Fig.~\ref{fig:rule4}.
\end{description}

\begin{figure}[htb]
  \centering
%  \vspace{-1.25\baselineskip}
 \includegraphics[scale=1.25]{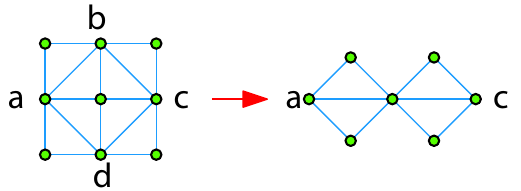}
  \caption{Rule IV for simplifying 4-connected Eulerian triangulations: contract two opposite edges of a 4-cycle.}
  \label{fig:rule4}
\end{figure}

\begin{lemma}
\label{lem:decompose}
Let $\Delta$ be a 4-connected Eulerian triangulation, and let $\delta$ be a designated triangle in $\Delta$. Then $\Delta$ can be simplified by one of rules I-IV above, and further in the cases of rules III or IV the simplification does not collapse triangle $\delta$, unless $\Delta$ is one of the two graphs shown at the center and right of Fig.~\ref{fig:rule3}.
\end{lemma}

\begin{figure}[t]
  \centering
%  \vspace{-1.25\baselineskip}
  \includegraphics{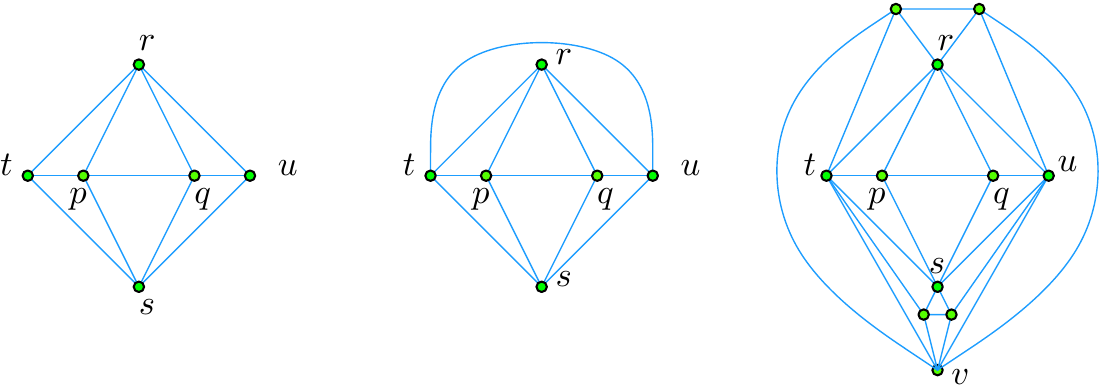}
  \caption{Left: vertex labeling for Lemma~\ref{lem:decompose}; right and center: the two indecomposable graphs for the lemma, the octahedral graph $\Delta_6$ and an 11-vertex graph $\Delta_{11}$ formed by adding two vertices inside each face of the complete bipartite graph $K_{2,3}$.}
  \label{fig:rule3}
\end{figure}

\begin{proof}
Due to Euler's formula, every Eulerian triangulation contains at least six degree-four vertices. It is not possible for all six to be part of $\delta$ or to be among the three vertices on adjacent triangles to $\delta$, unless $\Delta$ is the octahedral graph $\Delta_6$ shown at the center of Fig.~\ref{fig:rule3}. Therefore, let $p$ be a degree-four vertex that has at most one edge connecting it to a vertex of $\delta$.

First, suppose that $p$ has a degree-four neighbor $q$; $q$ cannot be a vertex of $\delta$.
If $p$ and $q$ have a common degree-four neighbor $r$, then triangle $pqr$ is surrounded by another triangle of $\Delta$; since $\Delta$ is by assumption 4-connected, this surrounding triangle must be a face of $\Delta$ and again we have the octahedral graph. Otherwise, let $r$ and $s$  be the two common neighbors of $p$ and $q$, let $t$ be the fourth neighbor of $p$, and let $u$ be the fourth neighbor of $q$, as shown in Fig.~\ref{fig:rule3} (left). If $t$ and $u$ are adjacent, we have the octahedral graph in the center of the figure. If there is a third vertex $v$ that belongs to a two-edge path
$tvu$, then $tvur$ and $tvus$ are 4-cycles. It is not possible for either of these 4-cycles to contain only a single vertex, because that would cause $r$ or $s$ to have odd degree. If either of these cycles has three or more vertices on each of its sides, then by Lemma~\ref{lem:sep-4-cycles} we may perform a Rule~I or Rule~II simplification. Otherwise, both of these cycles contain exactly
two degree-four vertices and we have determined the entire graph---see Fig.~\ref{fig:rule3}(right)). And if there is no such vertex~$v$, then we may perform a Rule~III simplification.

In the remaining case $p$ has degree four but none of its neighbors do. Let the neighbors of $p$ be (in clockwise order) $a,b,c,d$.
Then there are two different ways of collapsing two opposite neighbors
into a single supervertex: we could contract edges $ap$ and $pc$, or we
could contract edges $bp$ and $pd$. It is not possible for either contraction to create a self-loop, for if (say) $a$ and $c$ were already adjacent then $acp$ would be a separating triangle, which we have assumed do not exist in $\Delta$. And, if no Rule~I or Rule~II simplification exists, it is also not possible for either contraction to create a multiple adjacency, for if (say) $a$ and $c$ had a common neighbor $v\notin\{b,d,p\}$ then $apcv$ would be a separating 4-cycle. Each side of this 4-cycle would have to contain more than two vertices, for otherwise $b$ or $d$ would have degree~4, so one of the two rules for simplifying a separating 4-cycle would apply. However, it is possible for the contraction of $ap$ and $pc$ to create a separating triangle, if there exists a path $aefc$ where neither $e$ nor $f$ belongs to the set $\{b,d,p\}$. Similarly, there might exist a path of length three from $b$ to $d$ preventing the other contraction from being applied. But if both paths exist, they must cross, by planarity, and hence must share a vertex: without loss of generality we may assume that the path from $b$ to $d$ has the form $bged$, as shown in Fig.~\ref{fig:crossed-pentagons}, as the other cases are symmetric to this one. Then $cdef$ is a cycle, separating $a$, $b$, and $g$ from the neighbors of $d$ other than the four neighbors $a$, $c$, $e$, and $p$ that we have already identified; recall that $d$ has degree at least six, so there are at least two vertices inside cycle $cdef$. We may also assume without loss of generality that $\delta$ is not inside cycle $cdef$, for otherwise we may apply the same argument to cycle $abge$. If cycle $cdef$ contains exactly two vertices, then they are adjacent vertices of degree four, a case we have already analyzed. Otherwise, cycle $cdef$ contains three or more vertices and allows a Rule~I or Rule~II simplification.
\end{proof}

\begin{figure}[t]
  \centering
  \includegraphics{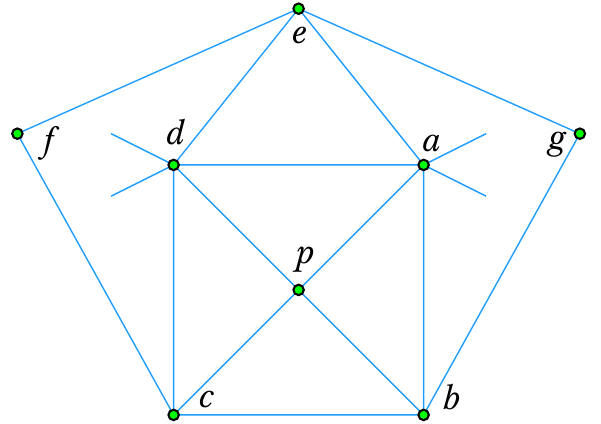}
  \caption{Vertex labeling for the final case of Lemma~\ref{lem:decompose}.}
  \label{fig:crossed-pentagons}
\end{figure}

\section*{Appendix VI: Cycle covers for 4-connected Eulerian triangulations}

In this section we prove that there exists a cycle cover for every 4-connected Eulerian triangulation $\Delta$ and every choice of root triangle. Throughout the section adopt the conventions that the chosen root triangle $uvw$ is drawn as the outer face of a planar embedding of the given graph. Recall also that, in every Eulerian triangulation, there exists a 2-coloring of the triangular faces and an assignment of three colors to the edges such that each triangle has an edge of each color; these colorings are unique up to permutation of the colors. We adopt the convention that the color of the outer triangle is white and that the other triangle color is blue, as shown in our figures.

Recall that a cycle cover is a collection of vertex-disjoint cycles in $\Delta$ that include every vertex other than the three vertices of the root triangle, and that together include exactly one edge from every white triangle. We will prove the existence of the rooted cycle cover for $\Delta$ inductively, using the decomposition rules from the previous section. As we describe in a different section, this inductive proof can be used as the basis of a recursive algorithm for finding cycle covers efficiently.

\begin{lemma}
\label{lem:cc1}
Suppose that the 4-connected Eulerian triangulation $\Delta$ admits a Rule I simplification into two smaller graphs $\Delta\OUT$ (the graph obtained from $\Delta$ by replacing the interior of the 4-cycle  $C=abcd$ with a hyper-vertex $g$) and $\Delta\IN$ the graph obtained by replacing the exterior of $C$ with a hypervertex $h$). Embed the two graphs as in Fig.~\ref{fig:rule1} so that the outer face of $\Delta\OUT$ coincides with that for $\Delta$, the outer face for $\Delta\IN$ includes the new vertex $h$, and the face colorings of $\Delta\OUT$ and $\Delta\IN$ are both consistent with that for $\Delta$. Suppose further that both $\Delta\OUT$ and $\Delta\IN$ admit cycle covers $\C\IN$ and $\C\OUT$ rooted at their outer triangles. Then $\Delta$ also admits a cycle cover rooted at the outer triangle.
\end{lemma}

\begin{proof}
In a Rule I simplification, the edges of $C$ have the same color in the rainbow partition of $\Delta$. Since the edges around each vertex of $\Delta$ alternate in color, each of the vertices of $C$ must be adjacent to an odd number of vertices in the interior of $C$, which implies that, among the four triangles inside $C$ and adjacent to its edges, two triangles on opposite edges of $C$ are white and the other two triangles are blue. Vertex $h$ is placed in such a way that the edges of $C$ adjacent to blue such triangles lie in the interior of $\Delta\OUT$.
Without loss of generality let $b$ and $c$ be the vertices of $\C$ that are not on the outer face of $\Delta\IN$. Consider the portions of covers $\C\IN$ and $\C\OUT$ within the triangles near~$C$.

In $\C\IN$ there must be a cycle containing the path $pbcq$, where $p$ is the common neighbor of $a$ and $b$ and $q$ is the common neighbor of $c$ and $d$ inside $C$. This follows from the observation that vertices $b$ and $c$ are included in cycles of the cover while vertices $a$, $d$ and $h$ belong to the outer triangle and are not included---see Fig.~\ref{fig:gin-gout} (left). Let $C_{bc}$ denote the cycle containing $pbcq$.

In $\C\OUT$ there must be one cycle $C_g$ that covers $g$. $C_g$ contains exactly one of the four paths $agc$, $bgd$, $agd$, and $bgc$, since exactly one edge of every white rectangle adjacent to $g$ is part of the cover.

We combine $\C\IN$ and $\C\OUT$ to obtain the cycle cover $\C$ of $\Delta$ by first including all cycles of $\C\OUT\setminus\{C_g\}$ and all cycles of $\C\IN\setminus\{C_{bc}\}$ to $\C$. We then replace $C_{bc}$ and $C_g$ by a single cycle $C'$ which covers all the remaining vertices other than those in the outer triangle of $\Delta$. We construct the final cycle $C'$ of $\C$ as follows:

\begin{enumerate}
  \item Remove the path $pbcq$ from $C_{bc}$.
  \item Let $x$ and $y$ be vertices of $C \bigcap C_g$ such that $x$ is adjacent to $p$ and $y$ is adjacent to $q$. Use edges $xp$ and $yq$ to connect the path $C_g \setminus {g}$  to $C_{bc} \setminus {pbcq}$ to form a cycle $C'$.
 \end{enumerate}
Fig.~\ref{fig:gin-gout} illustrates the construction for all possible choices of $x$ and $y$.
It is easy to see that the obtained set of cycles satisfies the definition of rooted cycle cover for $\Delta$.
\end{proof}

\begin{figure}[htb]
  \centering
%  \vspace{-1.25\baselineskip}
  \includegraphics{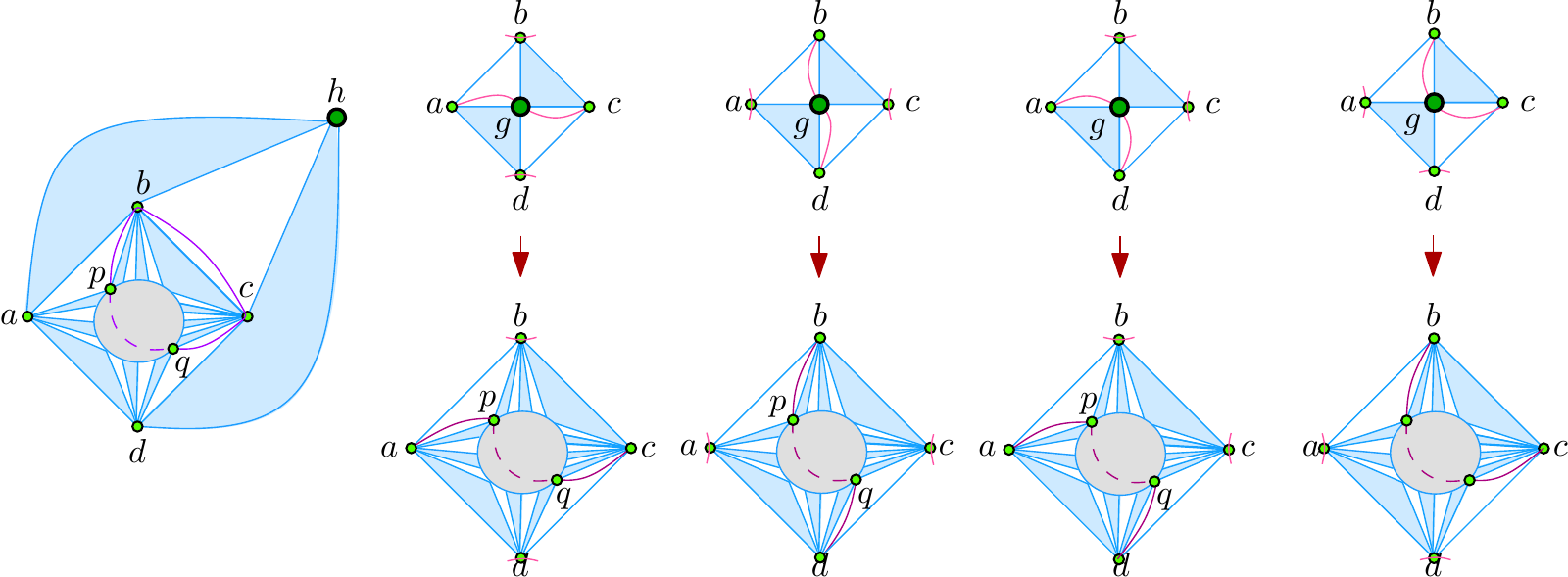}
  \caption{Cycle covers for Rule I simplification.}
  \label{fig:gin-gout}
\end{figure}

\begin{lemma}
\label{lem:cc2}
Suppose that the 4-connected Eulerian triangulation $\Delta$ admits a Rule II simplification into two smaller graphs $\Delta\OUT$ (the graph obtained from $\Delta$ by replacing the interior of the 4-cycle  $C=abcd$ with two hyper-vertices $g_1$ and $g_2$) and $\Delta\IN$ the graph obtained by replacing the exterior of $C$ with two hypervertices $h_1$ and $h_2$). Embed the two graphs as in Fig.~\ref{fig:rule2} so that the outer face of $\Delta\OUT$ coincides with that for $\Delta$, the outer face for $\Delta\IN$ includes the new vertices $h_1$ and~$h_2$, and the face colorings of $\Delta\OUT$ and $\Delta\IN$ are both consistent with that for $\Delta$. Suppose further that both $\Delta\OUT$ and $\Delta\IN$ admit cycle covers $\C\IN$ and $\C\OUT$ rooted at their outer triangles. Then $\Delta$ also admits a cycle cover rooted at the outer triangle.
\end{lemma}

\begin{proof}
In a Rule~II simplification, the edges of $C$ alternate colors, from which it follows that two opposite vertices of $C$ are adjacent to an odd number of vertices in $C$'s interior, and the other two vertices of $C$ have evenly many neighbors inside $C$. This means that there exists a single vertex of $C$---let it be $d$---that is adjacent to two blue faces inside $C$ containing edges of $C$. We embed $h_1$ and $h_2$ to keep $a$ in the interior of $\Delta\IN$.

Denote $C_{abc}$ the cycle in $\C\IN$ containing the path $pabcq$, where $p$ is the common neighbor of $a$ and $d$ inside $C$ and $q$ is the common neighbor of $c$ and $d$ inside $C$. Such a cycle exists because this is the only way to cover vertices $p,a,b,c,q$ without violating the properties of a rooted cycle cover.

In $\C\OUT$, there must be a cycle $C_g$ that includes edge $g_1g_2$, for otherwise it would not be possible to include both $g_1$ and $g_2$ in the cycle cover. However, there are two cases for how this cycle can be connected to the rest of $\Delta\IN$: it may contain the path $ah_1h_2c$, or it may include vertex $d$. In the latter case, we may assume without loss of generality that it contains the path $ah_1h_2d$; the other case, of a path $dh_1h_2c$, is symmetric.

\begin{figure}[htb]
  \centering
%  \vspace{-1.25\baselineskip}
  \includegraphics{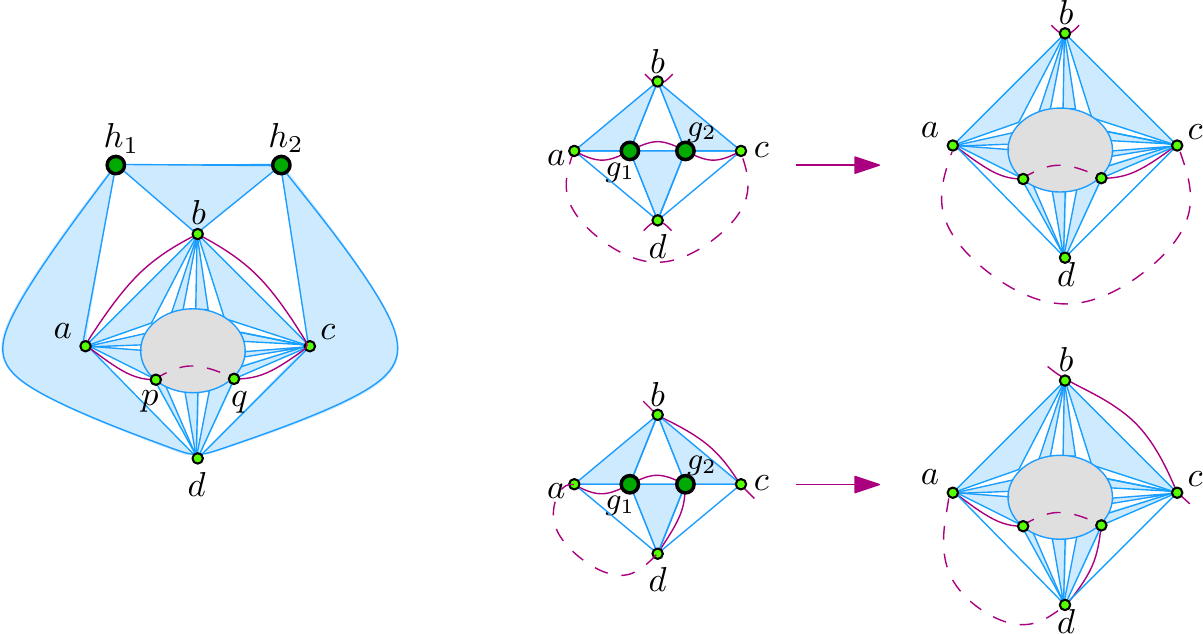}
  \caption{Cycle covers for Rule~II simplification.}
  \label{fig:gin-gout-1b}
\end{figure}

We construct the cycle cover $\C$ of $\Delta$ by first adding all cycles in $\C\OUT\setminus C_{abc}$ and all cycles in $\C\IN\setminus C_{g}$. In the case that $C_g$ contains the path $ah_1h_2c$,
we construct the final cycle $C'$ by removing the path $abc$ from $C_{abc}$ and merging it with the subpath of $C_g$ that connects $a$ and $c$ and does not contain $g_1$ and $g_2$. In the case that $C_g$ contains the path $ah_1h_2c$, we instead construct $C'$ by removing the path $abcq$ from $C_{abc}$ and merging it with the subpath of $C_g$ that connects $a$ and $g_2$ and does not contain $g_1$. All cycles are vertex disjoint, all internal vertices in $\Delta$ are covered and every white face contributes exactly one edge to $\C$, so $\C$ is a rooted cycle cover for $\Delta$.
\end{proof}

\begin{lemma}
\label{lem:cc3}
Suppose that the 4-connected Eulerian triangulation $\Delta$ admits a Rule~III simplification that replaces a pair of adjacent degree four vertices $p$ and $q$ by an edge $tu$ connecting their two neighbors, forming a simpler graph $\Delta'$, that neither $p$ nor $q$ belongs to the outer triangle, and that $\Delta'$ admits a cycle cover rooted at the outer triangle. Then $\Delta$ also admits a cycle cover rooted at the outer triangle.
\end{lemma}

\begin{proof}
Let the vertices $r$ and $s$ be the third vertices of the triangles on edge $pq$. Without loss of generality let $tru$ be the white triangle on the four-cycle $trus$, let $\C'$ be a cycle cover for $\Delta'$, and let $C_{tru}$ be the cycle of $\C'$ that includes an edge of triangle $tru$. If $C_{tru}$ includes edge $tu$ then we replace this edge by the path $tpqu$, and if $C_{tru}$ includes edge $tr$ then we replace this edge by the path $tpqr$; the case that $C_{tru}$ includes edge $ru$ is symmetric to the case that it includes edge $tr$. In all cases the result is a cycle cover $\C$ of $\Delta$---see Fig.~\ref{fig:6-2flip}.
\end{proof}
\begin{figure}[htb]
  \centering
%  \vspace{-1.25\baselineskip}
 \includegraphics{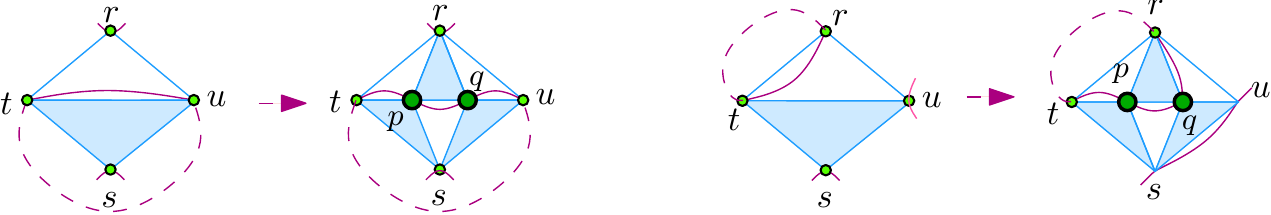}
  \caption{Cycle covers for Rule III simplification.}
  \label{fig:6-2flip}
\end{figure}

\begin{lemma}
\label{lem:cc4}
Suppose that the 4-connected Eulerian triangulation $\Delta$ admits a Rule~IV simplification in which a degree-four vertex $p$ has its neighbors in the cyclic order $abcd$, $p$ does not belong to the outer cycle of $\Delta$, all four neighbors have degree greater than four, and the simplification collapses edges $bp$ and $pd$ of $\Delta$ into a new hyper-vertex $g$ forming the simpler graph $\Delta'$. Suppose further that $\Delta$ has a cycle cover~$\C'$ rooted at the outer triangle. Then $\Delta$ also admits a cycle cover rooted at the outer triangle.
\end{lemma}

\begin{proof}
Let $C_g$ be the cycle of $\C'$ that passes through $g$. Suppose first that $C_p$ separates the blue triangle incident with edge $ag$ from the blue triangle incident with edge $gc$. In this case, the two edges of $C_g$ that are incident to $g$ correspond in $\Delta$ to one edge incident to $b$ and one edge incident to $d$, and we can form a cycle cover in $\Delta$ by adding the edges $bp$ and $pd$ to the cycle cover---see Fig.~\ref{fig:cc-rule4a}.
\begin{figure}[htb]
  \centering
%  \vspace{-1.25\baselineskip}
 \includegraphics{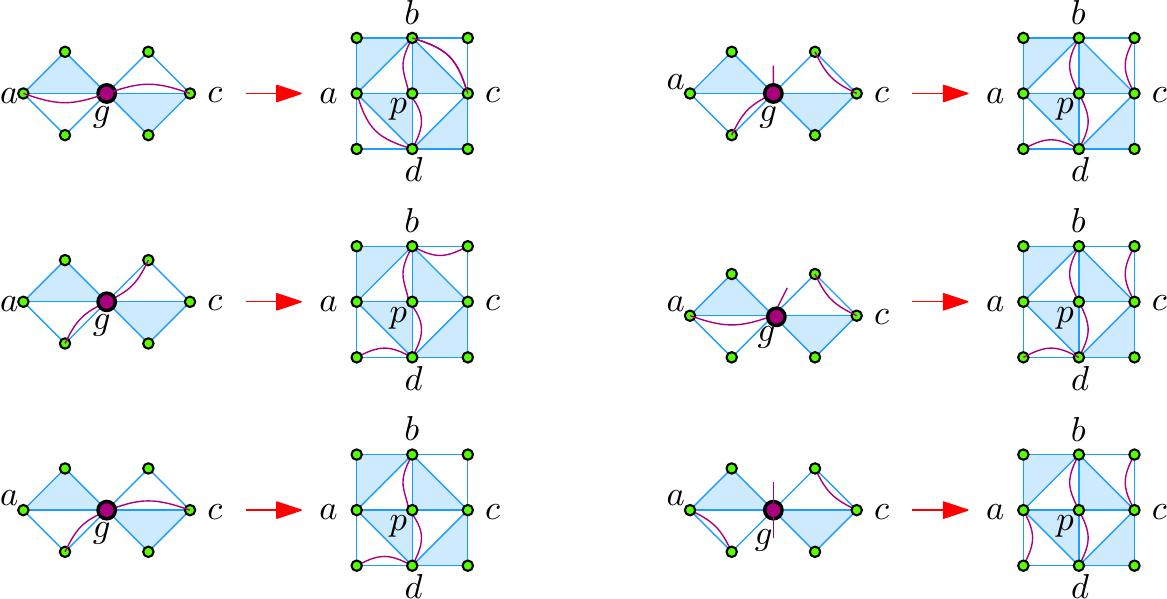}
  \caption{Cycle covers for Rule IV expansion when $C_g$ separates $a$ from $c$.}
  \label{fig:cc-rule4a}
\end{figure}

Alternatively, it may be the case that $C_g$ does not separate the blue triangles. In this case, it corresponds to a cycle in $\Delta$ that covers one of the two vertices $b$ or $d$ (without loss of generality $d$), leaving $p$ and $b$ uncovered. In this case, let $e$ be the third vertex of the white triangle $cge$ containing edge $cg$ in $\Delta'$; there must be a cycle $C_{ce}$ in $\C'$ that includes edge $ce$, because vertex $g$ is not included in the edge of $cge$ that belongs to $\C'$. In $\Delta$, we form a cycle cover $\C$ by removing edge $ce$ from $C_{ce}$ and replacing it by the three-edge path $cpbe$---see Fig.~\ref{fig:cc-rule4b}.
\end{proof}
\begin{figure}[htb]
  \centering
%  \vspace{-1.25\baselineskip}
 \includegraphics{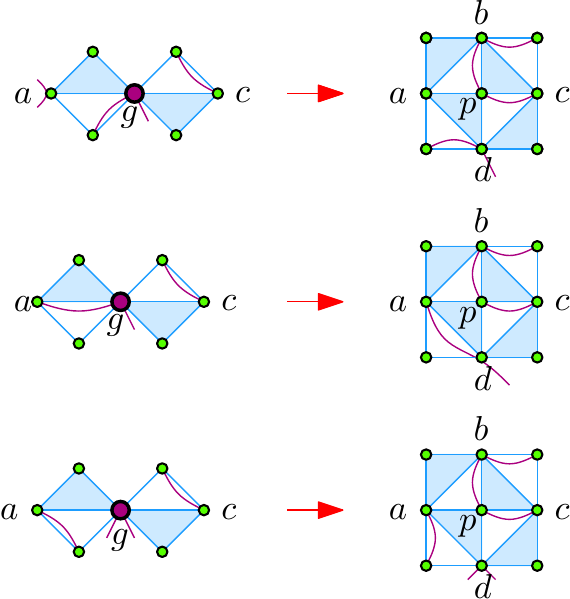}
  \caption{Cycle covers for Rule IV expansion when $C_g$ does not separate $a$ from $c$.}
  \label{fig:cc-rule4b}
\end{figure}

\begin{lemma}
\label{lem:cc4c}
For every 4-connected Eulerian triangulation $\Delta$, and every choice of a root triangle $uvw$, there is a cycle cover of $\Delta$ rooted at $uvw$.
\end{lemma}

\begin{proof}
We use induction on the number of vertices in $\Delta$. As a base case, $\Delta$ is one of the graphs depicted in Fig.~\ref{fig:rule3}(middle and right), the octahedron $\Delta_6$ and the 11-vertex graph $\Delta_{11}$. Both of these graphs allow rooted cycle covers---see Fig.~\ref{fig:simple-graph-cover} below. In the case of the octahedron in Fig.~\ref{fig:rule3}(middle) and Fig.~\ref{fig:simple-graph-cover} (left), all choices of root triangle are equivalent. In the case of the graph $\Delta_{11}$  in Fig.~\ref{fig:rule3}(right) and Fig.~\ref{fig:simple-graph-cover}(right), the symmetries of the graph take any choice of root triangle to one of two non-equivalent choices. The choice shown in Fig.~\ref{fig:simple-graph-cover} below corresponds to the corner polyhedron depicted in Fig.~\ref{fig:orthotypes} (left), while the other choice also admits a rooted cycle cover and corresponds to the corner polyhedron depicted together with its cycle cover in Fig.~\ref{fig:corners}(right).
\begin{figure}[htb]
  \centering
%  \vspace{-1.25\baselineskip}
 \includegraphics{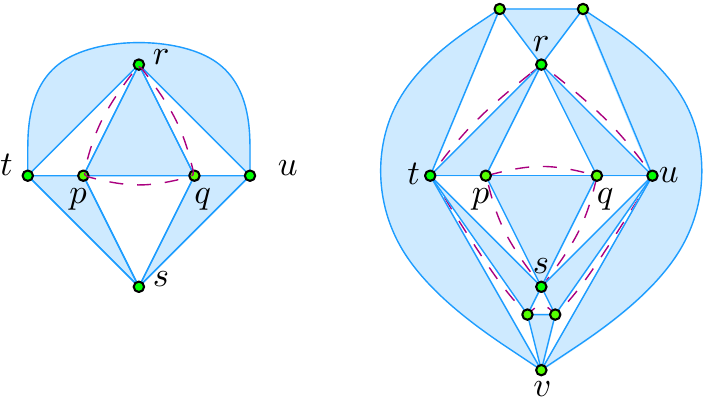}
  \caption{Base cases for Lemma~\ref{lem:cc4c}.}
  \label{fig:simple-graph-cover}
\end{figure}

If $\Delta$ is not one of these two base cases, then by Lemma~\ref{lem:decompose} it admits a Rule~I, Rule~II, Rule~III, or Rule~IV simplification to smaller 4-connected Eulerian triangulations. By induction, these smaller graphs have cycle covers for any choice of root triangle, and the existence of a cycle cover for $\Delta$ itself follows from Lemmas \ref{lem:cc1}, \ref{lem:cc2}, \ref{lem:cc3}, and~\ref{lem:cc4}.
\end{proof}

The proof of Theorem~\ref{thm:4-connected}, that graphs with 4-connected dual Eulerian triangulations may be represented as corner polyhedra, follows immediately from this lemma and from Theorem~\ref{thm:corner-cover} that graphs with cycle covers may be represented as corner polyhedra.

\section*{Appendix VII: Characterization of corner polyhedra}

We have seen in the previous sections that a graph $\G$ is the graph of a corner polyhedron, with hidden vertex $h$, if and only if the dual Eulerian triangulation $\Delta$ has a cycle cover rooted at the triangle dual to $h$. We have also seen that, for a 4-connected Eulerian triangulation, there is a cycle cover rooted at any triangle, so the duals of these triangulations always have corner polyhedron representations. In this section we extend these results to the non-4-connected case, by completing the proof of Theorem~\ref{thm:corner-characterization} characterizing the graphs of corner polyhedra.

Thus, let $\G$ be a 3-connected bipartite cubic planar graph, and $v$ a vertex of $\G$ such that we wish to determine whether $\G$ can be represented as a corner polyhedron with hidden vertex $v$. Let $\Delta$ be the Eulerian triangulation dual to $\G$, and as in previous sections adopt the convention that $\Delta$ is drawn so that the triangle dual to $v$ is the outer face, and so that the faces are 2-colored blue and white with the outer face white. If $uvw$ is a separating triangle in $\Delta$, then it follows from Lemma~\ref{lem:sep-tri-eulerian} that the three face triangles of $\Delta$ that are inside $uvw$ and incident to edges $uw$, $vw$, and $uw$ all have the same color; as in Section~\ref{sec:corner}, we define $uvw$ to have even parity if these three face triangles are white, and odd parity if these three face triangles are blue.

\medskip\noindent\emph{Proof of Theorem~\ref{thm:corner-characterization}.} The theorem states that $\G$ has a corner polyhedron representation, with $v$ as hidden vertex, if and only if all separating triangles of $\Delta$ have odd parity. In one direction, suppose that this is the case. Consider the collection of Eulerian triangulations formed by splitting $\Delta$ on each of its separating triangles, giving the inner split component the same coloring it had in $\Delta$, and in the outer split component replacing the portion of $\Delta$ within the separating triangle by a single blue face triangle. Then each of these graphs is an Eulerian triangulation by Lemma~\ref{lem:sep-tri-eulerian} and is 4-connected; therefore, it has a cycle cover rooted at its outer triangle. The union of the cycles in these covers forms a cycle cover in $\Delta$. Thus, $\Delta$ has a rooted cycle cover, from which it follows that $\G$ has a corner polyhedron representation.

In the other direction, suppose that some separating triangle $uvw$ has even parity, and let $\Delta'$ be the Eulerian triangulation formed by the triangles inside $uvw$. The dual to $\Delta'$ is a regular bipartite graph, so it has equal numbers of vertices on each side of its bipartition; translating this fact back to $\Delta'$ itself, and applying Euler's formula, if $\Delta'$ has $k$ vertices then it has $k-2$ blue triangles (including its outside triangle) and $k-2$ white triangles. We now assume for a contradiction that $\Delta$ has a cycle cover $\C$, and count in two different ways the number of pairs $(x,y)$ where $x$ is a vertex in $\Delta'$ and $y$ is an edge in $\Delta'\cap\C$, getting two incompatible bounds on the numbers of pairs.
First, we count the pairs $(x,y)$ by vertices. There are $k-3$ vertices interior to $\Delta'$, each of which belongs to two edges and forms two pairs. Additionally, the three vertices $u$, $v$, and $w$ have at least one pair each, because no matter how we choose an edge in $\C$ from the three outer white triangles in $\Delta'$ the chosen edge will include one of these vertices. Thus, counting by vertices, the number of pairs is at least $2(k-3)+3=2k-3$. However, counting by edges, there are $k-2$ white triangles in $\Delta'$, each of which supplies an edge of $\Delta'\cap\C$, so there are $k-2$ edges of $\Delta'\cap\C$ forming only $2(k-2)=2k-4$ pairs. This contradiction between a number being $2k-4$ when counted one way and at least $2k-3$ when counted a different way proves that a cycle cover $\C$ cannot exist and therefore that $\G$ has no corner polyhedron representation.\qed

\section*{Appendix VIII: Characterization of $xyz$ polyhedra}

We begin the proofs of the results claimed in Section~\ref{sec;xyz} by showing that the two classes of polyhedra described in that section are combinatorially equivalent.

\begin{lemma}
\label{lem:xyz=single}
Let $P$ be a singly-intersecting simple orthogonal polyhedron. Then there is an $xyz$ polyhedron with the same graph as $P$, in which all vertex coordinates are integers in the range $[1,n/4]$.
\end{lemma}

\begin{proof}
For each coordinate plane, number the faces of the polyhedron with the integers $1$, $2$, $\ldots$, so that faces with smaller coordinate values have smaller numbers; if two faces are coplanar, choose their numbers arbitrarily. The faces with a given orientation partition the $n$ vertices, and each has at least $4$ vertices, so there are at most $n/4$ numbers used. Now (as in~\cite{Epp-GD-08}) move each vertex of $P$ to the point given by the numbers of its incident faces; that is, if the face parallel to the $yz$ plane has number $k$, set the $x$-coordinate of the vertex to $k$. Call the resulting polyhedron $P'$.

Define a \emph{feature} of $P'$ to be the rectangular region of a face lying between two parallel edges of the face. Every point of the boundary of $P'$ belongs to a feature, and every feature of $P'$ corresponds to a (possibly degenerate) feature in $P$ because of the order-preserving nature of the transformation from $P$ to $P'$. If two features intersect in $P'$, they must intersect in $P$ as well, again because of the order-preserving nature of the transformation. But $P$ has no intersections of non-adjacent features, so neither does $P'$. Therefore, $P'$ remains a simple orthogonal polyhedron. Any line containing an edge of $P'$ is contained in the face planes only of the two faces that meet in that edge, and therefore contains no other edge of $P'$. Since each vertex is incident to an edge within each axis-parallel line containing it, it also follows that the line containing an edge of $P'$ cannot contain a vertex that is not an endpoint of that edge. Therefore, $P'$ forms an $xyz$ graph and is an $xyz$ polyhedron.
\end{proof}

\begin{corollary}
\label{cor:3cb}
Every singly-intersecting simple orthogonal polyhedron or $xyz$ polyhedron represents a 3-connected bipartite cubic planar graph.
\end{corollary}

\begin{proof}
3-connectivity follows from the more general fact that $xyz$ graphs are 3-connected from our previous paper~\cite{Epp-GD-08}. Every face is even, since it alternates between edges in two coordinate directions. Thus, as a planar graph with even faces, the graph of an $xyz$ polyhedron must be bipartite.
\end{proof}

\medskip\noindent
\emph{Proof of Theorem~\ref{thm:xyz}}.
Recall that Theorem~\ref{thm:xyz} states that 4-connected bipartite cubic planar graphs, graphs of $xyz$ polyhedra, and graphs of singly-intersecting simple orthogonal polyhedra are the same thing. The equivalence between the two types of polyhedron is Lemma~\ref{lem:xyz=single}, and the fact that every graph of one of these polyhedra is 3-connected and bipartite is Corollary~\ref{cor:3cb}. It remains to prove that every 3-connected bipartite cubic planar graph $\G$ can be represented as a singly-intersecting simple orthogonal polyhedron.

To do so, we use induction on the number of separating triangles in the dual graph $\Delta$ of $\G$, an Eulerian triangulation. As a base case, if $\Delta$ has no separating triangles, it is 4-connected, and we may represent $\G$ as a corner polyhedron, which is certainly a singly-intersecting simple orthogonal polyhedron. Otherwise, let $t$ be a triangle which separates $\Delta$ into two subgraphs $\Delta_1$ and $\Delta_2$, where $\Delta_2$ has as few vertices as possible among all such splits. Then $\Delta_1$ and $\Delta_2$ are themselves Eulerian (Lemma~\ref{lem:sep-tri-eulerian}). $\Delta_2$ can have no separating triangle itself, because that triangle would separate a graph with even fewer vertices. Let $\G_1$ and $\G_2$ be the duals of $\Delta_1$ and $\Delta_2$ respectively.

\begin{figure}[t]
\centering\includegraphics[width=5.5in]{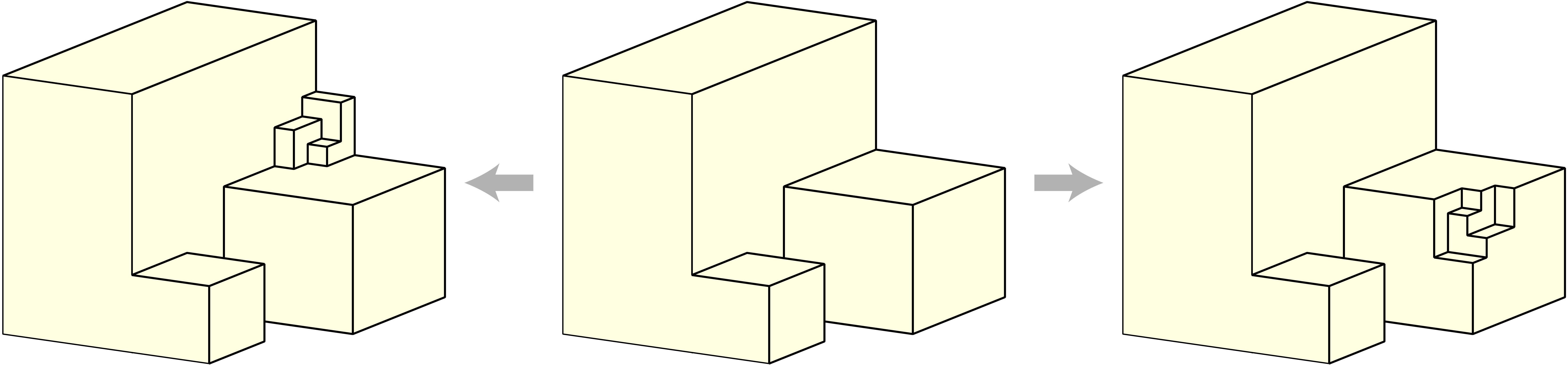}
\caption{Gluing a corner polyhedron onto a convex or concave vertex of an orthogonal polyhedron (right) and onto a saddle vertex (left).}
\label{fig:cornerglue}
\end{figure}

By induction, $\G_1$ can be represented as a singly-intersecting simple orthogonal polyhedron, and $\G_2$ can be represented as a corner polyhedron with the hidden vertex being the one dual to $t$.
A polyhedron representing $\G$ may be formed by replacing the vertex representing $t$ in $\G_1$ by a copy of the corner polyhedron representing $\G_2$ (minus its hidden vertex). The replacement can be performed geometrically by placing the corner polyhedron so that its copy of $t$ coincides with that for $\G_2$, and so that the edges adjacent to $t$ in $\G_1$ and $\G_2$ lie on the same rays, and then forming the $xyz$ graph from the union of the vertex sets of $\G_1$ and $\G_2$, minus the two copies of the shared vertex $t$.
There are two different cases for how $t$ may be represented in $\G_1$: the faces near it may separate one orthant of space from the other seven orthants, forming a convex vertex or its complement, or it may separate three orthants from the other five orthants, forming a saddle (Fig.~\ref{fig:cornerglue}). But in each case the replacement causes using only local changes to the polyhedron near the replaced vertex, so if the copy of the corner polyhedron that replaces $t$ is sufficiently small, it will not cross any other feature of the polyhedron, nor have any faces coplanar with existing faces of the polyhedron, producing a new $xyz$ polyhedron representing $\G$.\qed

\section*{Appendix IX: Characterization of simple orthogonal polyhedra}

Recall that Theorem~\ref{thm:atom}, to be proved in this section, characterizes simple orthogonal polyhedra in terms of the triconnected components of their graphs. The definition of triconnected components~\cite{Mac-DMJ-37,HopTar-SJC-73} involves recursively partitioning the graph using pairs of vertices the removal of which would disconnect the graph, but it is convenient for our purposes to use a modification of triconnected components that is specialized to 3-regular graphs and instead partitions a graph into subgraphs based on pairs of edges. Therefore, we define a \emph{split pair} of a graph $\G$ to be a pair of edges the removal of which disconnects the graph (necessarily into exactly two components). A \emph{split component} is formed from one of these two components by adding a \emph{virtual edge} between the two degree-two vertices of the component, or from repeating the same splitting process within larger split components. The \emph{atoms} of a 3-regular graph are its 3-connected split components, formed by repeatedly subdividing the graph into split components until every remaining graph is 3-connected. The split components and atoms of a graph may be multigraphs rather than simple graphs; for instance, the graph shown in Fig.~\ref{fig:nonpolyhedral} has four atoms, three of which are cubes and the fourth of which is a multigraph with two vertices and three edges. As we will later show, in a 3-regular graph, the atoms are the non-cyclic triconnected components.

We would like to claim that every split component formed in this way is itself a simple orthogonal polyhedron, with its faces in the same planes as the original polyhedron. Unfortunately, this is not true: if $\G$ is the graph of a simple orthogonal polyhedron with split components $\G_1$ and $\G_2$, then $\G_1$ may form a cavity into which part of $\G_2$ protrudes, so that $\G$ is simple but $\G_2$ (without the cavity) is not. To avoid this problem we define a larger class of combinatorial objects, which we call \emph{orthogonal polyhedroids}. An orthogonal  polyhedroid is a complex of vertices, edges, and faces, with the following properties:
\begin{itemize}
\item Every vertex is represented by a distinct three-dimensional point.
\item Every edge is represented by an axis-parallel line segment connecting its two vertices.
\item Three perpendicular edges meet at every vertex.
\item Every face is an abstract polygon the edges of which alternate between being parallel to two coordinate axes. Distinct edges may cross, lie along the same line, or even have a nontrivial line segment as their intersection.
\item The complex of vertices, edges, and faces, viewed as an abstract complex without regard to its three-dimensional embedding, has the overall topology of a sphere.
\end{itemize}

Every simple orthogonal polyhedron therefore forms an orthogonal polyhedroid, as does every planar $xyz$ graph.

\begin{lemma}
\label{lem:sop2graph1}
Every orthogonal polyhedroid forms a 3-regular bipartite 2-connected planar graph.
\end{lemma}

\begin{proof}
Planarity follows from the assumption that the polyhedroid has the topology of a sphere, and 3-regularity follows from the assumption that three edges meet at every vertex. Because of the alternation between directions of the edges around each face, each face has even length and by Lemma~\ref{lem:even-faces} the graph $\G$ must be bipartite. It remains to show that $\G$ is 2-connected. To do so, consider removing any vertex $v$ of $\G$; let $u$ and $w$ be any two other vertices. Then, because the polyhedroid has the topology of the sphere, it can have only one boundary component, so there is a path on its surface (viewed as an abstract topological space) from $u$ to $w$, avoiding the single point $v$. Any part of this path that passes through the interior of a face $f$ can be replaced by a path that instead follows the edges and vertices around the boundary of $f$ in one of two ways; if $f$ contains $v$ we choose the one of these two ways that does not pass through $v$. By making this replacement to all parts of the path, we find a path through $\G\setminus\{v\}$ from $u$ to $w$; since such a path can be found for any two vertices after the removal of any third vertex, $\G$ is 2-connected.
\end{proof}

\begin{lemma}
\label{lem:split-bipartiteness}
Let $\G$ be a 2-connected 3-regular graph, and $\G_1$ and $\G_2$ be the two split components formed from some split pair in $\G$. Then $\G$ is bipartite if and only if $\G_1$ and $\G_2$ are both bipartite.
\end{lemma}

\begin{proof}
Let the two edges of the split pair be $uv$ and $xy$, with $u$ and $x$ in $\G_1$ and $v$ and $y$ in $\G_2$. First, suppose that $\G_1$ and $\G_2$ are both bipartite; we may color them separately so that $u$ and $v$ have opposite colors. Since $x$ is adjacent to $u$ in $\G_1$ and $y$ is adjacent to $v$ in $\G_2$, it follows that $x$ and $y$ also have opposite colors and we have found a consistent 2-coloring of $\G$, showing that it is bipartite.

In the other direction, suppose that $\G$ is bipartite, and fix a 2-coloring of it. It then follows that $u$ and $x$ must have opposite colors; for if they had the same color, then the graph $\G'$ formed from $\G_1$ by removing edge $ux$ would have two degree-two vertices on one side of the bipartition and none on the other side of the bipartition. But this is an impossibility, because the total number of edges in any bipartite graph is equal to the number of vertex-edge adjacencies on a single side of the bipartition, and this number would be congruent to 1 mod 3 on the side with $u$ and $x$ and to 0 mod 3 on the other side. Thus, $\G_1$ is consistently 2-colored and is therefore bipartite. The same argument applies as well to $\G_2$.
\end{proof}

\begin{corollary}
\label{cor:split-bipartiteness}
A 2-connected 3-regular graph $\G$ is bipartite if and only if all of its atoms are bipartite.
\end{corollary}

\begin{lemma}
\label{lem:split-polyhedroidality}
Let $\G$ be a 2-connected 3-regular graph, and $\G_1$ and $\G_2$ be the two split components formed from some split pair in $\G$. If $\G$ is the graph of an orthogonal polyhedroid, then $\G_1$ and $\G_2$ are also graphs of orthogonal polyhedroids.
\end{lemma}

\begin{proof}
The two edges of a split pair must belong to the same two faces of $\G$, so they both lie on the line formed by the intersection of the planes containing these faces. Therefore, the new edge added to form the split component $\G_1$ must be axis-parallel, and must be perpendicular to the other two edges at its two endpoints, satisfying the requirements for an orthogonal polyhedroid. A system of faces forming a planar graph for $\G_1$ may be found from a spherical embedding of the abstract complex representing the polyhedroid for $\G$, by replacing the two faces containing the split pair by two simpler faces from which all vertices in $\G_2$ have been removed. The construction of an orthogonal polyhedroid representing $\G_2$ is symmetric.
\end{proof}

\begin{figure}[t]
\centering\includegraphics[width=3in]{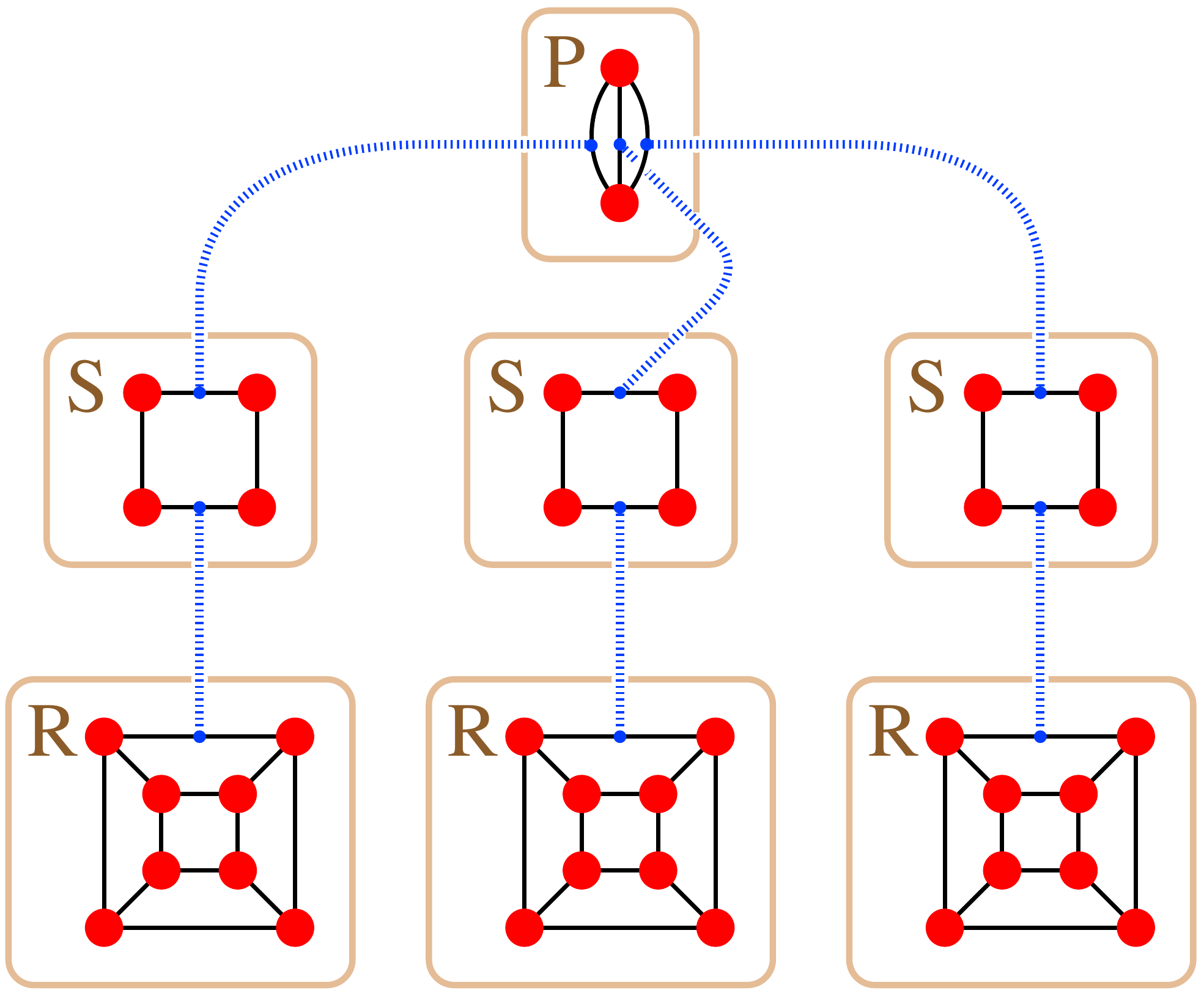}
\caption{The SPQR tree for the graph shown in Fig.~\ref{fig:nonpolyhedral}.}
\label{fig:spqr}
\end{figure}

An SPQR tree~\cite{DiBTam-FOCS-89,DiBTam-ICALP-90,GutMut-GD-01} is a data structure for representing the 3-connected components of any 2-connected graph. It takes the form of an unrooted tree (Fig.~\ref{fig:spqr}), the nodes of which are labeled by undirected graphs. The graph in each node may be a multigraph with two vertices and three or more edges (a P node), a cycle of three or more edges (an S node), or a 3-connected simple graph (an R node); the Q nodes from the original definition of SPQR trees are not needed for our purposes. Each edge $pq$ of the SPQR tree is associated with two oriented edges in the two graphs $\G_p$ and $\G_q$ associated with SQPR tree nodes $p$ and $q$. If an edge of a graph $\G_p$ is associated in this way with an SPQR tree edge, it is called a \emph{virtual edge}; a virtual edge can be associated with only one SPQR tree edge.

An SPQR tree $T$ represents a 2-connected graph $\G_T$, formed as follows: whenever SPQR tree edge $pq$ associates the virtual edge $uv$ of $\G_p$ with the virtual edge $wx$ of $\G_q$, form a single larger graph by merging $u$ and $w$ into a single supervertex, merging $v$ and $x$ into another single supervertex, and deleting the two virtual edges. By performing this gluing step on each edge of the SPQR tree we form the graph $\G_T$; the order of performing the gluing steps does not affect the result. For instance, the SPQR tree of Fig.~\ref{fig:spqr} represents in this way the graph of Fig.~\ref{fig:nonpolyhedral}.

With the additional assumptions that no two P nodes are adjacent and that no two S nodes are adjacent, there is a unique SPQR tree that represents any 2-connected graph; it may be constructed from the graph in linear time~\cite{GutMut-GD-01}. Each vertex in each graph $\G_p$ corresponds to a unique vertex in the overall graph $\G$. The 2-cuts of $\G$ (that is, the pairs of vertices the removal of which disconnects $\G$) are exactly the 2-cuts of its SPQR tree nodes: that is, pairs of vertices that are the two vertices of a P node, that are any two vertices of an S node, or that are two endpoints of a virtual edge in an R node.

\begin{lemma}
\label{lem:spqr3}
A 2-connected graph $\G$ is 3-regular if and only if its SPQR tree has the following form:
for every two adjacent nodes, one of the two nodes must be an S node, every P node $p$ must be associated with a graph $\G_p$ with three edges, every R node $r$ must be associated with a graph $\G_r$ that is itself 3-regular, and every S node $s$ must be associated with a cycle $\G_s$ that is of even length and that alternates between virtual and non-virtual edges.
\end{lemma}

\begin{proof}
If the SPQR tree has the given form, then every vertex of $\G$ is formed either from a vertex in an R node that is glued to from zero to three S nodes, or a vertex in a P node that is glued to two or three S nodes. Gluing an S node in place of a virtual edge does not change the degree of the associated vertices, so each vertex has the same degree it had in its associated P node or R node, that is, the graph is 3-regular.

Conversely, suppose that the SPQR tree does not have this form. If there is a high degree vertex in a P or an R node, then additional gluing cannot decrease the degree, so there is a high degree vertex in $\G$. There can be no vertex of degree lower than 3 in an R node $r$ because then the associated graph $\G_r$ would not be 3-connected. If an edge of the SPQR tree connects two nodes neither of which is an S node, then gluing those two nodes produces vertices of degree four or more at the endpoints of the glued virtual edges, and again those high degrees will persist into $\G$. If two non-virtual edges are adjacent in an S node's graph, then the degree two vertex where they meet will persist into $\G$. And if two virtual edges are adjacent in an S node's graph, then gluing in the P or R nodes connected to those virtual edges will increase the degree of the vertex in the middle to four, and again that high degree vertex will persist into $\G$. Thus, in all cases, a violation of the constraints on the form of the SPQR tree leads to a graph that is not 3-regular.
\end{proof}

\begin{lemma}
\label{lem:spqr-split}
For a 3-regular graph $\G$, each split pair of edges form two non-virtual edges in an S node of the SPQR tree of $\G$.
\end{lemma}

\begin{proof}
Any two endpoints of a split pair form a 2-cut in $\G$, and must therefore belong to a single node of the SPQR tree. By Lemma~\ref{lem:spqr3}, every two vertices that form a cut belong to an S node $s$, because the two vertices of a cut in a P node or an R node are also the vertices of a cut in an S node that is glued to it. By Lemma~\ref{lem:spqr3} again, there is only one pair of non-virtual edges adjacent to these two vertices in $\G_s$, and this pair of edges form a split pair. It must be the split pair we started with, because any other combination of edges with the same two vertices fails to separate the graph, as they belong to different subgraphs $\G_p$ each of which is at least 2-connected.
\end{proof}

\begin{lemma}
\label{lem:spqr-atoms}
In a 3-regular graph $\G$, the atoms of $\G$ are exactly the graphs $\G_p$ where $p$ is either a P node or an R node of the SPQR tree of $\G$.
\end{lemma}

\begin{proof}
Using Lemma~\ref{lem:spqr-split}, it follows by induction on the number of nodes of the SPQR tree of $\G$ that, for each split component $H$ of $\G$, the SPQR tree of H has a special form: it consists of a subtree $S$ of the SPQR tree of $\G$, where within each S node $s$ of $S$ we modify the graph $\G_s$ by removing the two endpoints of any virtual edge whose partner does not belong to $S$ and by reconnecting the remaining vertices in cyclic order. Additionally, the split components into which $\G$ has been partitioned by a sequence of splits together contain all of the P nodes and R nodes of the SPQR tree of $\G$.

Whenever a subtree of the above form does not consist of a single P node or R node, it can be further split by a split pair drawn from one of its S nodes, by Lemma~\ref{lem:spqr-split}. Therefore, once $\G$ has been split into the atoms, each atom must form a graph represented by a single P node or R node of the SPQR tree of $\G$, and all such nodes must represent atoms.
\end{proof}

It follows from the above characterization that every 2-connected 3-regular graph has an \emph{atomic decomposition}, a sequence of subdivisions of the graph into split components in which, at every split, one of the two split components is an atom.

\begin{lemma}
\label{lem:atomic}
Every 2-connected 3-regular graph has an atomic decomposition.
\end{lemma}

\begin{proof}
The result follows by induction on the number of nodes in the SPQR tree of the graph $\G$. As a base case, if there is a single node, it must be a P or R node and the result follows. Otherwise, let $p$ be a node of degree one in the SPQR tree. Then $p$ must be a P or R node, for by Lemma~\ref{lem:spqr3} each S node has two or more neighbors in the tree. The two endpoints of the single virtual edge of $p$ must also be the endpoints of two non-virtual edges in the S node adjacent to $p$, these two edges form a split pair, and splitting $\G$ according to that split pair produces the atom $\G_p$ as one component and another component with either one or two fewer nodes in its SPQR tree (accordingly as the S node adjacent to $p$ has length greater than four or not).
\end{proof}

\begin{lemma}
\label{lem:molecular}
Let $\G$ be a 2-connected 3-regular planar graph that is split by a split pair $uv$, $wx$ into two split components $\G_1$ and $\G_2$, where $\G_1$ is the graph of a simple orthogonal polyhedron and $\G_2$ is a 3-connected simple planar bipartite graph. Then $\G$ is itself the graph of a simple orthogonal polyhedron.
\end{lemma}

\begin{proof}
Let $P$ be a polyhedron representing $\G_1$. Edge $uw$ in $P$ is represented geometrically by an axis-parallel line segment, and its two adjacent faces form a right-angled wedge near that segment.

It follows from Theorem~\ref{thm:xyz} that $\G_2$ can itself be represented as a simple orthogonal polyhedron. More, the method of proving Theorem~\ref{thm:xyz} shows that, if $\Delta$ is the 4-connected Eulerian triangulation chosen as the root of the decomposition tree of the dual of $\G_2$ by its separating triangles, and if $abc$ is the root triangle of $\Delta$, then the polyhedron representing $\G_2$ lies within the three-dimensional positive orthant, with the vertex of $\G_2$ dual to $abc$ at the origin, its three neighbors placed on the coordinate axes, and all other vertices placed strictly interior to the positive orthant. We may choose any vertex of $\G_2$ to be the one placed at the origin; the vertex we choose is $v$.

We may now shrink and translate the polyhedron representing $\G_2$ so that, instead of lying within the positive orthant, it lies within the wedge defined by the faces adjacent to edge $uw$ of $\G_1$, with $v$ and $x$ placed in the correct order between $u$ and $w$ on the line through this edge, and with all other vertices strictly interior to this wedge. The faces of the two polyhedra adjacent to the two edges $uw$ and $vx$ form two coplanar pairs, and each pair may be replaced by a single face, the symmetric difference of the two faces in the coplanar pair. All other faces of $\G_1$ and $\G_2$ remain unchanged. The result is a polyhedral representation of the entire graph~$\G$.
\end{proof}

\begin{lemma}
\label{lem:23split}
For any 2-connected graph $\G$, the SPQR tree of $\G$ contains a P node if and only if there exists in $\G$ a pair of vertices $u$ and $v$ such that either (1) $u$ and $v$ are nonadjacent and their removal leaves at least three connected components, or (2) $u$ and $v$ are adjacent and their removal leaves at least two connected components.
\end{lemma}

\begin{proof}
In one direction, if the SPQR tree contains a P node then the two vertices $u$ and $v$ of that node have the described property. In the other direction, if $u$ and $v$ have this property, then using $u$ and $v$ as a split pair of vertices produces a split component in which $u$ and $v$ are connected by multiple edges, one for each component of $\G$ formed by removing $u$ and $v$ from the graph and one for edge $uv$ if that edge belongs to $\G$. This split component must form a P node in the SPQR tree.
\end{proof}

\begin{lemma}
\label{lem:23biconn}
Suppose that, in the cubic graph $\G$, there do not exist two vertices $u$ and $v$ such that removing $u$ and $v$ from $\G$ leaves three components (counting an edge between $u$ and $v$, if one exists, as a component). Then $\G$ is biconnected.
\end{lemma}

\begin{proof}
We prove the contrapositive, that if $\G$ is not biconnected then it has two vertices $u$ and $v$ as described. But if $\G$ is biconnected, it has an articulation point $v$. Because there are only three edges at $v$, one of the blocks that contains $v$ must have only a single edge $uv$ incident to $v$. Then $uv$ forms a bridge in $\G$, and removing $u$ and $v$ from $\G$ produces at least three components: one containing edge $uv$, one containing the other neighbors of $u$, and one containing the other neighbors of~$v$.
\end{proof}

\medskip\noindent
\emph{Proof of Theorem~\ref{thm:atom}}. The first two equivalent graph classes in the theorem are 2-connected 3-regular graphs such that each triconnected component is a bipartite polyhedral graph or an even cycle, and 3-regular planar bipartite graphs in which no 2-vertex removal leaves three components (counting an edge between the two removed vertices as a component). This equivalence follows easily from the known facts that a graph is planar if and only if its triconnected components are planar~\cite{Mac-DMJ-37}, that a graph is bipartite if and only if its triconnected components are bipartite (Corollary~\ref{cor:split-bipartiteness}), that a graph has a triconnected component that is neither a simple 3-connected graph nor a cycle if and only if it has two vertices the removal of which forms three components (Lemma~\ref{lem:23split}), and that a graph without any two such vertices is biconnected (Lemma~\ref{lem:23biconn}). By Lemma~\ref{lem:spqr-atoms}, this same class of graphs is also the class of bipartite 2-connected 3-regular planar graphs in which every atom is a simple graph, for every P node $p$ in an SPQR tree of such a graph is represented by a non-simple graph $\G_p$ (a multigraph) while every R node is represented by a simple graph.

Suppose that a graph $\G$ is the graph of a simple orthogonal polyhedron. Then by Lemma~\ref{lem:sop2graph1} it forms a bipartite 2-connected 3-regular planar graph, and by Lemma~\ref{lem:split-polyhedroidality} all of its atoms are also simple graphs; thus, $\G$ belongs to the class of graphs described above.

Conversely, suppose that $\G$ belongs to this class of graphs. Then by Lemma~\ref{lem:atomic} it has an atomic decomposition, and by repeatedly applying Lemma~\ref{lem:molecular} to the splits in this decomposition we may build up a representation of $\G$ as a simple orthogonal polyhedron. Therefore, every graph in this class is the graph of a simple orthogonal polyhedron.\qed

Finally, we observe that there is some freedom in choosing the relative orientation of the two polyhedra to be merged in Lemma~\ref{lem:molecular}. However, all of the planar embeddings of a 2-connected planar graph may be specified by choosing an orientation for the gluing at each virtual edge of its SPQR tree and by choosing a permutation of the edges in each P node. Because the SPQR trees of simple orthogonal polyhedra do not have P nodes, only the orientation of the gluing at each virtual edge is needed to determine an embedding. This shows that, when a graph $\G$ is the graph of a simple orthogonal polyhedron, all planar embeddings of $\G$ may be realized as simple orthogonal polyhedra.

\section*{Appendix X: Algorithmic Steinitz theorems}
In this section we describe an algorithm that takes a two-connected cubic planar graph as an input and represents it as a simple orthogonal polyhedron.  The algorithm for general simple orthogonal polyhedra should perform the following steps; the algorithms for the other cases are similar but with fewer steps.

\begin{enumerate}
\item \emph{Find an atomic decomposition.}
Finding an atomic decomposition is essentially just finding the SPQR tree and/or triconnected components, and can be done in linear time as described in~\cite{HopTar-SJC-73,GutMut-GD-01}. Before computing the SPQR tree, we check that the graph is 3-regular and bipartite, so that we may apply the lemmas in Appendix~VII restricting the structure of the SPQR tree in this case. If the SPQR tree has a P node, no simple orthogonal polyhedron representation can exist. Otherwise, the atoms in the atomic decomposition are given by the R nodes of the SPQR tree, and the atomic decomposition tree can be formed by choosing an arbitrary ordering for the atoms incident to each S node of the SPQR tree.

\item \emph{Transform each atom (a 3-connected bipartite cubic planar graph) into its dual Eulerian triangulation.}
We can find a planar embedding of the graph $\G$ in linear time~\cite{HopTar-JACM-74}. Creating the dual triangulation $\Delta$ can also be done in linear time, as each feature of $\Delta$ (a vertex, edge, or face) corresponds one-for-one with a dual feature (a face, edge, or vertex) of the embedding of~$\G$.

\item \emph{Partition each Eulerian triangulation into 4-connected Eulerian triangulations by splitting it on its separating triangles.}
All triangles in $\Delta$ may be found in $O(n)$ time~\cite{ChiNis-SJC-85,ChrEpp-TCS-91}, and the separating triangles are just the ones that are not faces; we elaborate this step in a subsection below.

\item \emph{Recursively decompose each 4-connected Eulerian triangulation into simpler 4-connected Eulerian triangulations using separating 4-cycles, pairs of adjacent degree-4 vertices, and isolated  degree-4 vertices.} While returning from the recursion, undo the steps of the decomposition and build a cycle cover for the Eulerian triangulation. This is the most complicated step of the algorithm, and we elaborate it below.

\item \emph{Convert the cycle covers into regular edge labelings.}
This step is a simple pattern matching process, detailed below.

\item  \emph{Construct the graphs $\Delta_{xy}$ and find an $st$-numbering of each such graph.}
$\Delta_{xy}$ is just a subgraph of $\Delta$, with edge inclusion and orientation determined by the regular edge labeling. Constructing the $st$-numbering of $\Delta_{xy}$ may be performed by using a breadth-first search on the graph: if we wish to construct an $xyz$ polyhedron, we may use the breadth-first traversal order as the numbering, but in the other cases (where a more compact representation is desirable) we may use the distances from the root node of the search as the numbering. In either case this step takes linear time.

\item \emph{Use the $st$-numbering to construct a corner representation dual to each 4-connected Eulerian triangulation.}
After the previous step each vertex $v$ of $\Delta$ has received a number in the $st$-numbering of the graph $\Delta_{xy}$ in which $v$ is incident to edges of two different colors; this number represents the coordinate of the face in $\G$ that is dual to~$v$. In this step we simply copy these numbers to the incident vertices of each face in $\G$. Each vertex of $\G$ gets three numbers assigned to it, which we use as its Cartesian coordinates in the corner representation.

\item \emph{Glue the corner polyhedra together to form orthogonal polyhedra dual to each non-4-connected Eulerian triangulation.} In order to perform this step efficiently we represent vertex coordinates implicitly as positions within a doubly linked list, as detailed below.

\item \emph{ Glue 3-connected polyhedra together to form arbitrary simple orthogonal polyhedra.} This part can be done in the same way as gluing corner polyhedra in the previous step, with the same analysis.
\end{enumerate}

The algorithm for constructing $xyz$ polyhedron representations uses only steps 2--8, after testing that the input graph really is 3-connected. Similarly, the algorithm for constructing a corner polyhedron representation uses only steps 3--7, together with a new step between steps 4 and 5 that computes the union of the cycles covers on the 4-connected components, forming a single cycle cover for the dual Eulerian triangulation to the entire input graph. We omit the details as they are not significantly different than for the simple orthogonal polyhedron case.

We describe and analyze the steps that need to be elaborated in the subsections below.

\subsection*{Decomposing an Eulerian triangulation on its separating triangles}

A \emph{separating triangle} in a plane graph is a 3-cycle that contains vertices of the graph both in its exterior and in its interior. That is, it is a 3-cycle that is not a face of the planar embedding. Equivalently, for a 3-connected planar graph, a separating triangle is a triangle with the property that removing its three vertices from the graph causes the remaining graph to have more than one connected component. Any two separating triangles are either interior-disjoint or one of them contains the other in its interior.

We first identify all separating triangles; all triangles may be found in linear time~\cite{ChiNis-SJC-85,ChrEpp-TCS-91} and the separating triangles are the triangles that are not themselves faces. By performing a single breadth-first search, we may also determine in linear time, for each face triangle $\delta$ of $\Delta$, the distance from $\delta$ to the outer face of $\Delta$ (measured in terms of the number of steps between two faces that share an edge). With this information, we may determine, for each separating triangle $\delta$, which side of $\delta$ is the inside and which is the outside: among the six face triangles sharing an edge with $\delta$, the one that's closest to the outer face is on the outside of $\delta$, and the other face triangles adjacent to and outside of $\delta$ are the ones with the same color.

We then form a directed graph from $\Delta$, by replacing each undirected edge $uv$ in $\Delta$ by either one directed edge or a pair of directed edges: if $u$ is contained within a triangle having $v$ as one of its three corners (either a separating triangle or the outer triangle), we direct the edge from $u$ to $v$, if $v$ is contained within a triangle having $u$ as one of its three corners, we direct the edge from $v$ to $u$, and in the remaining case we replace undirected edge $uv$ by a pair of directed edges going in both directions. This orientation can be determined in linear total time by scanning the edges and triangles incident to each vertex in cyclic order. The fact that separating triangles are properly nested implies that this case analysis determines a unique orientation or pair of orientations for every edge: it is not possible for $u$ to be inside a triangle through $v$ and simultaneously for $v$ to be inside a triangle through~$u$.
We compute the strongly connected components of this directed graph, in linear time~\cite{Tar-SJC-72}.

Then, if $\delta$ is a triangle that has other vertices inside it (that is, $\delta$ is either the outer triangle or a separating triangle), let $S_\delta$ be the set of vertices that are inside $\delta$ but not inside any other triangle. We claim that $S_\delta$ forms a strongly connected component of the directed graph described above. For, if $u$ and $v$ are any two vertices of $S_\delta$, then there exists a path in $\Delta$ from $u$ to $v$ that does not intersect $\delta$ itself (by the 3-connectivity and planarity of $\Delta$, the subset of vertices inside $\delta$ must be connected). The shortest such path cannot cross into any other separating triangle inside $\delta$, because if it did then we could form a shorter path using one of the triangle edges. Therefore, all edges of the shortest path are bidirected and $u$ and $v$ belong to the same strongly connected component. However, it is not possible for the strongly connected component of $u$ and $v$ to contain other vertices that do not belong to $S_\delta$, because the vertices outside $\delta$ do not have paths connecting from them to the inside of $\delta$ and the vertices inside smaller separating triangles within $\delta$ do not have paths connecting to them from the outside of these separating triangles.

The \emph{condensation} of the directed graph constructed above (a graph formed by replacing each strongly connected component by a single supervertex) forms an \emph{inclusion tree} of separating triangles with the outer face triangle of $\Delta$ as a root.  Each strongly connected component, together with the edges connecting it to its parent triangle and the edges of its parent triangle, forms a \emph{4-connected component} of $\Delta$, a minimal subgraph that is itself an Eulerian triangulation and has no separating triangles.
The result of this step is a collection of 4-connected Eulerian triangulations, together with a tree structure describing how to glue them back together again.

We store each Eulerian triangulation (a 4-connected component of the dual of our cubic input graph) as a collection of edges, vertices and triangles; each of these objects stores information about the objects it is adjacent to. Furthermore each triangle stores a bit indicating its color in the face two-coloring, each triangulation stores a pointer to its root triangle, and each edge stores its color in the rainbow partition.

\subsection*{Simplifying a 4-connected Eulerian triangulation}

In this subsection we discuss the algorithmic implementation of our decomposition of a 4-connected Eulerian triangulation into simpler Eulerian triangulations, via decomposition operations that split the graph along separating 4-cycles and that simplify pairs of adjacent degree-four vertices and single isolated degree-four vertices.

We recursively decompose each 4-connected Eulerian triangulation into simpler 4-connected Eulerian triangulations using separating 4-cycles, pairs of adjacent degree-4 vertices, and isolated degree-4 vertices according to the rules described in Appendix~V, until we end up with graphs $\Delta_6$ and $\Delta_{11}$ as the base cases of the recursion. We construct cycle covers for these simple graphs at the bottom of the recursion, then reverse the decomposition by returning from the recursion. As we return from each recursive call to the algorithm, we undo the decomposition step made at that call and update the cycle cover for the graph obtained by reversing the decomposition as described in Appendix~VI.
When we return from the outer call to the recursion, we will have constructed a cycle cover for the original 4-connected triangulation.

Each decomposition step reduces the potential function $\Phi=|V|-6|K|$, where $V$ is the set of vertices in the remaining partially decomposed graph and $K$ is the set of its connected components. This potential function is initially $n-6$, it is always non-negative (each component has at least six vertices), and it decreases by two or more at each step in which the number of components increases. Therefore, there can be at most $n/2$ such steps. As each step that increases the number of components can create at most four new vertices, the total number of vertices that are ever present at any one time during the decomposition is at most $3n$.

In the rest of this section we describe the way that we find the decomposition operations efficiently.
In order to find and perform the steps of the decomposition, we require the following data structures.

\begin{itemize}
\item A collection of Eulerian triangulations representing the partially decomposed graph. This collection will be represented by edge, vertex, and triangle objects. Each edge points to its adjacent vertices and triangles. Each triangle stores its color and whether it is the outer triangle of its component; it also points to its adjacent edges. Each vertex stores a cyclic list of its adjacent edges. Additionally, while we are returning from the recursion each edge will store a bit indicating whether it is part of the cycle cover, and each vertex and triangle will store a pointer to the incident cycle cover edges. We also store, for each vertex $v$ of the partially decomposed graph, the number $\min(36,d(v))$; we say that $v$ has \emph{high degree} if this number is 36.
These modified degree numbers may be computed in constant time for any vertex by scanning its adjacency list, and therefore may be updated in constant time whenever we perform a decomposition operation.
However, we do not store any information telling us how these edge, vertex, and triangle objects are partitioned into the connected components of the partially decomposed graph.

\medskip
\item A data structure allowing us to quickly determine whether any two vertices are adjacent in the partially decomposed graph, and if so enabling us to find the edge object connecting them. In the randomized setting this can be simply a hashed dictionary having as its keys pairs of adjacent vertices, with the dictionary for each key pointing to the edge object connecting that pair of vertices. With a hashed dictionary, adding or removing an edge, and testing adjacency, may be performed in expected $O(1)$ time per operation.

\medskip
In the deterministic setting, adjacency testing of a static planar graph can still be performed in constant time using bounded-degree orientations~\cite{ChrEpp-TCS-91}, but dynamic adjacency testing appears to be more difficult. To solve it, we store with each vertex object $v$ a number $N(v)$, a unique identifier for that vertex, in the range $0\le N(v)<3n$. When a 4-cycle causes us to split the graph into two, we assign new unused identifiers to the vertices on one side of the split; the bound of $3n$ on the number of vertices present at any time in the algorithm ensures that this is always possible. We associate with each edge $uv$ the number $(min(N(u),N(v))\cdot 3n + \max(N(u),N(v))$; this number uniquely determines the endpoints of the edge, as their identifiers are the quotient and remainder formed when dividing the edge number by $3n$.  The problem of testing for the existence of an edge then becomes one of searching for this number among a collection of $O(n)$ other numbers (one for each edge in the graph), each of which has magnitude $O(n^2)$. This integer searching problem may be solved deterministically in linear space, using time $O((\log\log n)^2/\log\log\log n)$ per edge addition, edge removal, or adjacency test~\cite{AndTho-STOC-00}.

\medskip
\item A collection of \emph{good vertices}. Define a \emph{collapsible 4-cycle} to be a set of vertices $p,q,r,s,t,u$ as in Fig.~\ref{fig:rule3a}~(left), forming a chain $t-p-q-u$ where $p$ and $q$ have degree four, $t$ and $u$ have high degree, and $r$ and $s$ are both adjacent to all of $t-p-q-u$. We call a degree-four vertex \emph{bad} if it belongs to the root triangle, to a collapsible 4-cycle or to a face triangle in which the other two vertices are high degree. A good vertex is a degree-four vertex that is not bad. We may test whether a vertex $v$ is good in constant time, by checking that it has four neighbors and examining the modified degree numbers of its neighbors and of the neighbors of its degree-four neighbors. We maintain the collection of good vertices as a doubly-linked list, and we store with each vertex object a pointer to its position in the list (if it is a good vertex; the pointer's value is undefined otherwise). With this representation, it is possible to find a good vertex in constant time, and to change the status of a vertex to be good or to be not good in constant time per status change. Each decomposition operation may change whether $O(1)$ other vertices are good (the vertices whose neighborhoods change because of the operation, and the degree-four vertices that are adjacent to other vertices that changed from being low to high degree or vice versa) and therefore the updates to this data structure take $O(1)$ time per decomposition operation.
\end{itemize}

\begin{figure}[t]
\centering\includegraphics{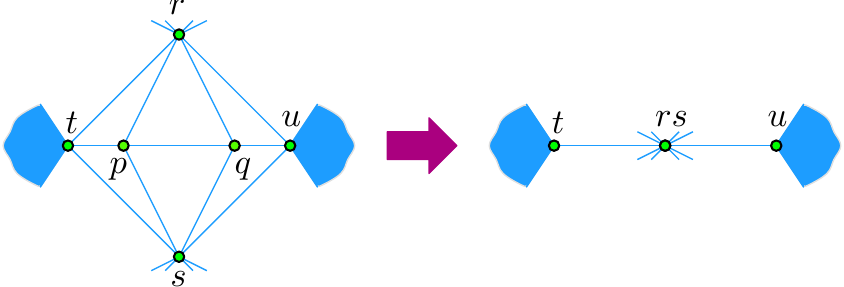}
\caption{Contracting a collapsible 4-cycle.}
\label{fig:contract}
\end{figure}

Observe that, if we contract every collapsible 4-cycle by deleting $p$ and $q$ and merging $r$ and $s$ into each other (see Fig.~\ref{fig:contract}), every high degree vertex will still have degree 18 or more (because at most half of the adjacent vertices of a high degree vertex can be the degree four vertices of collapsible 4-cycles) and there will still be no degree-3 or degree-5 vertices even though some of the high degree vertices might now have odd degree.
By Euler's formula, every triangulation without degree-3 or degree-5 vertices has at least $6h+6$ degree-4 vertices, where $h$ is the number of vertices of degree 18 or more. But when $h > 2$ then (in the graph formed by contracting every collapsible 4-cycle) there are at most $3h-6$ edges connecting pairs of degree-18 vertices and at most $6h-12$ triangles involving those edges, and at most three degree-four vertices in the root triangle, so there are at least three degree-four vertices that are not part of the root triangle and not part of any triangle with the other two vertices having degree eighteen or more. These three vertices are good in the original uncontracted graph, so it is also true that in the original uncontracted graph there are at least three good vertices. By the same argument, at every step of our decomposition algorithm, there always exists at least three good vertices in every connected component of the partially decomposed graph.

At each recursive call to our decomposition algorithm, we remove a single good vertex from our data structure that lists good vertices. As shown below, each such vertex helps us to find a feature in the triangulation, allowing to simplify it according Rules~I-IV, in linear time; we work though the list of good vertices until the graph becomes empty, at which point we return from the recursion.
Based on a type of a good vertex (described by 4 cases below) we decide what reduction operation to use, perform the reduction, and update our data structures.

\begin{description}
\item[Case 1:] The good vertex belongs to a triangle of degree-four vertices. We may determine that we are in this case, in constant time, by examining the modified degree numbers of the four adjacent vertices. In this case, since the graph has no separating triangles, it must be the octahedral graph $\Delta_6$---see Fig.~\ref{fig:rule3}. We may find the rest of the vertices, edges, and triangles of this graph by performing a depth-first search starting from the given good vertex. The decomposition step consists of removing this whole component from the partially decomposed graph, and then calling our decomposition algorithm recursively. When the recursive call returns, we add back the component to the graph, and mark the edges of its unique cycle cover.

\item[Case 2:] The good vertex belongs to a copy of $\Delta_{11}$ (Fig.~\ref{fig:rule3}).  We may determine that we are in this case, in constant time, by performing a depth-first search, starting from the good vertex, until either the entire graph has been explored and matches $\Delta_{11}$ or a feature not present in $\Delta_{11}$ is found. In this case too we remove the entire component from our partially decomposed graph, call the decomposition algorithm recursively, and when it returns restore the removed component and mark its cycle cover.

\item [Case 3:] The good vertex $p$ has a degree-4 neighbor $q$. This means the local neighborhood is a subgraph like the one in Fig.~\ref{fig:rule3a}~(left) with a chain of vertices $t-p-q-u$ with $r$ above and $s$ below connected to all of them. It is not possible that $t$ and $u$ both have high degree, because this would cause $p$ to be part of a collapsible cycle and therefore $p$ would not be good. Therefore, at least one of $t$ or $u$ is not high degree; we can determine which one does not have high degree by examining the modified vertex degrees. Without loss of generality suppose that the vertex with low degree is~$t$; the case that it is $u$ is symmetric. Then, for each neighbor $v$ of $t$ we test whether $v$ is adjacent to $u$. If we find a vertex $v$ that is indeed adjacent to both $t$ and $u$, that gives us 4-cycles $t-s-u-v$ and $t-r-u-v$, one of which must enclose three or more vertices on each of its sides (for otherwise we would be in Case~2), and we can check by a constant amount of exploration which of these two cycles has three or more vertices on each side. In this case we can cut our triangulation into two along this separating cycle according to Rule~I or Rule~II depending on the colors of its adjacent triangles. If we do not find such a vertex $v$ then we have a pair of adjacent degree-4 vertices satisfying the conditions of Rule~III, and we may replace them by a single edge. In either case we perform the appropriate decomposition step, call our decomposition algorithm recursively, undo the decomposition step, and update the cycle cover of the partially decomposed graph as shown in Figures~\ref{fig:gin-gout},~\ref{fig:gin-gout-1b},~and~\ref{fig:6-2flip}.

\item[Case 4:] The good degree-four vertex $p$ has its four neighbors (in cyclic order) $a, b, c, d$, where all four of $a, b, c$, and $d$ have degree greater than four. Because $p$ is good, some pair of opposite vertices (say $a$ and $c$) have low degree. By testing all pairs of neighbors of $a$ and $c$, we can determine in constant time whether there is a 4-cycle $a-p-c-x$ (in which case it must be separating and we can use Rule~I or Rule~II to cut along the 4-cycle) or whether there is a 5-cycle $a-p-c-x-y$. If there is no 4-cycle or 5-cycle, we can collapse $a-p-c$ into a single supervertex according to Rule~IV. So the remaining case is the one in which a 5-cycle exists.

\smallskip
In this remaining case, test the adjacency of $b$ and $d$ with $x$ and $y$. If none of these four pairs of vertices is adjacent, we can collapse $b-p-d$ into a single supervertex according to Rule~IV. Otherwise, there exists an edge between one of these four pairs; by symmetry we may assume without loss of generality that this edge is $b-x$, as the other three cases may be converted into this one by relabeling the vertices. But then $a-b-x-y$ is a 4-cycle containing the remaining neighbors of $b$. If there are only two vertices inside the 4-cycle, then they both have degree four, the low degree vertex $a$ is a neighbor of exactly one of them, and we have an instance of Case 3. Otherwise, $a-b-x-y$ is a separating 4-cycle and we can use Rule~I or Rule~II.

\smallskip
Regardless of which decomposition rule we find for this case, as in the previous cases, we perform the appropriate decomposition step, call the decomposition algorithm recursively, undo the decomposition step, and update the cycle cover.
\end{description}

\noindent We now discuss the running time analysis of the algorithm described above.
At each recursive call we perform a subset of the following operations:
\begin{itemize}
 \item Cutting a triangulation over a separating 4-cycle.

\smallskip
This step is implicit: we simply recurse on a single disconnected graph with small local changes, caused by splitting the previous graph along the cycle and adding a constant number of new faces to the interior and exterior of the split.

\item Performing local changes near degree-four vertices.

\item Updating the vertex degrees, adjacency testing data structures, and list of good vertices according to each local change to the graph.

\item Reversing a decomposition step.

\smallskip
Here we undo the decomposition changes to the triangulation and merge the cycle cover of the parts the way it is described in Appendix~VI. Undoing each decomposition step is no harder than making it, and the appropriate set of changes to the cycle cover can be found in constant time.
\end{itemize}

The slowest parts of each step involve a constant number of adjacency tests (to determine which decomposition case to use) and a constant number of changes to the adjacency testing data structure (whenever we make a change to the features of the triangulation). These steps take expected time $O(1)$ when the edge list is implemented as a randomized hash table and $O((\log\log n)^2/\log\log\log n)$ time when it is implemented as an integer searching data structure. All other parts of the algorithm take constant time per step.

Since we perform at most a linear number of decomposition steps for any 4-connected triangulation, the total time to perform the decomposition and find a cycle cover is $O(n)$ randomized expected time, or $O(n(\log\log n)^2/\log\log\log n)$ deterministic time with linear space.

\subsection*{Converting cycle covers to regular edge labelings}
The cycle cover is represented by adding a special marker to each edge of triangulation belonging to a cycle in the cover and storing at each vertex pointers to its cycle edges. To orient the edges we start at one of the source vertices and then propagate the orientation and edge coloring from a vertex that has already been colored to its neighbors. Each such vertex will have at least one labeled edge, and we can color the rest of the edges around it based on which side of the cycle covering that vertex they belong to as it is described in Appendix~III. Since every vertex knows its cycle cover edges, we can orient all edges of our triangulations in linear time.

\subsection*{Gluing the corner polyhedra together to form orthogonal polyhedra dual to each non-4-connected Eulerian triangulation}

In order to form a simple orthogonal polyhedron from a 3-connected graph $\G$ whose dual Eulerian triangulation $\Delta$ is not 4-connected, we split $\Delta$ into its 4-connected components (as described above), find a corner polyhedron representation separately for each component, and then glue these polyhedra together into a single more complex polyhedron.
When we described this gluing step earlier, in our graph-theoretic but non-algorithmic characterization of $xyz$ polyhedra, we described it in geometric terms as finding a "small enough" scale for each corner polyhedron so that when it is glued in it does not get in the way of anything else. But this idea of finding a small enough scale can be interpreted combinatorially, in terms of the sorted ordering of the coordinate values, without the need to treat these coordinate values numerically.

Essentially, when we glue a corner polyhedron $C$ into a larger polyhedron $P$, the coordinate values in $P$ may be represented using three sorted sequences, representing the ordering of the faces perpendicular to each of the three coordinate axes. Similarly, the coordinate values in $C$ may be represented as three sorted sequences. For each pair of sorted sequences representing the same coordinates, one from $P$ and one from $C$, we need to insert the sequence of coordinates from $C$ as a contiguous subsequence between some two coordinates of $P$. So we need data structures that can (a) represent a sorted sequence of coordinate values; (b) find the object in this sequence that represents the coordinate of some particular face plane (allowing for the possibility that more than one face may have the same coordinate); and (c) insert one sequence between some two values in another sequence.

This can be done by storing each such sequence as a doubly-linked list for (a), together with a pointer from each face to the position of its coordinate in the list for (b). If an initial set of integer coordinate values is given, we may convert them to positions in such a list in linear time by bucket sorting; alternatively, for our problem, it may be simpler to construct a doubly-linked list representing coordinate values at the time we created these values (that is, at the time we constructed the $st$-numbering of a graph $\Delta_{xy}$, we already represent this numbering implicitly as the set of positions in a doubly-linked list).
With these simple data structures each gluing operation takes constant time: we simply splice the sorted list representing the coordinates of $C$ into the appropriate position within the sorted list representing the coordinates of $P$. Thus, with this data structure, gluing all the polyhedra from our tree of 4-connected components of $\Delta$ into a single polyhedron takes linear time. Once everything is glued together, we can reconvert positions in the doubly-linked list back into integer coordinates in linear time: the coordinate values are just the positions in the doubly-linked list, and converting a list to its positions (the list ranking problem) may be solved straightforwardly by scanning the list.

The same data structure is also needed for the final step of the algorithm, gluing triconnected components together to form a simple orthogonal polyhedron, so the conversion from doubly-linked lists to numeric coordinate values should be delayed until all gluing steps are complete.

\end{document}